\documentclass[letterpaper,11pt]{article}

\newcommand{\blind}{0} 

\usepackage{listings}
\usepackage[noend]{algpseudocode}
\algnewcommand\algorithmicinput{\textbf{Input:}}
\algnewcommand\Input{\item[\algorithmicinput]}
\algnewcommand\algorithmicoutput{\textbf{Output:}}
\algnewcommand\Output{\item[\algorithmicoutput]}
\algblockdefx[PFOR]{PFor}{EndPFor}[1]{\textbf{Parallel For} #1 \algorithmicdo}{\algorithmicend} 
\algnewcommand{\IIf}[1]{\State\algorithmicif\ #1\ \algorithmicthen}
\makeatletter
\ifthenelse{\equal{\ALG@noend}{t}}%
  {\algtext*{EndPFor}}
  {}%
\makeatother
\makeatletter
\algrenewcommand\ALG@beginalgorithmic{\footnotesize}
\makeatother
\makeatletter
\newcounter{HALG@line}
\renewcommand{\theHALG@line}{\thealgorithm.\arabic{ALG@line}}
\makeatother

\usepackage[font=footnotesize]{caption} 
\usepackage{pifont}

\usepackage[autonum,colorhypersetup]{shortex}
\usepackage{natbib}
\usepackage{enumitem}
\setlist{noitemsep,topsep=3pt} 

\newcommand{\statespace}{\mathscr{X}}
\newcommand{\prodss}{\widetilde{\statespace}}
\newcommand{\prop}[1]{#1_{\mathrm{prop}}}
\newcommand{\xprop}{\prop{x}}
\newcommand{\iprop}{\prop{i}}
\newcommand{\epsprop}{\prop{\epsilon}}
\newcommand{\target}{\pi} 
\newcommand{\targetv}{\target^{(V)}} 
\newcommand{\pibeta}{\target^{(\beta)}} 
\newcommand{\lifted}{\widetilde{\pi}} 
\newcommand{\sigalg}{\Sigma} 
\newcommand{\prodsigalg}{\widetilde{\sigalg}} 
\newcommand{\mespace}{(\statespace,\sigalg)} 
\newcommand{\prodmespace}{(\prodss,\prodsigalg)} 
\newcommand{\tdxpath}{\boldsymbol{\tdx}} 
\newcommand{\pathspace}{\prodss^\infty} 
\newcommand{\pathmespace}{(\pathspace, \prodsigalg^\infty)} 
\newcommand{\indset}{\{0,\dots,N\}} 
\newcommand{\indsetone}{\{1,\dots,N\}} 
\newcommand{\epsset}{\{-1,+1\}} 
\newcommand{\kexp}{\mathbf{K}_{\text{e}}} 
\newcommand{\ktem}{\mathbf{K}_{\text{t}}} 
\newcommand{\knrst}{\mathbf{K}} 
\newcommand{\EEi}[1]{\EE^{(#1)}} 
\newcommand{\colldist}{\{\pi_i\}_{i=0}^N}
\newcommand{\atom}{\mcA} 

\newcommand{\rest}{\widehat{R}} 
\newcommand{\acc}[1]{\alpha_{#1}} 
\newcommand{\rej}[1]{\rho_{#1}} 
\newcommand{\dbeta}{\Delta\beta} 
\newcommand{\rtt}{T_{\mathrm{RT}}} 
\DeclareMathOperator{\TE}{TE} 
\newcommand{\TEest}{\widehat{\TE}} 
\newcommand{\TEinfty}{\TE^{\infty}} 
\newcommand{\Kmin}{K_{\text{min}}} 
\newcommand{\mknrst}{\mathbf{k}} 
\DeclareMathOperator{\med}{med} 
\DeclareMathOperator{\LSE}{LSE} 
\newcommand{\Vs}{\mathbf{V}} 
\DeclareMathOperator{\Id}{Id} 

\usepackage{amsmath,amssymb,amsthm,bm,bbm,amsfonts,mathtools,thmtools} 
\usepackage[capitalize,sort,compress]{cleveref} 

\newcommand{\plotwidth}{\textwidth}
\newcommand{\shorterplotwidth}{0.73\textwidth}

\author{Miguel Biron-Lattes\\
    \and
    Trevor Campbell \\
    \and 
    Alexandre Bouchard-C{\^o}t{\'e} \\
    Department of Statistics, University of British Columbia}

\title{Automatic regenerative simulation via non-reversible simulated tempering}

\begin{document}

\maketitle

\begin{abstract}
Simulated Tempering (ST) is an MCMC algorithm for complex target distributions 
that operates on a path between the target and a more amenable reference
distribution. Crucially, if the reference enables \iid sampling, 
ST is \emph{regenerative} and can be parallelized across independent tours. 
However, the difficulty of tuning ST has hindered
its widespread adoption. 
In this work, we develop a simple \emph{nonreversible} ST (NRST) algorithm,
a general theoretical analysis of ST, and an automated tuning procedure
for ST. A core contribution that arises from the analysis is a
novel performance metric---\emph{Tour Effectiveness (TE)}---that 
controls the asymptotic variance of estimates from ST for bounded test functions. 
We use the TE to show that NRST dominates its reversible counterpart.
We then develop an automated tuning procedure for NRST algorithms that targets the TE 
while minimizing computational cost. This procedure enables
straightforward integration of NRST into existing probabilistic programming languages.  
We provide extensive experimental evidence that our tuning scheme improves the performance
and robustness of NRST algorithms on a diverse set of probabilistic models.
\end{abstract}

\section{Introduction}\label{sec:intro}

Regenerative Markov chain Monte Carlo (MCMC) is an approach to MCMC that 
splits the Markov chain into independent tours at 
\emph{regeneration times}---steps at which the state of the chain is 
independent of the previous \citep[see][\S 6.4]{ripley1987stochastic}. 
There are two important advantages to regenerative simulation: 1) computation can be 
trivially parallelized by running tours concurrently,
and 2) it enables straightforward estimation of Monte Carlo error based
on the fact that the tours are \iid; indeed, regenerative 
MCMC is one of few ways to obtain an ``honest'' estimate \citep{jones_honest_2001}.
Despite the benefits of regenerative MCMC, several 
practical difficulties have prevented its widespread adoption.
First, existing methods typically require model-specific mathematical derivations to identify
regeneration times. For example, many methods rely on 
minorization conditions on the transition kernel of the Markov chain 
\citep{mykland1995regeneration,hobert_applicability_2002}, which 
must be derived on a case-by-case and are often difficult to obtain \citep{jones_honest_2001}.
Some other approaches avoid model-specific derivations, both in the discrete time setting 
\citep{brockwell2005identification,minh_regenerative_2012,%
nguyen_regenerative_2017,douc2022kickkac}, and more recently using Markov 
processes 
\citep{wang2021regeneration,mckimm2022sampling}, but still are not at the stage where they can be 
used in an automatic fashion.

The starting point of this paper is the observation that Simulated Tempering 
(ST) \citep{geyer1995annealing} can be used to automate regenerative MCMC in a large
class of probabilistic programs such as BUGS \citep{Lunn2000} and Stan \citep{stan2023stan} (details in \cref{sec:PPLs}). Given a path between a reference distribution
and a target distribution of interest, ST builds an augmented Markov chain
jointly on the original variable and the position along the path.
As long as the reference enables \iid simulation, the steps at 
which the sampler visits this distribution are regeneration times by construction. 
Past studies on ST---including recent \emph{non-reversible} variants 
\citep{sakai2016irreversible,faizi2020simulated}---have 
not leveraged this property. Although ST enables automated regenerative MCMC, 
it is not without its own drawbacks; in particular, it is difficult to tune robustly \citep[see e.g.][]{mitsutake2000replica,park2007choosing}.

In this work, we introduce several key contributions to fully realize the potential 
of ST as an automated regenerative MCMC method:
\begin{enumerate}
\item A simple non-reversible ST algorithm based on the \emph{lifting approach} \citep{chen1999lifting}.
\item A diagnostic---\emph{Tour Effectiveness (TE)}---that 
characterizes the performance of ST as a number between 0 and 1. 
In contrast to the typical MCMC notion of Effective Sample 
Size (ESS), the tour effectiveness can be computed without making reference 
to a specific test function. From the user's point of view, TE has a 
``look and feel'' akin to relative ESS in Sequential Monte Carlo (SMC).
\item A novel idealized model of ST algorithms that 
elucidates the relationship between TE and 
tuning parameters of the algorithm. The model also leads to an extension 
of the \emph{communication barrier}---originally
developed for Parallel Tempering \citep{syed2022nrpt}---to ST algorithms.
We show that the magnitude of the barrier is inversely related to TE.
\item A novel automated ST adaptation algorithm aimed at maximizing TE while
minimizing cost, based on the aforementioned relationship between TE and tuning parameters.
This tuning scheme has no free hyper-parameters and 
handles a wide range of problems.
\end{enumerate} 
We emphasize that these contributions exploit the  
regeneration structure of ST, which allows us to define a generic performance metric, 
the tour effectiveness, and then to optimize it and to formulate tuning 
recommendations.
After these developments, we provide a variety of experiments demonstrating that our tuning
procedure delivers substantial gains in efficiency for multiple reversible and non-reversible ST algorithms.
Finally, we discuss the implementation of NRST in Probabilistic Programming Languages (PPLs),
and provide guidance on deploying NRST to distributed systems.

\section{Background}\label{sec:background}

\subsection{Regenerative Markov chain Monte Carlo}

Markov chain Monte Carlo (MCMC) \citep[e.g.][]{robert2004monte} is a standard
tool for approximating integrals under an intractable probability distribution
$\target$ of interest. Under certain conditions, sequentially simulating a path 
$\{x_n\}_{n\in\posInts}$ of a Markov chain with invariant distribution $\pi$ 
allows us to consistently estimate the expectation $\target(h)$ of a real-valued function 
$h$ via ergodic averages $m_S := \frac{1}{S}\sum_{n=0}^{S-1} h(x_n)$.
In practice, however, the simulation must stop at a finite $S$, and assessing the
\emph{Monte Carlo error} of $m_S$ becomes important \citep{flegal2008markov}. Under 
further conditions \citep[e.g.][Thm.\ 1]{hobert_applicability_2002}, this error is captured by a Central Limit Theorem (CLT)
$\sqrt{S}(m_S - \target(h)) \xrightarrow{\text{d}} \distNorm(0, \gamma_s^2)$,
where $\gamma_s^2$ is known as the \emph{asymptotic variance} of $m_S$. If
a consistent estimator of $\gamma_s$ is available, it is possible 
to construct asymptotic confidence intervals for the Monte Carlo error. 
This classical approach to MCMC has two important drawbacks.
First, even when the conditions for the CLT hold, designing a 
consistent estimator for $\gamma_s$ is challenging \citep{jones_honest_2001}. 
Second, since the Markov chain proceeds sequentially in time, 
it is difficult to leverage the parallelism 
offered by current systems 
\citep{wilkinson2005parallel}.

\emph{Regenerative MCMC} \citep{mykland1995regeneration} is 
an alternative approach that naturally leverages massively parallel computation 
and yields straightforward estimates of Monte Carlo
error. Roughly, a Markov 
chain is regenerative if there exist stopping times
$0=r_0<r_1<r_2<\dots$ with $r_k\uparrow \infty$, such that for every $k\in\nats$,
the process after $r_k$ is
independent of its past. These \emph{regeneration times} can be used to split the 
sample path of the process into disjoint \emph{tours} 
$x^{(k)} := \{x_n\}_{n=r_{k-1}}^{r_k-1}$, for all $k\in\nats$. 
As the tours are \iid, they can trivially be simulated concurrently \citep{fishman1983accelerated}.
And at time $r_k$,
\[
m_{r_k} = \frac{1}{r_k}\sum_{n=0}^{r_k-1} h(x_n) = \frac{\sum_{j=1}^k s_j}{\sum_{j=1}^k \tau_j},
\]
where $\tau_k:=r_k-r_{k-1}$ is the length of the $k$-th tour, and
$s_k := \sum_{n=r_{k-1}}^{r_k-1} h(x_n)$. Since $r_k\uparrow \infty$, it holds
that $m_{r_k}\xrightarrow{\text{a.s.}} \pi(h)$ as $k\to \infty$. Under further
conditions, the independence of the tours can be used to leverage the CLT for
\iid random variables to show that
\[
 \gamma_r^2 := \frac{\EE[(s_1 - \tau_1\pi(h))^2]}{\EE[\tau_1]^2}, \qquad \sqrt{k}(m_{r_k} - \pi(h))  \xrightarrow{\text{d}} \distNorm(0, \gamma_r^2), \qquad k\to\infty.
\]
In contrast to the standard serial implementation of MCMC, $\gamma_r^2$ admits 
a simple and consistent estimator.

The main drawback of regenerative MCMC is the fact that finding regeneration
times for Markov chains on arbitrary spaces is not straightforward.
\citet{mykland1995regeneration} provided the first attempt, exploiting
the splitting technique \citep{nummelin1978splitting} to introduce regeneration 
when a small set is known to exist. But finding small sets is a challenging task in general. 
Since then, several  alternative procedures to induce regeneration in general Markov chains 
have been proposed \citep{sahu2003self,brockwell2005identification,minh_regenerative_2012,
nguyen_regenerative_2017,douc2022kickkac}. In the present work, we focus on a design of
Simulated Tempering (ST) \citet{geyer1995annealing} that admits regenerative structure.

\subsection{Simulated tempering}\label{sec:simulated_tempering}

Let $\statespace$ be a set, and let $\pi_0$ be a \emph{reference distribution} on $\statespace$ for which \iid sampling is
tractable. Assume that $\pi_0$ has a density with respect to a base measure $\dee x$ on 
$\statespace$; we overload notation so that $\pi_0$ also
refers to this density.
Let $\target$ be a \emph{target distribution} on $\statespace$ with density
\[\label{eq:def_target_dist}
\target(x) := \frac{\pi_0(x)e^{- V(x)}}{\mcZ}, \qquad \mcZ := \int \pi_0(x) e^{-V(x)}\dee x < \infty.
\]

Simulated tempering \citep{lyubartsev1992new,marinari1992simulated,geyer1995annealing}
aims to produce Monte Carlo samples from $\target$ by leveraging the tractability of 
$\pi_0$. It does so by connecting $\pi_0$ to $\pi$ with a path of distributions, and 
then moving draws from $\pi_0$ along the path.
A common choice is the path of \emph{tempered distributions} defined by
\[\label{eq:def_temp_dist}
\forall \beta\in[0,1]:\ \pi^{(\beta)}(x) :=  \frac{\pi_0(x)e^{-\beta V(x)}}{\mcZ(\beta)}, \quad \mcZ(\beta) := \int \pi_0(x)e^{-\beta V(x)}\dee x,
\]
where $\mcZ(\beta)$ is the normalization constant of $\pi^{(\beta)}$. These constants are finite
(\cref{lemma:finite_norm_const}) but 
the values $\mcZ(\beta)$ are usually unknown. 
For example, consider the case of a Bayesian model with prior
density $\varpi(x)$ and likelihood $L(y|x)$ for fixed observation $y$.
If $\varpi$ is proper,
we can set $\pi_0=\varpi$ and $V(x)=-\log(L(y|x))$ in 
\cref{eq:def_temp_dist} to obtain a valid path from $\varpi$ to the 
posterior density $\target(x) \propto \varpi(x)L(y|x)$. 
Note that much of the ST literature 
suggests choosing $\pi_0\equiv1$ and 
$V(x)=-\log(\varpi(x) L(y|x))$ \citep[e.g.][]{woodard2009conditions}, but this
is a valid path only in the restrictive setting 
where the uniform measure on $\statespace$ is finite, thus excluding many 
interesting statistical problems. Another benefit of the prior--posterior path
is that it can be built automatically in many probabilistic programming
languages with no additional user input (see \cref{sec:PPLs}). 

We now discretize the path to obtain a sequence of distributions $\{\pi_i\}_{i=0}^N$ to be used in computation. 
For $N\in\nats:=\{1,2,\dots\}$ and $0=\beta_0<\beta_1<\dots<\beta_N=1$, 
define $\mcP:=\{\beta_i\}_{i=0}^N$ and the sequence $\colldist$ by
$\pi_i := \pi^{(\beta_i)}$ for  $i\in\indset.$
ST builds a Markov chain $\{x_n,i_n\}_{n\in\posInts}$ on the 
product space $\statespace\times\indset$ with invariant density
\[\label{eq:def_ST_lifted}
\pi_\text{st}(x, i) := p_i\pi_i(x),
\]
where $\{p_i\}_{i=0}^N$ is the marginal distribution of the $i$ component,
and the conditional distribution of $x$ given $i$ is precisely $\pi_i$.
The state of the Markov chain defined by ST is updated by
alternating \emph{exploration} and \emph{tempering} steps. 
During exploration, $i$ is fixed, so only the $x$ component is updated. 
If $i=0$, a new state $x'\sim\pi_0$ is drawn independently from the reference distribution.
Otherwise, $x$ is updated via $x'|x,i \sim K_i(x,\cdot),$
where $\{K_i\}_{i=1}^N$ is a collection of Markov kernels on $\statespace$
such that each $K_i$ is $\pi_i$-irreducible and $\pi_i$-invariant (e.g. Gibbs, 
HMC, etc). The tempering step,
on the other hand, fixes $x$ and proposes a move $i \to i'$ with a Metropolis--Hastings correction. 
If the proposal is symmetric, then the move is accepted with probability 
\[\label{eq:def_comm_step_acc_prob}
\acc{i,i'}(x) &= \min\left\{1,\frac{\pi_{i'}(x)}{\pi_i(x)} \frac{p_{i'}}{p_i} \right\} = e^{-\max\left\{0, (\beta_{i'}-\beta_i)V(x) - (c_{i'}-c_i)\right\}},
\]
where $\mcC:=\{c_i\}_{i=0}^N$ is a collection of real-valued tuning parameters---the
\emph{level affinities}---that are related to the $p_i$ via
\[\label{eq:pseudo_prior_repar}
\forall i\in\indset: \quad p_i := \frac{\mcZ(\beta_i)e^{c_i}}{\sum_{j=0}^N\mcZ(\beta_j)e^{c_j}}.
\]
This relationship between $c_i$ and $p_i$ has the effect of canceling the
$\mcZ(\beta_i)$ in $\pi_i$ that would otherwise appear in the Metropolis--Hastings 
ratio. Thus, $\acc{i,i'}(x)$ can be evaluated \emph{without knowing the 
normalizing constants}.

As highlighted by \citet{geyer1995annealing}, the \iid draws from $\pi_0$
induce regeneration in ST: whenever the chain visits $\statespace\times\{0\}$, 
the next state is drawn independently of the previous one. Sets with this property 
are called \emph{atoms} \citep[][Def. 6.1.1]{douc2018markov}. Indeed, if 
$\{t_k\}_{k\in\nats}$ denote the times at which the chain is in the atom, then 
the stopping times $r_k:=t_k+1$ are regeneration times. Hence, ST allows 
regenerative simulation.

\section{Methodology}\label{sec:methodology}

\subsection{Non-reversible simulated tempering}\label{sec:NRST}

The effectiveness of simulated tempering can
be hindered if moving between the boundaries $i=0$ and $i=N$
is difficult. Indeed if, starting from $i=0$, the chain returns to
this level before visiting $i=N$, the tour produces no sample from the 
target, resulting in wasted effort.
Similarly if the chain never returns to $i=0$, then
regenerations never occur. 
The original versions of ST have been found to 
exhibit this problem due to the reversibility of the tempering step
\citep[][see \cref{sec:related_work} for a discussion of these methods]{sakai2016irreversible,faizi2020simulated}.

We therefore opt for a \emph{non-reversible} simulated tempering 
method (NRST), using lifting \citep{chen1999lifting} to 
induce momentum in the tempering dynamics and thus reduce 
undesirable random walk behavior. Concretely, we augment the Markov chain 
state $\tdx:=(x,i,\eps)$ to be in the product space 
\[\label{eq:def_prodss}
\prodss := \statespace \times \indset\times \epsset.
\]
The $\eps$ component represents the current direction of movement along $\indset$.
Indeed, the new tempering step makes a deterministic proposal to move the $i$ index
one unit along the current $\eps$ direction; i.e., $\iprop = i+\eps$. 
If $\iprop\notin\indset$, the proposal is immediately rejected. Otherwise, it
is accepted with probability $\alpha_{i,\iprop}(x)$, as defined in
\cref{eq:def_comm_step_acc_prob}. If the tempering step is rejected, the direction 
is reversed $\eps\gets-\eps$. After tempering, NRST
uses an exploration step identical to the one in simulated tempering.

In \cref{app:pseudo-code} we show pseudo-code for a function that performs a full NRST step 
(\cref{alg:NRST_step}). We also detail how this function is used to run a complete NRST 
tour (\cref{alg:NRST_tour}). Both the exploration and the tempering steps are invariant 
(\cref{prop:lifted_invariance}) with respect to the \emph{lifted} distribution $\lifted$ 
with density
\[\label{eq:def_lifted_dist}
\forall \tdx \in \prodss:\ \lifted(x,i,\eps) := \frac{1}{2}p_i\pi_i(x),
\]
where $\{p_i\}_{i=0}^N$ are as in \cref{eq:pseudo_prior_repar}. As in ST, the conditional distribution of $x$ given $i$ (marginalizing $\eps$) under $\lifted$ 
is $\pi_i$. Indeed, it is not difficult to show that
\[
\target(h) = \frac{\lifted(h(x)\ind\{i=N\})}{\lifted(\ind\{i=N\})}.
\]
for any $h:\statespace\to\reals$. Moreover, NRST preserves the regenerative 
structure of ST: consider 
$\atom := \statespace\times\{0\}\times\{-1\}$,
and note that starting from any $\tdx\in\atom$, the law of the next state $\tdx'$ 
is
\[\label{eq:def_nu_atom_measure}
 \nu(\tdx') := \pi_0(x')\ind\{i'=0, \eps'=+1\}.
\]
Hence, $\atom$ is an atom with regeneration measure $\nu$. 
We use $\Pr_\nu$ to denote the law of the Markov chain with 
initial distribution $\nu$, and use $\EE_\nu$ for corresponding expectations. 

The presence of the atom allows us to use the machinery of regenerative simulation
to build a consistent estimator of $\pi(h)$. To this end, define the sequence of
times at which the chain visits $\atom$ by setting $T_0:=-1$ and recursively defining
\[\label{eq:def_time_atom_visits}
\forall k\in\nats:\ T_k := \inf\{n>T_{k-1}: \tdx_n \in \atom \}.
\]
Furthermore, for any function $f:\prodss\to\reals$, define the \emph{tour sums of $f$}
\[\label{eq:def_tour_sums}
\forall k\in\nats:\ s_k(f) := \sum_{n=T_{k-1}+1}^{T_k} f(\tdx_n).
\]
Under $\Pr_\nu$, $\{s_k(f)\}_{k\in\nats}$ are \iid (\cref{lemma:tour_sums_iid}),
$s_j(\ind\{i=N\})$ is the number of visits to the top level during the 
$j$-th tour, and the length of the $k$-th tour is $\tau_k := s_k(1) = T_k - T_{k-1}.$
Finally, let $\rest_k(h)$ be an estimator of $\target(h)$ given by
\[\label{eq:def_ratio_est_implicit}
\rest_k(h) := \frac{\sum_{j=1}^k s_j(h(x)\ind\{i=N\})}{\sum_{j=1}^k s_j(\ind\{i=N\})}.
\]
$\rest_k(h)$ is consistent due to the Law of Large Numbers (LLN) for the lifted chain (\cref{thm:LLN}).
This is a known result \citep[Thm.\ 6.6.2][]{douc2018markov}, 
but we reproduce it here for convenience (proof in \cref{app:proofs}).
Let $s(f)$ and $\tau$ be independent copies of $s_k(f)$ and $\tau_k$.
\begin{theorem}\label{thm:LLN}
For any $f:\prodss\to\reals$ such that $\lifted(|f|) < \infty$,
\[\label{eq:LLN_toursums_identity}
\EE_\nu[s(f)] = \EE_\nu\left[\tau\right]\lifted(f) = \frac{2}{p_0}\lifted(f).
\]
Moreover, if $g:\prodss\to\reals$ satisfies $\lifted(|g|)<\infty$
and $\lifted(g)\neq 0$, then
\[\label{eq:LLN_complete_tours}
\lim_{k\to\infty} \frac{\sum_{j=1}^k s_j(f)}{\sum_{j=1}^k s_j(g)} = \frac{\EE_\nu[s(f)]}{\EE_\nu[s(g)]} = \frac{\lifted(f)}{\lifted(g)}, \quad \Pr_{\nu}\text{-a.s.}
\]
\end{theorem}

Using $f(\tdx):=h(x)\ind\{i=N\}$ and $g(\tdx):=\ind\{i=N\}$ in \cref{thm:LLN} shows
that $\rest(h)$ is a consistent estimator of $\target(h)$. 
Beyond consistency, a central limit theorem (\cref{thm:CLT}) holds for $\rest(h)$.
This result is equivalent to the one given in \citet{geyer1995annealing} for ST,
but we reproduce it here for convenience (proof in \cref{app:proofs}).

\begin{theorem}\label{thm:CLT}
Let $h:\statespace\to\reals$ be a function with $\target(|h|)<\infty$, and define
\[\label{eq:def_asymp_var}
\sigma(h)^2 := \frac{\EE_\nu\left[s\left(\ind\{i=N\}(h(x)-\target(h))\right)^2\right]}{\EE_\nu[s(\ind\{i=N\})]^2}.
\]
If $\sigma(h)^2<\infty$, then $\sqrt{k}\left( \rest_k(h) - \target(h) \right) \convd \distNorm\left(0,\sigma(h)^2\right)$ as $k\to\infty$.
\end{theorem}

The variance $\sigma(h)^2$ can be estimated directly from the output of the 
regenerative simulation; indeed, by \cref{thm:LLN}, 
\[\label{eq:def_asymp_var_est}
\widehat{\sigma(h)^2}_k := \frac{k\sum_{j=1}^k s_j\left(\ind\{i=N\}(h(x)-\rest_k(h))\right)^2}{\left(\sum_{j=1}^k s_j(\ind\{i=N\})\right)^2}
\]
is a consistent estimator of $\sigma(h)^2$.
Then, an asymptotically valid $\alpha$-confidence interval may be constructed for
$\alpha\in(0,1)$, such that
\[\label{eq:def_asymptotic_CI}
\Pr_\nu\left( \left| \rest_k(h) - \target(h)\right|  \leq  \frac{\widehat{\sigma(h)}_k}{\sqrt{k}}z_\alpha \right) \xrightarrow{k\to\infty} \alpha, \qquad z_\alpha := \Phi^{-1}((1+\alpha)/2),
\]
with $\Phi^{-1}$ the quantile function
of a standard normal distribution.

\subsection{Tour effectiveness}\label{sec:methodology_TE}

The asymptotic variance estimate $\widehat{\sigma(h)^2}_k$ depends on $h$ in general; but 
for bounded test functions $|h| \leq 1$, we can uniformly bound $|h(x) - \rest_k(h)|$
and obtain a diagnostic for simulated tempering performance that we 
call the  \emph{estimated tour effectiveness},
 \[\label{eq:def_TEest}
\TEest := \frac{\left(\sum_{j=1}^k s_j(\ind\{i=N\})\right)^2}{k\sum_{j=1}^k s_j(\ind\{i=N\})^2}.
\]
The $\TEest$ diagnostic is a number between 0 and 1, and applies to any ST algorithm.
A higher $\TEest$ indicates more efficient motion
of the Markov chain along the path from the reference $\pi_0$ to the target $\pi$. 
Note that unlike the typical effective sample size (ESS) for MCMC algorithms, $\TEest$ does not
require the choice of a test function $h$. $\TEest$ is more akin to the relative effective sample size (rESS)
from sequential Monte Carlo (SMC) \citep{liu1995blind}:  both diagnostics exist in $[0,1]$ and are defined in terms of 
ratios of the first and second moments of a nonnegative ``weight'' (a density ratio for rESS, and a number of target visits for $\TEest$).
We will provide a more in-depth analysis of tour effectiveness in \cref{sec:analysis}.

For any confidence level $\alpha$ and
tolerance $\delta>0$, we can estimate the minimum number of tours $K_{\text{min}}$ 
required to obtain an approximate $\alpha$-confidence interval of half-width 
less than $\delta$ for all functions $|h|\leq 1$ by setting
\[\label{eq:def_Kmin}
K_{\text{min}}(\alpha,\delta,\TEest) := \left\lceil \frac{4}{\widehat{\TE}}\left(\frac{z_\alpha}{\delta}\right)^2 \right\rceil .
\]
\cref{alg:NRST_par} (\cref{app:pseudo-code}) describes how this method is used to run a
parallel regenerative simulation of NRST. Of course, after obtaining the results, it is 
possible to produce tighter intervals by computing $\widehat{\sigma(h)^2}$ for each 
function of interest.

\section{Theoretical analysis}\label{sec:analysis}

\subsection{The index process}\label{sec:index_process}

\begin{figure}[t]
\centering
\includegraphics[width=\plotwidth]{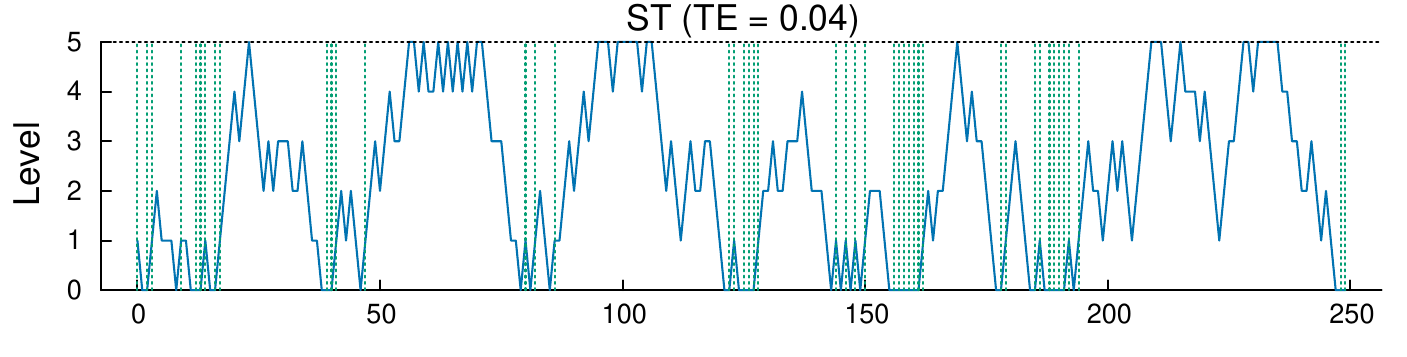}
\includegraphics[width=\plotwidth]{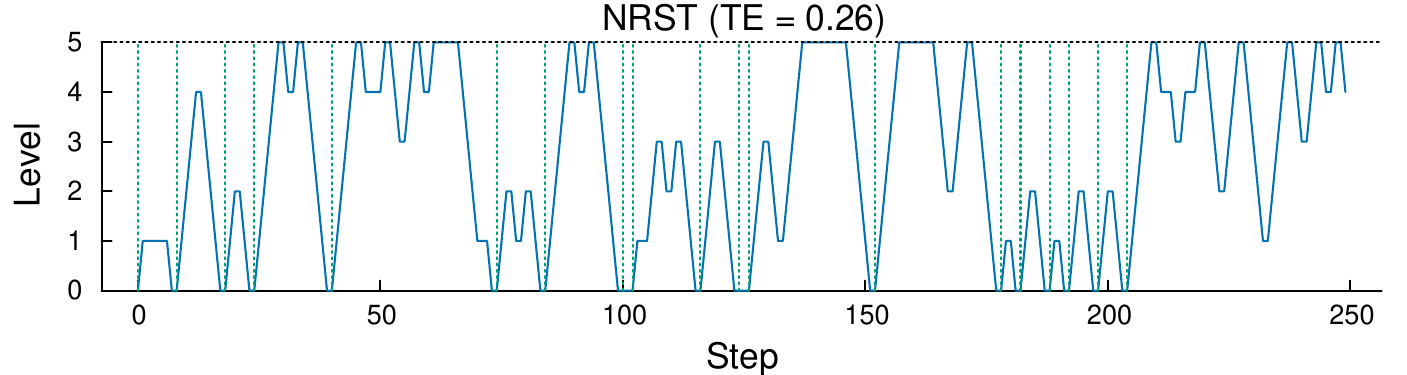}
\caption{First $250$ steps of the index processes of ST and NRST applied to a
toy problem, using the same grid ($N=6$), level affinities, and stream of
pseudo-random numbers on both algorithms. The start of each tour is marked by
vertical dotted lines.}
\label{fig:index_process}
\end{figure}

We now analyze tour effectiveness and use it to characterize the performance
of NRST and traditional ST. To begin, note that $\TEest$ depends
only on the tempering indices, and can be
computed without knowing the values of the state $x$.
Therefore, we focus the analysis on the \emph{index process} $\{(i_n,\eps_n)\}_{n\in\posInts}$. 
\cref{fig:index_process} shows a realization of the index process for both ST and 
NRST applied to a toy problem (see \cref{app:models_description}). 
Both algorithms use the same grid, level
affinities, and pseudo-random stream. Regenerations are
indicated by a vertical dotted line. \cref{fig:index_process} shows a 
stark difference between ST and NRST: while NRST produces 
tours that travel quickly between the reference and target, ST
produces meandering tours that inflate the variance of Monte Carlo estimators. This difference appears in the $\TEest$ values. 

Alas, studying the index process itself is hard,
because it is \emph{not} a Markov chain in general.
However, the tempering step depends only on $x$ via $V(x)$; so 
in cases where the exploration kernels are efficient enough that $V(x)$ 
appears to be \iid, the index process becomes Markovian. The following idealized 
assumption describes this condition precisely.

\begin{assumption}\label{assu:ELE}
For $i\in\{1,\dots,N\}$, let $K_i(x,\dee v')$ be the distribution of $V(x')$
conditional on the current state $x\in\statespace$, where $K_i$ is the $i$-th
exploration kernel and $x'$ is the state after exploration. There 
exists a probability distribution $q_i$ on $\reals$ such that 
$K_i(x, \dee v') = q_i(\dee v')$, for $\target$-almost all $x$.
\end{assumption}

\cref{assu:ELE} requires that  $V(x_n)$ ``mixes perfectly'' in one 
exploration step. Note that this is \emph{not} the same as the stronger 
requirement that $x_n$ itself does so during exploration 
(the latter would defy the point of using ST). 
\cref{assu:ELE} can be thought of as a \emph{model}, i.e., a simplifying 
framework that is not expected to hold exactly in real scenarios. 
However this model is \emph{useful} in the sense that we obtain
approximations allowing us to derive an adaptation scheme 
(\cref{alg:NRST_adapt}) that performs robustly across a variety of models---even 
when \cref{assu:ELE} does not exactly hold. 
For this reason, our experiments (\cref{sec:numerical_results}) focus on real 
target distributions in which \cref{assu:ELE} is not satisfied.
See \cref{sec:serial_correlation} for a discussion on how one can enforce 
\cref{assu:ELE} and the practical effects of doing so.
Related performance models have been used to analyze parallel tempering; see    
\citet[Sec.\ 3.3]{syed2022nrpt} for empirical examples where $V(x_n)$ mixes well 
but $x_n$ does not. 

We characterize the Markov index process using the acceptance and rejection 
probabilities under $\pi_i$, and their symmetrized versions, given by
\[\label{eq:def_average_acc_rej_probs}
\acc{i,i+\eps} &:= \EEi{i}[\alpha_{i,i+\eps}(x)], & \rej{i,i+\eps} &:=1-\alpha_{i,i+\eps}\\
\alpha_i &:= (\acc{i,i-1} + \acc{i-1,i})/2, & \rho_i &:= 1-\alpha_i.
\]
\cref{prop:marginal_chain} describes the index process using these probabilities.
\begin{proposition}\label{prop:marginal_chain}
Under \cref{assu:ELE}, the marginal distribution of the index process for both ST and NRST
is a homogeneous Markov chain on $\indset\times\epsset$, with invariant distribution $\lifted(i,\eps) := p_i/2.$
Moreover, the transition kernels for NRST and ST are:
\[
\mknrst_\mathrm{NRST}((i,\eps), (i',\eps')) &:= \acc{i,i+\eps}\ind\{i'=i+\eps, \eps'=\eps\} + \rej{i,i+\eps}\ind\{i'=i, \eps'=-\eps\}, \\
\mknrst_\mathrm{ST}((i,\eps), (i',\eps')) &:= (\acc{i,i+\eps'}\ind\{i'=i+\eps'\} + \rej{i,i+\eps'}\ind\{i'=i\})/2.
\]
\end{proposition}

\subsection{The Tour Effectiveness diagnostic}\label{sec:toureff}

The aim of this section is to rigorously analyze the tour effectiveness diagnostic 
introduced briefly in \cref{sec:methodology_TE}. 
Let $v_k = s_k(\ind\{i=N\})$ be the number of visits to level $N$ during the $k^\text{th}$ tour,
and $v$ be an independent copy of $v_k$ from $\Pr_\nu$.
Then we define the \emph{tour effectiveness} (TE) of a generic simulated tempering sampler as
\[\label{eq:def_TE_N}
\TE := \frac{\EE_\nu[v]^2}{\EE_\nu[v^2]}.
\]
Intuitively, $\TE$ describes how easy it is for the process to move between the
two endpoints of $\indset$. To see this, note first that $\TE\in[0,1]$, since
$\EE_\nu[v]^2\leq \EE_\nu[v^2]$. When $\TE=1$,
the number of visits per tour to the top level is deterministic and equal to
$2p_N/p_0$ (\cref{eq:LLN_toursums_identity} with $f\equiv1$). At the opposite 
extreme, $\TE=0$ when $v$ has infinite second moment. Note that, since $\{v_k\}_{k\in\nats}$ is an \iid collection, the 
statistic $\TEest$ introduced in \cref{sec:NRST} (\cref{eq:def_TEest}) is a 
consistent estimator of $\TE$.
\cref{thm:CLT_bdd_h} shows that $\TE$ controls the ST asymptotic variance 
$\sigma(h)^2$ from \cref{eq:def_asymp_var},
and hence is a relevant diagnostic for simulated tempering.

\begin{theorem}\label{thm:CLT_bdd_h}
$\sup_{|h|\leq 1} \sigma(h)^2 \leq 4/\TE$, and under \cref{assu:ELE}, $\TE > 0$.
\end{theorem}

\cref{thm:CLT_bdd_h} shows that the more \emph{effective} 
the tours are---as measured by $\TE$---the lower (a bound on) the asymptotic 
variance of Monte Carlo estimators for bounded functions will be.
Furthermore, \cref{thm:CLT_bdd_h} states that under \cref{assu:ELE},
this variance is always finite. This is not guaranteed in general, as
$\sigma^2(h)$ can be infinite even for bounded $h$.

To further understand the role of $\TE$ as a diagnostic tool, we require a
better grasp of its dependence on algorithm parameters. Whenever
$c_i=-\log(\mcZ(\beta_i))$, the distribution over levels is uniform (see
\cref{eq:pseudo_prior_repar}), and the average rejection/acceptance probabilities 
become symmetric (\cref{lemma:unif_pprior_sym_rejs}),
\[\label{eq:def_sym_rejs}
\forall i\in\{1,\dots,N\}:\ \rej{i-1,i}=\rej{i,i-1}=\rho_i \quad \text{and} \quad \acc{i-1,i}=\acc{i,i-1}=\alpha_i.
\]
In this setting, \cref{thm:TE_formula} provides an expression for $\TE$ in terms 
of the rejection probabilities (proof in \cref{app:proofs}).

\begin{theorem}\label{thm:TE_formula}
Under \cref{assu:ELE}, if $\{p_i\}_{i=0}^N$ is uniform, then the $\TE$ for NRST, 
$\TE_\mathrm{NRST}$, and the $\TE$ for reversible ST, $\TE_\mathrm{ST}$, satisfy
\[\label{eq:TE_formula}
\TE_{\mathrm{NRST}} = \frac{1}{1 + 2\sum_{i=1}^N \frac{\rho_i}{1-\rho_i}} > 
\frac{1}{4N-1 + 4\sum_{i=1}^N \frac{\rho_i}{1-\rho_i}} = \TE_\mathrm{ST}.
\]
\end{theorem}

Given uniform levels and \cref{assu:ELE}, \cref{thm:TE_formula} shows that NRST will produce more effective tours than ST, 
and that the difference increases as the size of the grid $N$ grows.
Additionally, \cref{thm:TE_formula} shows that TE depends on the grid $\mcP$ only through the 
quantity $\sum_{i=1}^N \frac{\rho_i}{1-\rho_i}$.

\subsection{Limit of infinite temperatures}\label{sec:limit_grid}

Motivated by the link between TE and
the rejection probabilities, we now analyze the limiting behavior 
of these quantities for NRST, in the hypothetical case where the grid becomes increasingly
fine; i.e., when we let $\|\mcP\| := \max\{\beta_i-\beta_{i-1}: i\in\{1,\dots,N\}\}$ 
go to zero. We will show that the optimal grid has a finite mesh size 
(\cref{sec:tune_N}), but the asymptotic results developed here provide 
the basis for approximations used as a stepping stone towards these (approximate) 
optimality results. 
 To proceed, let us extend 
\cref{eq:def_comm_step_acc_prob} to arbitrary pairs $\beta,\beta'\in[0,1]$,
\[\label{eq:def_extended_acc_prob}
\acc{\beta,\beta'}(x) := \exp\left(-\max\left\{0, (\beta'-\beta)V(x) - [c(\beta')-c(\beta)]\right\}\right),
\]
where $c:[0,1]\to\reals$ is the \emph{affinity function}, which extends the level 
affinities to a continuum by setting
$c_i = c(\beta_i)$. \cref{eq:def_extended_acc_prob} induces an extension of the 
 average acceptance and rejection probabilities (c.f. \cref{eq:def_average_acc_rej_probs}), i.e., 
$\acc{\beta,\beta'} := \EEi{\beta}\left[\acc{\beta,\beta'}(x)\right]$ and $\rej{\beta,\beta'} := 1 - \acc{\beta,\beta'}.$
\cref{thm:rejection_expansions} gives a second order 
approximation of the rejection probabilities between two close grid points.
Let us write $V=V(x)$ and $V\sim \pibeta$ whenever 
$x\sim \pibeta$.

\begin{assumption}\label{assu:c_smooth}
$c(\cdot)$ is thrice continuously differentiable on $(0,1)$.
\end{assumption}
\begin{assumption}\label{assu:nullset}
For all $\beta,\beta'\in(0,1)$, $\beta\neq\beta'$,
$\pi_0\left\{V = \frac{c(\beta') - c(\beta)}{\beta'-\beta}\right\}=0$.
\end{assumption}

\begin{theorem}\label{thm:rejection_expansions}
Under \cref{assu:c_smooth,assu:nullset}, for $0<\beta<\beta'<1$, $\dbeta:=\beta'-\beta$ and $\bar\beta:=(\beta+\beta')/2$,
\[
\rej{\beta,\beta'} = &\EEi{\bar\beta}[\max\{0,V-c'(\bar\beta)\}]\dbeta  - \frac{1}{2} \EEi{\bar\beta}[V-c'(\bar\beta)] \EEi{\bar\beta}[\max\{0,V-c'(\bar\beta)\}]\dbeta^2 + o(\dbeta^2) \\
\rej{\beta',\beta} = -&\EEi{\bar\beta}[\min\{0,V-c'(\bar\beta)\}]\dbeta -\frac{1}{2} \EEi{\bar\beta}[V-c'(\bar\beta)]\EEi{\bar\beta}[\min\{0,V-c'(\bar\beta)\} ] \dbeta^2 + o(\dbeta^2).
\]
\end{theorem}

In the special case where the affinity function satisfies $c'(\beta)=\EEi{\beta}[V]$, the second
order error vanishes. This mimics the result obtained by
\citet{predescu2004incomplete} for parallel tempering. To better understand the
implications of different choices of affinity functions, we will focus on 
first-order effects. To this end, define the 
\emph{directional rates of rejection} as
\[\label{eq:def_dir_inst_rate_rej}
\forall \beta\in(0,1): \rho_+'(\beta) := \lim_{\dbeta\downarrow0} \frac{\rej{\beta,\beta+\dbeta}}{\dbeta}, \qquad \rho_-'(\beta) := \lim_{\dbeta\downarrow0} \frac{\rej{\beta,\beta-\dbeta}}{\dbeta}.
\]
Furthermore, let the \emph{rate of rejection} be the average of the directional 
rates; i.e., $\rho'(\beta) := (\rho_+'(\beta) + \rho_-'(\beta))/2$.
The following result is a direct consequence of \cref{thm:rejection_expansions}.

\begin{corollary}\label{cor:rhoprime_eqs}
Under \cref{assu:c_smooth,assu:nullset}, for all $\beta\in(0,1)$,
\[
\rho_+'(\beta) := \phantom{-}\EEi{\beta}[\max\{0,V-c'(\beta)\}], \quad
\rho_-'(\beta) := -\EEi{\beta}[\min\{0,V-c'(\beta)\}].
\]
It follows that the rejection rate is
\[\label{eq:rhoprime_sum_abs_vals}
\rho'(\beta) = \frac{1}{2} \EEi{\beta}[|V-c'(\beta)|],
\]
while the difference in directional rates is
\[\label{eq:rhoprime_diff_abs_vals}
\rho_+'(\beta) - \rho_-'(\beta) = \EEi{\beta}[V]-c'(\beta).
\] 
\end{corollary}

\cref{cor:rhoprime_eqs} lets us analyze the behavior of the rejection rates given the affinity function $c$. We discuss two 
important cases. Since constants shifts in $c$ are irrelevant, we 
assume without loss of generality that $c(0)=0$.
\begin{assumption}[Mean energy affinity]\label{assu:mean_energy_affinity}
The affinity function satisfies
\[\label{eq:def_mean_energy_affinity}
\forall \beta\in[0,1]:\ c(\beta) = \int_{0}^{\beta} \EEi{\beta'}\left[ V \right] \dee \beta'.
\]
\end{assumption}
Differentiating the affinity function in \cref{eq:def_mean_energy_affinity}
shows that
$c'(\beta) = \EEi{\beta}\left[ V \right].$
Thus the directional rates are identical by \cref{eq:rhoprime_diff_abs_vals}. Moreover, using the thermodynamic 
identity \citep[see][and \cref{cor:mean_V_properties}]{gelman1998simulating}
$\log\left(\mcZ(\beta)\right) = -\int_{0}^{\beta} \EEi{\beta'}\left[ V \right]\dee \beta'$,
we see that
\[\label{eq:mean_energy_aff_as_free_energy}
c(\beta) = -\log\left(\mcZ(\beta)\right).
\]
Thus this choice of level affinities 
results in a uniform distribution over levels with symmetric rejections;
the equality of the directional rates is the infinitesimal counterpart
of \cref{eq:def_sym_rejs}. Furthermore, since $\mcZ\in\mcC^\infty((0,1))$ (see
\cref{lemma:norm_const_derivs}), the representation in 
\cref{eq:mean_energy_aff_as_free_energy} shows that this choice of affinity
function is $\mcC^\infty((0,1))$ too, and thus satisfies \cref{assu:c_smooth}.
Substituting $c$ into 
\cref{eq:rhoprime_sum_abs_vals}, the rate of rejections is
\[
\rho'(\beta) = \frac{1}{2}\EEi{\beta}\left[\left|V -  \EEi{\beta}[V] \right|\right].
\]
Compare this quantity to the corresponding rate for 
non-reversible parallel tempering, the \emph{local communication barrier}, defined as
\citep{syed2022nrpt}
\[
\lambda(\beta) := \frac{1}{2}\EEi{\beta}\left[\left|V - V' \right|\right],
\]
where $V'$ is an independent copy of $V$ defined on the same probability space. 
\cref{lemma:mad_bounds} shows that rejections should be lower for NRST (but 
not less than half).

\begin{assumption}[Minimum rejection affinity]\label{assu:min_rej_rate_affinity}
Let $\med^{(\beta)}(V)$ be a median of $V$ under $\pibeta$. Then, the affinity
function satisfies
\[\label{eq:def_min_rej_rate_affinity}
\forall \beta\in[0,1]:\ c(\beta) = \int_{0}^{\beta} \med^{(\beta')}(V)\dee \beta'.
\]
\end{assumption}

Differentiating the affinity function in \cref{eq:def_min_rej_rate_affinity}
shows that
\[\label{eq:cprime_median_energy}
\der{c}{\beta}(\beta) = \med^{(\beta)}(V).
\]
Since $\med^{(\beta)}(V)$ minimizes the absolute deviation of $V$, this choice of
$c$ function minimizes the rejection rate. Its value becomes
\[
\rho'(\beta) &= \frac{1}{2}\EEi{\beta}\left[\left|V - \med^{(\beta)}(V) \right|\right] \leq \frac{1}{2}\EEi{\beta}\left[\left|V -  \EEi{\beta}[V] \right|\right],
\]
where the inequality is due to the optimality of the median.
Using \cref{eq:cprime_median_energy} in \cref{eq:rhoprime_diff_abs_vals} 
shows that the directional rates of rejection do not agree; their difference
equals the difference between the mean and the median of $V$ under $\pibeta$.
Finally, integrating \cref{eq:cprime_median_energy} and using the thermodynamic identity, the pseudo-prior corresponding to this choice of $c$ is
\[
p_i \propto \exp(\log(\mcZ(\beta_i)) + c(\beta_i)) \propto \exp\left(-\int_0^{\beta_i}(\EEi{\beta}[V]-\med^{(\beta)}(V))\dee\beta\right).
\]
Again, the marginal probability of level $i$ depends on the 
cumulative effect of the skewness of the distribution of $V$ along the path.

The rates of rejection describe the \emph{local} behavior of the sampler, 
measuring how difficult it is to move between two infinitesimally close levels.
However, these rates are also helpful in studying the \emph{global} behavior of
the Markov chain, revealing how hard it is to transport \iid samples from the
reference towards the target measure through the path $\beta\mapsto \pibeta$.
To this end, let us define
the \emph{tempering barrier function} as the cumulative rate of rejection
$\Lambda(\beta) := \int_0^\beta \rho'(b)\dee b.$
\cref{prop:rejs_Riemann_sums} relates $\Lambda(\beta)$
and the symmetrized rejection rates.

\begin{assumption}\label{assu:integrability_at_extremes}
For $\beta\in\{0,1\}$, $\EEi{\beta}[V^2]<\infty$.
\end{assumption}
\begin{assumption}\label{assu:bounded_c}
There exist constants $\{\bar{c}_k\}_{k=0}^3$ such that
\[
\forall k\in\{0,1,2\}, \forall\beta\in[0,1]:\ \left|\frac{\mathrm{d}^kc}{\mathrm{d}\beta^k}(\beta)\right| \leq \bar{c}_k. 
\]
\end{assumption}

\begin{proposition}\label{prop:rejs_Riemann_sums}
Given \cref{assu:c_smooth,assu:nullset,assu:integrability_at_extremes,%
assu:bounded_c}, define
\[
\forall \beta\in[0,1]:\ \ell(\beta) := \max\{i\in\indset: \beta_i\leq\beta\}.
\]
Then, for all $\beta\in[0,1]$,
\[\label{eq:rejs_Riemann_sums}
\lim_{\|\mcP\|\to 0}\sum_{i=1}^{\ell(\beta)} \frac{\rho_i}{1-\rho_i} = \lim_{\|\mcP\|\to 0}\sum_{i=1}^{\ell(\beta)} \rho_i = \Lambda(\beta).
\]
\end{proposition}

\cref{cor:limit_TE_unif} combines \cref{prop:rejs_Riemann_sums} and \cref{thm:TE_formula}
to understand the asymptotic behavior of the tour effectiveness in terms of the \emph{total tempering barrier}
$\Lambda:=\Lambda(1)$. 
\begin{corollary}\label{cor:limit_TE_unif}
Under \cref{assu:ELE,assu:nullset,assu:mean_energy_affinity,assu:integrability_at_extremes,%
assu:bounded_c}
\[\label{eq:limit_TE_unif}
\TEinfty := \lim_{\|\mcP\|\to 0} \TE_\mathrm{NRST} = \frac{1}{1+2\Lambda}, \qquad \lim_{\|\mcP\|\to 0} \TE_\mathrm{ST} = 0.
\]
\end{corollary}

\cref{cor:limit_TE_unif} states that as the tempering grid 
becomes finer, the NRST tour effectiveness converges to a positive value 
depending solely on the total tempering barrier. In contrast, the standard ST TE
collapses to zero in the same limit.
Furthermore, \cref{cor:limit_TE_unif} establishes that, after precisely tuning 
level affinities and introducing sufficiently many grid points,
the tour effectiveness of NRST is controlled by the difficulty of the problem, $\Lambda$,
and the precise grid $\mcP$ becomes irrelevant. This is similar to the
fundamental limit to the asymptotic round trip rate imposed by the global
communication barrier in NRPT \citep{syed2022nrpt}.

\section{Adaptation algorithms}\label{sec:adaptation}

Drawing on the insights from \cref{sec:analysis}, we present an adaptation
strategy that works well in a wide range of problems and requires no
user input. We will first consider a fixed-size grid,
and focus on the adaptation of grid points and level affinities. 
The procedure requires only a ``dataset'' 
$\Vs := \left\{\{V_n^{(i)}\}_{n=1}^{S_i}: i\in\indset\right\},$
containing ergodic samples of $V\sim\pibeta$ for each $\beta\in\mcP$.
One way to obtain $\Vs$ is to bootstrap NRST with an arbitrary grid and level affinities. 
However, when the affinities are not well-tuned, we find that ST fails catastrophically 
and produces no draws of $V$ for levels even just 
a few steps up from the reference. 
Another option is to simply run the exploration 
kernels $\{K_i\}_{i=1}^N$ independently in parallel at each grid 
point; however, such independent exploration does
not produce high-quality $V$ samples when the path of distributions is even
slightly complex. 

\begin{algorithm}[t]
\caption{Adaptation of grid and level affinities}
\begin{algorithmic}
\Function{AdaptNRST}{$N, n_\text{max-rounds}$}
\State $\mcP \gets \left\{\frac{i}{N}: i\in\indset\right\}$
\State $n_\text{scan}\gets 1$
\vspace{.3em}
\State \texttt{\# Loop adaptation with doubling effort until convergence}
\While{$n_\text{scan} < 2^{n_\text{max-rounds}}$}
	\State $n_\text{scan}\gets 2n_\text{scan}$
	\State $\Vs\gets$ \textproc{runNRPT}($\mcP, n_\text{scan}$)
	\State $\mcC\gets$ \textproc{adjustAffinities}($\Vs$) \Comment{\cref{sec:adjust_affinities}}
	\State $\widehat{\Lambda} \gets$ \textproc{buildBarrierFunction}($\mcP, \mcC, \Vs$) \Comment{\cref{eq:def_estimate_temp_barrier} in \cref{sec:adapt_grid}}
	\State $\mcP\gets$ \textproc{optimizeGrid}($\widehat{\Lambda}$) \Comment{\cref{sec:adapt_grid}}
	\IIf{\textproc{hasConverged}()} \textbf{break}  \Comment{\cref{sec:detect_convergence}}
\EndWhile
\vspace{.1em}
\State \texttt{\# Tune the affinities for the final grid}
\State $\Vs\gets$ \textproc{runNRPT}($\mcP, n_\text{scan}$)
\State $\mcC\gets$ \textproc{adjustAffinities}($\Vs$)
\State $\widehat{\Lambda} \gets$ \textproc{buildBarrierFunction}($\mcP, \mcC, \Vs$)
\State \Return $\mcP, \mcC, \widehat\Lambda$
\EndFunction
\end{algorithmic}
\label{alg:NRST_adapt}
\end{algorithm}

We recommend instead running non-reversible parallel tempering (NRPT)
\citep{syed2022nrpt} to generate $\Vs$. Notably, this
requires no additional evaluations of $V(x)$ compared with running the kernels
individually, but often improves the quality of the dataset of potentials
considerably. \cref{alg:NRST_adapt} gives an overview of the adaptation process that is
described in this section. At each iteration, NRPT is run with a doubling 
computational budget to construct  $\Vs$, which then is used to
adapt both the grid and the affinities. In turn, the improved grid should result
in NRPT samples of higher quality on the next iteration.

Although NRPT is used to tune NRST, NRST has key advantages over NRPT once tuned.
Indeed, parallelization in NRPT is tightly linked to the size of the grid. Since
there is an optimal value for the latter, adding more processors past this point
results in diminishing gains \citep[][\S 5.3]{syed2022nrpt}. In contrast,
regenerative processes like NRST always benefit from adding more processors
\citep{fishman1983accelerated}.
Moreover, whereas NRST requires no communication between workers, NRPT 
involves repeated messaging between random pairs of processors, 
making distributed implementations highly non-trivial \citep[see e.g.][]{surjanovic2023pigeons}.
NRST also admits the TE diagnostic derived from its simple CLT, whereas no 
such diagnostic for NRPT exists. Finally, as discussed in
\cref{sec:limit_grid}, NRST exhibits a lower barrier than NRPT for any given 
problem.

\subsection{Adjusting the affinities}\label{sec:adjust_affinities}
We restrict attention to the two affinity settings from \cref{sec:limit_grid}. 
In the \textbf{mean energy} setting, the affinity function has the closed 
form expression $c(\beta)=-\log(\mcZ(\beta)))$. We can thus employ existing 
approaches for approximating log-normalizing constants 
\citep[see][for a review]{fourment2019dubious}.
Among these, the \emph{stepping-stone} estimator \citep{xie2010improving} is a
good fit for simulated tempering algorithms. 
In \cref{app:norm_constants} we detail how this approach can be used to produce
an estimate $\widehat{\log(\mcZ(\beta_i))}$ of the log-normalization constant
at each grid point. 
In the \textbf{minimum rejection rate} setting, computing the minimum rejection affinity 
function requires numerically integrating \cref{eq:def_min_rej_rate_affinity}.
We use a trapezoidal approximation
\[
\frac{1}{2}\sum_{j=1}^i (\med^{(\beta_{j-1})}(V)+\med^{(\beta_j)}(V))\dbeta_j.
\]
$\med^{(\beta)}(V)$ can then be approximated using the sample medians of $\Vs$.

\subsection{Adapting the grid}\label{sec:adapt_grid}

\cref{thm:TE_formula} shows that under mild conditions, $\TE$ is maximized when 
$\sum_{i=1}^N \frac{\rho_i}{1-\rho_i}$ is minimized. 
\cref{lemma:opt_constrained_sum_convex} 
indicates that for fixed $\sum_{i=1}^N \rho_i$ (justified by \cref{prop:rejs_Riemann_sums}), 
the minimum occurs when $\forall i, \,\,\rho_i = \rho \in (0,1)$, i.e., the grid exhibits 
\emph{equi-rejection}.
In order to find a grid that satisfies this condition, we adapt the tuning
approach outlined in \citet{syed2022nrpt} to ST. The
idea leverages \cref{prop:rejs_Riemann_sums}: 
under equi-rejection, the finite grid approximation of
\cref{eq:rejs_Riemann_sums} becomes
\[\label{eq:finite_Riemann_sum_equi-rejection}
\forall i\in\indset:\ \sum_{j=1}^i \rho_i = i \rho \approx \Lambda(\beta_i), \qquad \rho \approx \Lambda/N.
\]
This suggests setting $\beta_i = \Lambda^{-1}\left((i/N)\Lambda\right)$. In practice, we must
use estimates of $\rho_i$ to approximate this expression. In particular, let 
$r_{i,i+\eps}$ be the Monte Carlo estimate of $\rho_{i,i+\eps}$ using the $\Vs$ dataset, and $r_i = (r_{i-1,i}+r_{i,i-1})/2$. 
Then we let $\widehat\Lambda(\beta)$, $\beta\in[0,1]$ be a monotonic interpolation of
\[\label{eq:def_estimate_temp_barrier}
\forall i\in\indset:\ \widehat\Lambda(\beta_i) := \sum_{j=1}^i r_j, \qquad \widehat\Lambda(1) = \widehat\Lambda = \sum_{j=1}^N r_j.
\]
Finally, we set $\beta_i = \widehat\Lambda^{-1}\left( (i/N)\widehat\Lambda\right)$ using a root finding algorithm.

\begin{figure}[t]
\centering
\includegraphics[width=\plotwidth]{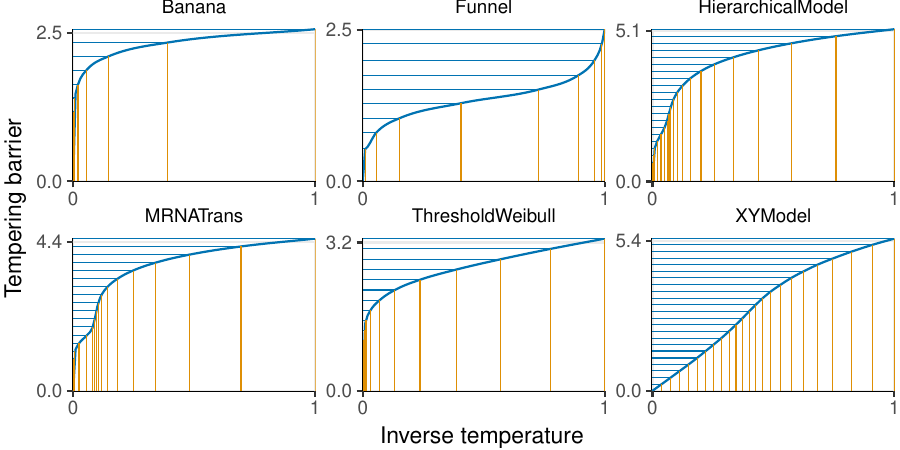}
\caption{Monotonic interpolation of the estimated tempering barrier
(\cref{eq:def_estimate_temp_barrier}) for an array of six models. Vertical lines 
(yellow) show the grid after running \cref{alg:NRST_adapt}.
Evenly-spaced horizontal lines (blue) signal that the grid satisifes equi-rejection.
}
\label{fig:barriers}
\end{figure}

\cref{fig:barriers} shows the result of applying the grid adaptation strategy to
a selection of six models (see \cref{app:models_description}).
Note that the optimal points are far from uniformly distributed. 
They also do not follow
a common pattern, like a logarithmic spacing concentrated at $\beta=0$. Indeed,
the optimal grid can concentrate in arbitrary parts of the unit interval,
highlighting the necessity of the adaptation scheme described in this section.

\subsection{Convergence detection}\label{sec:detect_convergence}

We propose three heuristic indicators and thresholds that, 
when \emph{jointly} satisfied, signal that \cref{alg:NRST_adapt} can 
be safely terminated. The first---the ratio of the standard deviation and mean of the rejection probabilities---tracks 
agreement with the  equi-rejection criterion. Agreement occurs when
\[
\frac{\sqrt{N^{-1}\sum_{i=1}^N (r_i-\br)^2}}{\br} < L_r,
\]
where $\br := N^{-1}\sum_{i=1}^N r_i$. The other two indicators track the
convergence in level affinities and tempering barrier, which have stabilized when
\[
\frac{|c_\text{new}(1)-c_\text{old}(1)|}{c_\text{old}(1)} < L_c \qquad \text{and} \qquad  \frac{|\widehat\Lambda_\text{new} - \widehat\Lambda_\text{old}|}{\widehat\Lambda_\text{old}} < L_\Lambda.
\]
Furthermore, when the mean-energy strategy is chosen for the affinity function,
the difference in directional rejection probabilities is also a useful diagnostic
(c.f. \cref{eq:def_sym_rejs}), which we monitor via the threshold
\[
\frac{\frac{1}{N}\sum_{i=1}^N|r_{i,i-1}-r_{i-1,i}|}{\br} < L_d.
\]
Based on experimental evidence, we recommend setting $L_r=0.1$, $L_c=0.005$,
$L_\Lambda=0.01$, and $L_d=0.05$. These thresholds do not have a
theoretical grounding---they are heuristics based on extensive testing on a variety of models.

\subsection{Choice of exploration kernels}\label{sec:choice_exp_kernel}

The basic theory described in \cref{sec:NRST} for the correctness of NRST---%
particularly the Law of Large numbers (\cref{thm:LLN})---places minimal 
requirements on the collection of kernels $\{K_i\}_{i=1}^N$; only 
$\target_i$-invariance and $\target_i$-irreducibility are assumed. In practice,
using different exploration kernels, each with many tunable parameters, can make 
the adaptation of NRST harder. This problem is further complicated by the
``moving target'' effect generated by the adaptation described in 
\cref{alg:NRST_adapt}: at each round we have a new grid and therefore a new
path of distributions, which in theory requires re-tuning the explorers.

Therefore, it is convenient to rely on a single parametric family of exploration 
kernels that can be used without much or any tuning effort. If adaptation is 
required, a simple yet effective strategy is to optimize the kernel for the 
target distribution $\target=\target_N$---which is always part of the path $\{\target_i\}_{i=0}^N$, thus avoiding the moving target issue.
In our experiments (\cref{sec:numerical_results}) we opt for using the slice 
sampler \citep{neal2003slice} because it
can be used without any adaptation across all the experiments we covered.
\citet{biron2023automala} is a more recent example of an MCMC sampler that can be 
used with minimal inter-round adaptation.

\subsection{Selecting the grid size}\label{sec:tune_N}

We now discuss tuning the grid size $N$. If $N$ is too low, the sampler will struggle to move between 
endpoints due to the high probability of rejection in tempering. Increasing $N$, on the other hand, 
generally increases the cost of running a single tour. A natural candidate to capture this trade-off is the ratio of
expected tour length and tour effectiveness. This ratio captures the total number of exploration steps taken
to produce a result for fixed 
$(\alpha, \delta)$ in \cref{eq:def_Kmin}, i.e., fixed result quality:
\[
\frac{\EE_\nu[\tau]}{\TE} \propto \Kmin(\alpha,\delta,\TE)\EE_\nu[\tau].
\]

Under \cref{assu:ELE,assu:mean_energy_affinity}, the expected tour length is 
$\EE_\nu[\tau] = 2(N+1)$. Assuming further that equi-rejection is approximately
satisfied, and further that $\rho \approx \Lambda/N$, the ratio has the form 
\[
\frac{\EE_\nu[\tau]}{\TE} \approx 2(N+1)\left(1 + 2N\frac{\Lambda/N}{1-\Lambda/N}\right) = \frac{2(N+1)(N(1+2\Lambda)-\Lambda)}{N-\Lambda}.
\]
\cref{lemma:TE_over_tourlength_maximizer} shows that 
this expression is minimized over $N\in\reals$, $N > \Lambda$ by
\[\label{eq:def_opt_N}
N^\star := \Lambda\left[1+\sqrt{1+\frac{1}{1+2\Lambda}}\right].
\]
Note that $N^\star \in (2\Lambda, (1+\sqrt{2})\Lambda)$, so that the grid
size grows linearly with the total tempering barrier.
A practical way of using \cref{eq:def_opt_N} is to estimate $\Lambda$ by 
$\widehat{\Lambda}$ during the tuning rounds in \cref{alg:NRST_adapt}. Then,
if the difference between the current $N$ and the estimated $N^\star$ is 
substantial, the tuning can be re-started with this new value.
In practice, we use $\gamma N^\star$ tempering levels for some $\gamma \geq 1$ 
to account for the various approximations and assumptions used in deriving $N^\star$. In \cref{app:hyperparams} we show that the safety factor 
$\gamma=2$ results in more robust performance without sacrificing efficiency.

\subsection{Managing the serial correlation in exploration steps}\label{sec:serial_correlation}

\cref{assu:ELE} is an important condition 
of many findings in \cref{sec:analysis}, and suggests that a high
correlation in the potential $V(x)$ before and after
an exploration step might hinder the performance of ST.
One can decrease this correlation, if necessary, by composing the exploration kernels $\{K_i\}_{i=1}^N$ 
with themselves. Specifically, for each 
$i\in\indsetone$, let $K_i^0:=\Id$ be the identity kernel, and define for $n > 1$, $K_i^n := K_i^{n-1} \circ K_i$.
Under mild regularity conditions, the serial correlation in $V(x)$ can be
decreased by substituting $K_i$ in the exploration step with $K_i^n$ for $n> 1$.
The best value of $n$ for a given inference problem
will generally depend on the kernel $K_i$. 
We suggest selecting a
maximum tolerable correlation $\bkappa<1$, and then setting, for all $ i \in \indsetone$, $n_i^\star := \inf\{n\in\posInts: \kappa_i(n) \leq \bkappa\},$
where $\kappa_i(n)$ denotes the correlation $\kappa_i(n) = \corr(V(x_0), V(x_n))$ for $x_0\sim\pi_i$ and $x_n\sim K_i^n(x_0, \cdot)$.
The autocorrelation functions can be approximated, for example, by running the
exploration kernels at each point in the final
grid obtained in \cref{alg:NRST_adapt}. In \cref{app:hyperparams} we show that
$\bkappa=0.95$ results in the best performance.

\section{Related work}\label{sec:related_work}

Non-reversible extensions of simulated tempering have appeared previously
in the literature. \citet{sakai2016irreversible} apply the general technique
described in \citet{turitsyn2011irreversible} to transform reversible chains 
into non-reversible ones. Indeed, \citet{sakai2016irreversible} start from the 
symmetric random walk ST and build a continuum of MCMC algorithms controlled by 
a parameter $\delta\in[0,1]$. The original reversible algorithm is recovered when 
$\delta=0$. For $\delta>0$ however, the Markov chain is non-reversible and its 
performance improves monotonically as $\delta$ increases. NRST is similar to the 
result of taking $\delta=1$, although not exactly equal: the algorithm with $\delta=1$ allows for the 
possibility of not updating neither $i$ nor $\eps$ during one step,
while this is impossible in our formulation. \citet{faizi2020simulated} 
also apply the \citet{turitsyn2011irreversible} approach, but starting instead 
from a reversible ST algorithm that updates the temperature using the full 
conditional $\pi_\text{st}(i|x)$ corresponding to \cref{eq:def_ST_lifted} \citep{rosta2010error}.

There exists a vast literature describing heuristics to set the affinities
\citep[see e.g.][]{park2007choosing,chelli2010optimal,
nguyen2013communication,sakai2016learning}. The mean-energy strategy is adopted 
(implicitly or explicitly)
in all of these works, so that log-normalization constants need to be estimated. 
In particular, \citet{mitsutake2000replica} proposes using a preliminary run of 
parallel tempering with the same grid in order to estimate these constants.
\citet{geyer1995annealing} also considers the possibility of using parallel 
tempering, but recommends instead an alternative method.

In contrast, less attention has been put on automatically selecting grid points. 
These are generally determined using either a uniformly or logarithmically spaced 
grid between endpoints. \citet{geyer1995annealing} is one of the few works that 
tackle this issue, aiming to obtain equi-rejection via an ad-hoc model of the 
acceptance rate as a function of $\beta$.
Alternatively, recent work \citep{gobbo2015extended,graham2017continuously}
has sidestepped the issue of finding a grid by working with a continuum of 
inverse temperatures. Furthermore, building on these developments, 
\citet{martinsson2019simulated} propose marginalizing-out $\beta$ via averaging 
results for Markov processes \citep[see e.g.][Ch.\ 17]{pavliotis2008multiscale}, 
however this line of work applies only to
smooth target distributions, whereas NRST is applicable to general state-spaces.

Throughout this paper we have suggested running NRST for a fixed number of tours
determined via \cref{eq:def_Kmin} and then estimating expectations under $\target$
via \cref{eq:def_ratio_est_implicit}. Nevertheless, alternative stopping criteria 
and estimators have been proposed in the regenerative simulation literature 
\citep[see e.g.][and references therein]{glynn2006simulation}.

The coefficient of variation of the tour lengths $\{\tau_k\}_{k\in\nats}$ is a 
common diagnostic used in regenerative simulation 
\citep[see][and references therein]{mykland1995regeneration}. 
It arises from bounding the asymptotic variance of estimates of unconditional expectations---in 
contrast to TE which bounds the error of estimating a  conditional expectation that is specific to ST. 

Regenerations allow adapting MCMC by using past information without disrupting the
stationary distribution (see \citealt{gilks1998adaptive} and 
\citealt{lan2014wormhole} for an application). 
Lastly, there has been work towards unifying ST with other MCMC 
methods aimed at sampling from finite collections of distributions---like 
Wang-Landau and trans-dimensional MCMC---in order to propose more efficient
sampling strategies for all \citep{atchade2010wang,tan2017optimally}.

\section{Numerical results}\label{sec:numerical_results}

In this section we show experimental evidence for the effectiveness of our 
proposed adaptation algorithm. Moreover, we provide a numerical assessment of 
the closeness between the estimator $\TEest$ (\cref{eq:def_TEest}) and the limiting
value derived in \cref{cor:limit_TE_unif}.

Since the focus of this paper is on \emph{automatic} regenerative samplers,
in the following we have excluded comparisons against non-ST regenerative samplers.
As discussed in \cref{sec:intro,sec:background}, we
have not found other regenerative MCMC methods that are as automatic as ST.%
\footnote{%
One notable exception is Independence MH (IMH), for 
which regenerative simulation is possible regardless of the proposal distribution 
\citep[][\S 4.1]{mykland1995regeneration}. But note that we can view IMH as an 
improperly tuned NRST sampler with $N=1$, $K_1=\text{Id}$, and reference given by 
the proposal. Following the reasoning of \cref{sec:adaptation}, we would expect 
that NRST with adaptation and non-trivial exploration kernel to dominate IMH.%
}
Hence, we restrict ourselves to ST algorithms because 
putatively more efficient regenerative samplers that require model specific 
derivations lie outside of the scope of this paper.

\subsection{Evaluation of the ST adaptation algorithm}\label{sec:eval_ST_adaptation}
We first evaluate the effect and applicability of the adaptation
algorithm in \cref{sec:adaptation}.  In particular, we compare the performance
of multiple ST algorithms both with and without parameters set using the proposed strategy:
our novel NRST method, the original reversible simulated tempering of
\citet{geyer1995annealing} (abbreviated henceforth as ``GT95''), and the two
non-reversible algorithms of \citet{sakai2016irreversible} (``SH16'') and 
\citet{faizi2020simulated} (``FBDR''). When \emph{not} using our proposed method for tuning,
we default to each past method's recommended tuning method (see 
\cref{app:tuning_details} for further details). 

We compare the performance of these ST methods using a notion of 
\emph{sample-quality per unit cost}. We fix quality across all ST algorithms by 
fixing a particular value of the confidence interval half-width $\delta>0$ 
and level $\alpha$ in \cref{alg:NRST_par}. We measure \emph{cost} via the number 
of $V$ evaluations. This is an implementation-agnostic 
proxy of wall-clock time: repeated evaluation of $V$ accounts for the vast majority
of the time during our experiments. We define the ``parallel cost'' as the maximum
number of $V$ evaluations across tours, and the ``serial cost'' as the sum of $V$ evaluations
across tours.

For all ST algorithms and all tempering levels, we use slice sampling within Gibbs
as the exploration kernel \citep{neal1997markov,neal2003slice}; the natural adaptivity of slice sampling
avoids the need to re-adjust exploration kernels at each round of \cref{alg:NRST_adapt}.

\begin{figure}[t]
\centering
\includegraphics[width=\plotwidth]{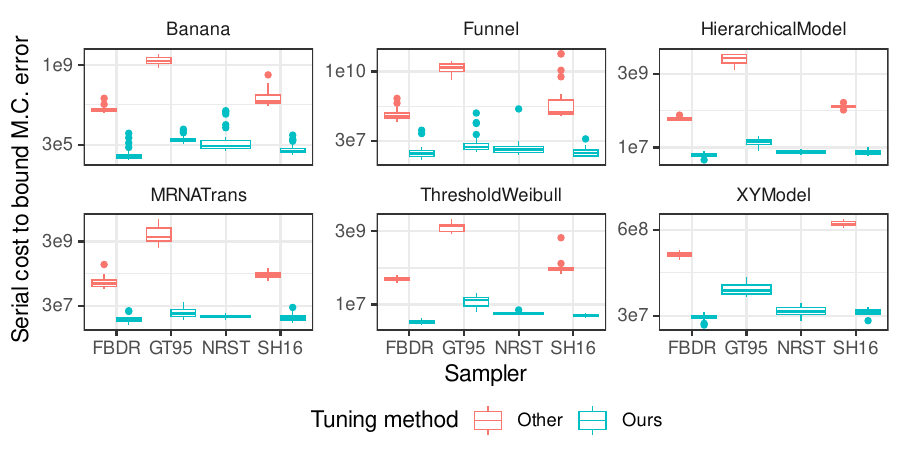}
\caption{Side-by-side box-plots of serial cost (as defined in 
\cref{sec:eval_ST_adaptation}) for each simulated tempering algorithm and tuning 
strategy (color), for $6$ selected models (lower is better). Each experiment is replicated $30$ 
times. The number of tours is determined in a quality-consistent 
approach using \cref{eq:def_Kmin} with $\alpha=0.95$, $\delta=0.5$, and $\TEest$ 
estimated in a preliminary run. Missing box for GT95 in XYModel is due to all 
replications exceeding the time budget.}
\label{fig:benchmark}
\end{figure}

\cref{fig:benchmark} shows the results of the comparison. We first 
notice that our automatic adaptation strategy achieves multiple orders of 
magnitude in cost reduction across all competitors and models. In particular, the
worst performing algorithm under the alternative tuning strategy is GT95. This is
due to the deterioration of performance that this reversible sampler experiences
in the dense grid regime. However, when inspecting the results under our 
adaptation scheme, we see a much smaller difference between GT95 and the rest.
This is because the grids produced by our tuning tend to be sparse, and
non-reversibility gains are smaller (although usually still meaningful) on such grids.
Lastly, it is interesting that the three non-reversible algorithms become 
essentially equivalent using our adaptation strategy, although FBDR achieves
consistent (if small) gains across examples.

\subsection{Tightness of the asymptotic tour effectiveness}\label{sec:TEinfty_tightness}

\begin{figure}[t]
\centering
\includegraphics[width=\shorterplotwidth]{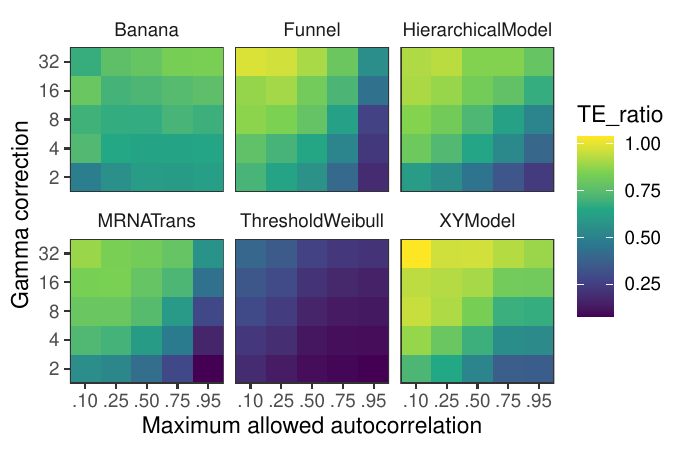}
\caption{Ratio between the average $\TE$ from $30$ replications of a given
configuration, over the estimated $\TEinfty$ using the same 
replications. Grid represents different combinations of $\bkappa$ and $\gamma$,
with the cost-optimal value represented in the bottom-right corner.}
\label{fig:TE_ELE}
\end{figure}

In \cref{sec:limit_grid} we showed that, under certain conditions, TE converges to a fixed value $\TEinfty$ as the grid refines.
This experiment assesses how close the non-asymptotic $\TE$ is to its 
limiting value, and if it is possible to achieve $\TEinfty$ in certain situations.

\cref{fig:TE_ELE} shows the ratio between the mean $\TEest$ and 
$\TEinfty$ over $30$ repetitions of NRST runs that use a range of 
combinations of the hyper-parameters $(\gamma,\bkappa)$. The
default, cost-optimal values are represented in the bottom-right corner.
Notice that these configurations achieve already about $20\%$ of
$\TEinfty$. Also, as either $\gamma$ increases
or $\bkappa$ decreases, the estimated tour effectiveness approaches the limiting
value, validating our analyses. However, convergence is
slower for more complex target distributions; e.g., the ``ThresholdWeibull'' 
achieves only $40\%$ of $\TEinfty$. This occurs because estimating
the log-partition function for this model is hard, since the 
distribution of $V$ under $\pibeta$ becomes heavy tailed as $\beta\to0$.

\section{Regenerative NRST in practice}\label{sec:practical-nrst}

\subsection{Regenerative NRST and Massive Parallelization}\label{sec:parallel}

In this section, we offer practical advice on deploying NRST for complex problems on
massive computing platforms, and how the particular
workload patterns of regenerative MCMC can take advantage of newer architectures for 
distributed computation. Throughout this section we assume that (1) the model is complex
enough that the compute time for a single tour is fairly long; (2) the computing platform 
dynamically assigns the NRST tours among a pool of concurrent workers of pre-specified size; 
and (3) that we have already run the adaptation routine in 
\cref{sec:adaptation}. Our aim is to devise a strategy to choose the 
size of the pool to satisfy a time or cost constraint.

We begin with a preliminary run of $K_\text{trial}$ parallel NRST tours 
(\cref{alg:NRST_par}). To determine $K_\text{trial}$ we need an estimate of TE in \cref{eq:def_Kmin};
we use the mean-energy strategy (see \cref{sec:adjust_affinities})
and estimate TE via the asymptotic formula in \cref{cor:limit_TE_unif}; i.e.,
$\widehat{\TEinfty} := 1 / (1+2\widehat\Lambda)$, where 
$\widehat\Lambda=\widehat\Lambda(1)$ is the 
estimated tempering barrier (\cref{eq:def_estimate_temp_barrier}) returned by
\cref{alg:NRST_adapt}.
The preliminary run provides an estimate of $\TEest$
using \cref{eq:def_TEest}, which determines the required number 
of tours $K$ to satisfy the $(\alpha,\delta)$ accuracy guarantee. Note that 
the $K_\text{trial}$ tours can be used here too, so only 
the difference $K_\text{extra}:=K-K_\text{trial}$ needs to be executed in the 
next phase. Given that the asymptotic TE tends to overestimate the true TE
(\cref{sec:TEinfty_tightness}), it is usually the case that $K_\text{extra}$ 
is significantly larger than $K_\text{trial}$.
The preliminary run also provides CPU times incurred by tours;
we fit this distribution via a bulk-tail mixture with an $80^\text{th}$
percentile threshold, using an empirical distribution for the bulk, and a Weibull 
density for the tail. This CPU time model enables simulating hypothetical running 
times and costs for samples of $K_\text{extra}$ tours.

\begin{figure}[t]
\centering
\includegraphics[width=\plotwidth]{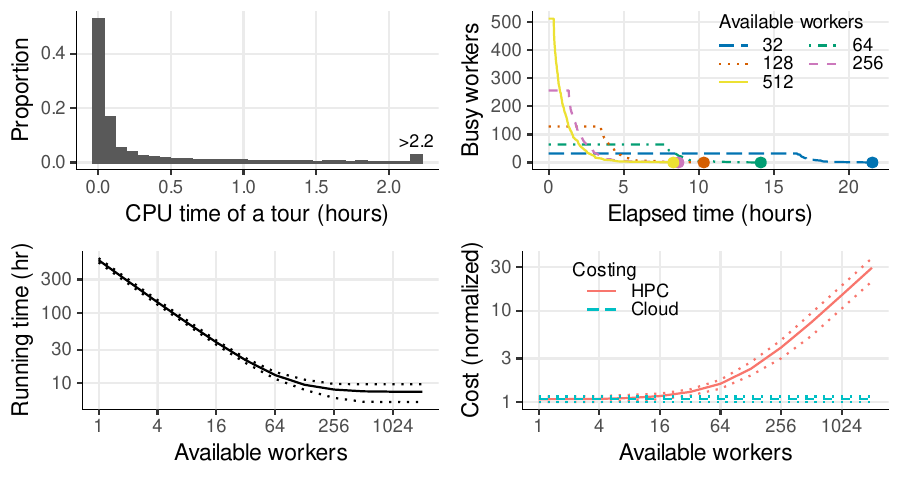}
\caption{Graphic analysis of a simulation plan to execute $2048$ expensive tours 
in parallel.
\textbf{Top-left}: histogram of simulated CPU times (truncated at the $97.5$-th 
percentile). 
\textbf{Top-right}: number of active workers at any given 
time when running $2048$ tours using a varying number of available workers,
assuming dynamic scheduling of tours. Dots mark the total running time in each 
case.
\textbf{Bottom-left}: estimated running time of the algorithm versus number of 
available workers. Solid line denotes the mean over $30$ repetitions, while 
dotted lines present an $80\%$ confidence interval.
\textbf{Bottom-right}: estimated cost (normalized so that the minimum is $1$) 
of running the algorithm versus available workers, depending on costing rule. 
Solid and dashed lines denote the means over $30$ repetitions, while dotted 
envelopes give $80\%$ confidence intervals.
}
\label{fig:workers_time_cost}
\end{figure}

\cref{fig:workers_time_cost} displays various simulations that can be used to plan the execution of $K_\text{extra}$ 
tours (we fix $K_\text{extra} = 2048$ for this example). 
The top-left panel in \cref{fig:workers_time_cost} shows the CPU time distribution of 
$K_\text{extra}$ tours, highlighting how unbalanced the 
workload of a typical parallelized run is.
The top-right and bottom-left panels
display hypothetical usage patterns for worker pools of varying sizes
when processing the $K_\text{extra}$ tours. These figures demonstrate that when few workers are available,
doubling the pool size almost halves the runtime; but after a certain point, the longest tour (about $8$ hours in this sample) begins 
to dominate and returns diminish.

This workload pattern has implications for execution 
on High-Performance Computing (HPC) and Cloud Computing (Cloud) platforms,
which have very different characteristics. 
Jobs submitted to HPC systems often do not incur monetary costs---as they are often owned 
by the user's organization---but they are queued until
enough resources are free to accommodate them 
\citep[see][]{fan2019scheduling}. And typically, the scheduler
assumes that all workers requested will be busy for the duration of the job;
i.e.,
\[
\text{HPC cost} \propto \text{Running time} \times \text{Pool size}.
\]
As we know from the top-right plot in \cref{fig:workers_time_cost}, this is not 
the case for the usual workloads in regenerative MCMC. Thus, if we requested 
$2048$ workers to execute the running example, the job might wait for an 
unacceptably long time. A better choice would be
$64$ workers; when compared to a single worker, this configuration achieves 
upwards of a $40$ fold reduction in runtime while only doubling the cost.

In contrast, Cloud platforms often have little to no scheduling wait time, but
do commonly incur monetary cost proportional to the total CPU time,
\[\label{eq:def_total_cpu_time_cost}
\text{Cloud cost} \propto \sum_{k=1}^{K_\text{extra}} (\text{CPU time of the }k^\text{th}\text{ tour}).
\]
The sum on the right-hand side of \cref{eq:def_total_cpu_time_cost} 
corresponds to the area under the curves in the top-right panel of 
\cref{fig:workers_time_cost}. Since this sum is invariant to tour allocation, 
the cost is constant versus pool size. In theory, then, we should always select a pool of
size equal to $K_\text{extra}$ on a Cloud platform. However, Cloud platforms usually have a maximum concurrency 
limit---for example, Amazon Lambda \citep{amazonlambda} currently allows a 
maximum of $1000$ concurrent workers---and if the initialization time of 
a worker is significant, then it may be better to 
amortize that time by assigning more than one tour to each worker.

\subsection{Regenerative NRST and Probabilistic Programming}\label{sec:PPLs}

Our approach to regenerative MCMC can be easily integrated
within many PPLs. In the Bayesian setting, when the prior $\varpi$ is 
proper and allows \iid simulation, then using $\varpi$ as the 
reference and setting $V(x)=-\log(L(y|x))$ gives a valid path of 
distributions to the target posterior density
$\target(x) \propto \varpi(x)L(y|x)$. If this model is coded in a PPL
that exposes a minimal set of functions, we can apply regenerative MCMC to 
the model without any extra user input. The basic required building 
blocks are
1) a simulator for $\varpi$, 2) a function that returns $\log(L(y|x))$ for any 
$x$, and 3) a method to evaluate $\log(\varpi(x))$ (required for the exploration 
step). For example, our Julia implementation of NRST\footnote{%
\if0\blind{\url{https://github.com/UBC-Stat-ML/NRST.jl}}\fi%
\if1\blind{[\textcolor{blue}{repo URL withheld for double blind review}]}\fi%
}
can seamlessly handle models written in DynamicPPL \citep{tarek2020dynamic}, which exposes 
these three building blocks. More generally, PPLs that compile probabilistic models into graphical 
models---like WinBUGS and descendants \citep{Lunn2000}---can in theory
provide the type of interface described above \citep[\S 3]{vandemeent2021introduction}.

While the scenario above covers many cases of interest, it
fails if the prior is improper or difficult to sample from
(e.g., a Markov random field), or if the PPL does not
expose the necessary functionality. However, as long as the PPL provides
the ability to evaluate $\log(\target(x))$ for any $x$---as is the case with Stan \citep{stan2023stan}---we
can still apply our approach. In particular, we select a parametrized family of references $\mcQ:=\{q_\theta:\theta \in \Theta\}$ 
such that each $q_\theta$ enables \iid sampling and evaluation of 
$\log(q_\theta(x))$, and then set $V(x)=-\log(\target(x)/q_\theta(x))$ to obtain a valid path of distributions.
The \emph{best} member $q_{\theta^\star}$ may then be chosen by any automatic variational 
procedure, e.g., \citet{kucukelbir2015automatic}. This is similar to the approach for NRPT 
described in the work of \citet{surjanovic_parallel_2022}.

\section{Discussion}

In this work we developed the tools necessary to use NRST for automated 
regenerative MCMC. We derived the TE diagnostic as a key summary of the behavior 
of ST algorithms, and used its relationship 
with algorithmic parameters to provide a robust adaptation procedure that 
facilitates the use of NRST in a wide range of probabilistic models. Our 
experiments demonstrated that our tuning recommendations produce sizable gains 
on both reversible and non-reversible ST algorithms.

Future work includes generalizing \cref{thm:TE_formula} to arbitrary level
distributions. This would allow us to assess the robustness of NRST under
imperfect tuning---i.e., when $c_i\neq -\log(\mcZ(\beta_i))$---which in practice
is how NRST is implemented. A generalized \cref{thm:TE_formula} could 
also provide a better understanding of the median affinity strategy, which in our 
experiments seems to yield performance gains in settings where there are more 
workers available than tours to run (see \cref{fig:hyperparams_fun} in 
\cref{app:hyperparams}).

Elucidating the connection between TE and the rate of convergence of NRST to its 
stationary distribution is another potential avenue of future research.
One way to approach this would be to leverage Theorem 13.3.1 in
\citet{douc2018markov}, which
proves geometric ergodicity for atomic chains when the moment generating function 
(m.g.f.) of the return time $T_1$ is analytic at the origin. Under 
\cref{assu:ELE}, the index process is a finite Markov chain, so the m.g.f.\
condition holds \citep[see e.g.][Prop.\ 6.1]{nemirovsky2013tensor}. By 
linking the radius of convergence of the m.g.f.\ to the TE diagnostic, we could
obtain an expression for the rate of geometric ergodicity in terms of TE.

Annealed importance sampling (AIS, \citealp{neal_annealed_2001}) is a Monte
Carlo algorithm that bears some similarities to ST. It works by sequentially 
pushing \iid samples from a reference $\target_0$ towards the target $\target$
through the tempering path in \cref{eq:def_temp_dist} using a collection of
$\pibeta$-invariant Markov kernels. Moreover, this process is also embarrassingly parallelizable. On the other hand, AIS is not a proper MCMC method and thus we
have excluded it from our experiments. Nevertheless, it would be interesting to 
establish both theoretical and experimental comparisons between NRST and AIS.

\section*{Acknowledgements} 
We thank Son Luu for his helpful feedback on Theorem 4.4. We also thank an
editor and two reviewers for their engaging comments and suggestions.
We further thank one of our reviewers for helping us simplify the proofs of 
\cref{lemma:opt_constrained_sum_convex,lemma:maximize_rational_fun,lemma:TE_over_tourlength_maximizer}.
ABC and TC acknowledge the support of an NSERC Discovery Grant.
We also acknowledge use of the ARC Sockeye computing platform from the
University of British Columbia.

\bibliographystyle{apalike}
\bibliography{ref}

\begin{thebibliography}{}

\bibitem[Amazon, 2023]{amazonlambda}
Amazon (2023).
\newblock {AWS Lambda} developer guide.
\newblock \url{https://docs.aws.amazon.com/lambda/latest/dg/welcome.html}.
\newblock Accessed: 2023-06-29.

\bibitem[Atchad{\'e} and Liu, 2010]{atchade2010wang}
Atchad{\'e}, Y.~F. and Liu, J.~S. (2010).
\newblock The {W}ang-{L}andau algorithm in general state spaces: applications
  and convergence analysis.
\newblock {\em Statistica Sinica}, 20(1):209--233.

\bibitem[Ballnus et~al., 2017]{ballnus2017comprehensive}
Ballnus, B., Hug, S., Hatz, K., G\"{o}rlitz, L., Hasenauer, J., and Theis,
  F.~J. (2017).
\newblock Comprehensive benchmarking of {M}arkov chain {M}onte {C}arlo methods
  for dynamical systems.
\newblock {\em {BMC} Systems Biology}, 11(1).

\bibitem[{Biron-Lattes} et~al., 2023]{biron2023automala}
{Biron-Lattes}, M., {Surjanovic}, N., {Syed}, S., {Campbell}, T., and
  {Bouchard-C{\^o}t{\'e}}, A. (2023).
\newblock {autoMALA: Locally adaptive Metropolis-adjusted Langevin algorithm}.
\newblock {\em arXiv e-prints}, page arXiv:2310.16782.

\bibitem[Brockwell and Kadane, 2005]{brockwell2005identification}
Brockwell, A.~E. and Kadane, J.~B. (2005).
\newblock Identification of regeneration times in {MCMC} simulation, with
  application to adaptive schemes.
\newblock {\em Journal of Computational and Graphical Statistics},
  14(2):436--458.

\bibitem[\c{C}inlar, 2013]{cinlar2013introduction}
\c{C}inlar, E. (2013).
\newblock {\em Introduction to Stochastic Processes}.
\newblock Dover Books on Mathematics. Dover Publications, reprint edition.

\bibitem[Chelli, 2010]{chelli2010optimal}
Chelli, R. (2010).
\newblock Optimal weights in serial generalized-ensemble simulations.
\newblock {\em Journal of Chemical Theory and Computation}, 6(7):1935--1950.

\bibitem[Chen et~al., 1999]{chen1999lifting}
Chen, F., Lov\'{a}sz, L., and Pak, I. (1999).
\newblock Lifting {M}arkov chains to speed up mixing.
\newblock In {\em ACM symposium on theory of computing}.

\bibitem[Cheng, 2017]{cheng2017inference}
Cheng, R. (2017).
\newblock {\em {Non-Standard Parametric Statistical Inference}}.
\newblock Oxford University Press.

\bibitem[{Douc} et~al., 2022]{douc2022kickkac}
{Douc}, R., {Durmus}, A., {Enfroy}, A., and {Olsson}, J. (2022).
\newblock {Boost your favorite Markov Chain Monte Carlo sampler using Kac's
  theorem: the Kick-Kac teleportation algorithm}.
\newblock {\em arXiv:2201.05002}.

\bibitem[Douc et~al., 2018]{douc2018markov}
Douc, R., Moulines, E., Priouret, P., and Soulier, P. (2018).
\newblock {\em Markov chains}.
\newblock Springer International Publishing.

\bibitem[Dudley, 2002]{dudley2002real}
Dudley, R.~M. (2002).
\newblock {\em Real analysis and probability}.
\newblock Cambridge studies in advanced mathematics. Cambridge University
  Press, 2nd edition.

\bibitem[Faizi et~al., 2020]{faizi2020simulated}
Faizi, F., Buigues, P.~J., Deligiannidis, G., and Rosta, E. (2020).
\newblock Simulated tempering with irreversible {G}ibbs sampling techniques.
\newblock {\em The Journal of Chemical Physics}, 153(21):214111.

\bibitem[Fan et~al., 2019]{fan2019scheduling}
Fan, Y., Lan, Z., Rich, P., Allcock, W.~E., Papka, M.~E., Austin, B., and Paul,
  D. (2019).
\newblock Scheduling beyond {CPUs} for {HPC}.
\newblock In {\em International Symposium on High-Performance Parallel and
  Distributed Computing}.

\bibitem[Fishman, 1983]{fishman1983accelerated}
Fishman, G.~S. (1983).
\newblock Accelerated accuracy in the simulation of {M}arkov chains.
\newblock {\em Operations Research}, 31(3):466--487.

\bibitem[Flegal et~al., 2008]{flegal2008markov}
Flegal, J., Haran, M., and Jones, G. (2008).
\newblock {M}arkov chain {M}onte {C}arlo: Can we trust the third significant
  figure?
\newblock {\em Statistical Science}, pages 250--260.

\bibitem[Fourment et~al., 2019]{fourment2019dubious}
Fourment, M., Magee, A., Whidden, C., Bilge, A., Matsen, F., and Minin, V.
  (2019).
\newblock 19 dubious ways to compute the marginal likelihood of a phylogenetic
  tree topology.
\newblock {\em Systematic Biology}, 69(2):209--220.

\bibitem[Gelman and Meng, 1998]{gelman1998simulating}
Gelman, A. and Meng, X.-L. (1998).
\newblock Simulating normalizing constants: From importance sampling to bridge
  sampling to path sampling.
\newblock {\em Statistical science}, pages 163--185.

\bibitem[Geyer and Thompson, 1995]{geyer1995annealing}
Geyer, C.~J. and Thompson, E.~A. (1995).
\newblock Annealing {M}arkov chain {M}onte {C}arlo with applications to
  ancestral inference.
\newblock {\em Journal of the American Statistical Association},
  90(431):909--920.

\bibitem[Gilks et~al., 1998]{gilks1998adaptive}
Gilks, W.~R., Roberts, G.~O., and Sahu, S.~K. (1998).
\newblock Adaptive {M}arkov chain {M}onte {C}arlo through regeneration.
\newblock {\em Journal of the American Statistical Association},
  93(443):1045--1054.

\bibitem[Glynn, 2006]{glynn2006simulation}
Glynn, P.~W. (2006).
\newblock Simulation algorithms for regenerative processes.
\newblock In Henderson, S.~G. and Nelson, B.~L., editors, {\em Simulation},
  volume~13 of {\em Handbooks in Operations Research and Management Science},
  pages 477--500. Elsevier.

\bibitem[Gobbo and Leimkuhler, 2015]{gobbo2015extended}
Gobbo, G. and Leimkuhler, B.~J. (2015).
\newblock Extended hamiltonian approach to continuous tempering.
\newblock {\em Phys. Rev. E}, 91:061301.

\bibitem[Graham and Storkey, 2017]{graham2017continuously}
Graham, M.~M. and Storkey, A.~J. (2017).
\newblock Continuously tempered {H}amiltonian {M}onte {C}arlo.
\newblock In {\em Proceedings of the Conference on Uncertainty in Artificial
  Intelligence, {UAI} 2017}. {AUAI} Press.

\bibitem[Hasenbusch, 2005]{hasenbusch2005xymodel}
Hasenbusch, M. (2005).
\newblock The two-dimensional {XY} model at the transition temperature: a
  high-precision {M}onte {C}arlo study.
\newblock {\em Journal of Physics A: Mathematical and General}, 38(26):5869.

\bibitem[Hobert et~al., 2002]{hobert_applicability_2002}
Hobert, J.~P., Jones, G.~L., Presnell, B., and Rosenthal, J.~S. (2002).
\newblock {On the applicability of regenerative simulation in Markov chain
  Monte Carlo}.
\newblock {\em Biometrika}, 89(4):731--743.

\bibitem[Jones and Hobert, 2001]{jones_honest_2001}
Jones, G.~L. and Hobert, J.~P. (2001).
\newblock Honest exploration of intractable probability distributions via
  {M}arkov chain {M}onte {C}arlo.
\newblock {\em Statistical Science}, 16(4):312--334.

\bibitem[Kemeny and Snell, 1960]{kemeny1960finite}
Kemeny, J.~G. and Snell, J.~L. (1960).
\newblock {\em Finite Markov chains}.
\newblock Van Nostrand, Princeton, N.J.

\bibitem[{Kucukelbir} et~al., 2015]{kucukelbir2015automatic}
{Kucukelbir}, A., {Ranganath}, R., {Gelman}, A., and {Blei}, D.~M. (2015).
\newblock {Automatic Variational Inference in Stan}.
\newblock {\em arXiv e-prints}, page arXiv:1506.03431.

\bibitem[Lan et~al., 2014]{lan2014wormhole}
Lan, S., Streets, J., and Shahbaba, B. (2014).
\newblock Wormhole {H}amiltonian {M}onte {C}arlo.
\newblock {\em Proceedings of the AAAI Conference on Artificial Intelligence}.

\bibitem[Liu and Chen, 1995]{liu1995blind}
Liu, J.~S. and Chen, R. (1995).
\newblock Blind deconvolution via sequential imputations.
\newblock {\em Journal of the {A}merican {S}tatistical {A}ssociation},
  90(430):567--576.

\bibitem[Lunn et~al., 2000]{Lunn2000}
Lunn, D.~J., Thomas, A., Best, N., and Spiegelhalter, D. (2000).
\newblock {WinBUGS} -- a {B}ayesian modelling framework: concepts, structure,
  and extensibility.
\newblock {\em Statistics and Computing}, 10(4):325--337.

\bibitem[Lyubartsev et~al., 1992]{lyubartsev1992new}
Lyubartsev, A.~P., Martsinovski, A.~A., Shevkunov, S.~V., and
  Vorontsov‐Velyaminov, P.~N. (1992).
\newblock New approach to {M}onte {C}arlo calculation of the free energy:
  Method of expanded ensembles.
\newblock {\em The Journal of Chemical Physics}, 96(3):1776--1783.

\bibitem[Marinari and Parisi, 1992]{marinari1992simulated}
Marinari, E. and Parisi, G. (1992).
\newblock Simulated tempering: a new {M}onte {C}arlo scheme.
\newblock {\em EPL (Europhysics Letters)}, 19(6):451.

\bibitem[Martinsson et~al., 2019]{martinsson2019simulated}
Martinsson, A., Lu, J., Leimkuhler, B., and Vanden-Eijnden, E. (2019).
\newblock The simulated tempering method in the infinite switch limit with
  adaptive weight learning.
\newblock {\em Journal of Statistical Mechanics: Theory and Experiment},
  2019(1):013207.

\bibitem[{McKimm} et~al., 2022]{mckimm2022sampling}
{McKimm}, H., {Wang}, A.~Q., {Pollock}, M., {Robert}, C.~P., and {Roberts},
  G.~O. (2022).
\newblock {Sampling using Adaptive Regenerative Processes}.
\newblock {\em arXiv e-prints}, page arXiv:2210.09901.

\bibitem[Minh et~al., 2012]{minh_regenerative_2012}
Minh, D. L.~P., Minh, D. D.~L., and Nguyen, A.~L. (2012).
\newblock Regenerative {M}arkov chain {M}onte {C}arlo for any distribution.
\newblock {\em Communications in Statistics - Simulation and Computation},
  41(9):1745--1760.

\bibitem[Mitsutake and Okamoto, 2000]{mitsutake2000replica}
Mitsutake, A. and Okamoto, Y. (2000).
\newblock Replica-exchange simulated tempering method for simulations of
  frustrated systems.
\newblock {\em Chemical Physics Letters}, 332(1-2):131--138.

\bibitem[Modi et~al., 2023]{modi2021delayed}
Modi, C., Barnett, A., and Carpenter, B. (2023).
\newblock Delayed rejection {H}amiltonian {M}onte {C}arlo for sampling
  multiscale distributions.
\newblock {\em Bayesian Analysis}, pages 1--28.

\bibitem[Murphy, 2012]{murphy2012machine}
Murphy, K.~P. (2012).
\newblock {\em Machine Learning: A Probabilistic Perspective}.
\newblock The MIT Press.

\bibitem[Mykland et~al., 1995]{mykland1995regeneration}
Mykland, P., Tierney, L., and Yu, B. (1995).
\newblock Regeneration in {M}arkov chain samplers.
\newblock {\em Journal of the American Statistical Association},
  90(429):233--241.

\bibitem[Neal, 1997]{neal1997markov}
Neal, R.~M. (1997).
\newblock Markov chain {M}onte {C}arlo methods based on `slicing' the density
  function.
\newblock Technical Report 9722, University of Toronto, Department of
  Statistics.

\bibitem[Neal, 2001]{neal_annealed_2001}
Neal, R.~M. (2001).
\newblock Annealed importance sampling.
\newblock {\em Statistics and Computing}, 11:125--139.

\bibitem[Neal, 2003]{neal2003slice}
Neal, R.~M. (2003).
\newblock {Slice sampling}.
\newblock {\em The Annals of Statistics}, 31(3):705--767.

\bibitem[Nemirovsky, 2013]{nemirovsky2013tensor}
Nemirovsky, D. (2013).
\newblock Tensor approach to mixed high-order moments of absorbing {M}arkov
  chains.
\newblock {\em Linear Algebra and its Applications}, 438(4):1900--1922.

\bibitem[Nguyen, 2017]{nguyen_regenerative_2017}
Nguyen, A. (2017).
\newblock Regenerative {M}arkov chain importance sampling.
\newblock {\em Communications in Statistics--Simulation and Computation},
  46(5):3892--3906.

\bibitem[Nguyen et~al., 2013]{nguyen2013communication}
Nguyen, P.~H., Okamoto, Y., and Derreumaux, P. (2013).
\newblock Communication: Simulated tempering with fast on-the-fly weight
  determination.
\newblock {\em The Journal of chemical physics}, 138(6):061102.

\bibitem[Nummelin, 1978]{nummelin1978splitting}
Nummelin, E. (1978).
\newblock A splitting technique for {H}arris recurrent {M}arkov chains.
\newblock {\em Zeitschrift f{\"u}r Wahrscheinlichkeitstheorie und verwandte
  Gebiete}, 43(4):309--318.

\bibitem[Oliver, 1954]{oliver1954exact}
Oliver, H.~W. (1954).
\newblock The exact {P}eano derivative.
\newblock {\em Transactions of the American Mathematical Society},
  76(3):444--456.

\bibitem[Pagani et~al., 2022]{pagani2022banana}
Pagani, F., Wiegand, M., and Nadarajah, S. (2022).
\newblock An n-dimensional {R}osenbrock distribution for {M}arkov chain
  {M}onte {C}arlo testing.
\newblock {\em Scandinavian Journal of Statistics}, 49(2):657--680.

\bibitem[Park and Pande, 2007]{park2007choosing}
Park, S. and Pande, V.~S. (2007).
\newblock Choosing weights for simulated tempering.
\newblock {\em Physical Review E}, 76(1):016703.

\bibitem[Pavliotis and Stuart, 2008]{pavliotis2008multiscale}
Pavliotis, G. and Stuart, A. (2008).
\newblock {\em Multiscale methods: averaging and homogenization}.
\newblock Springer Science \& Business Media.

\bibitem[Predescu et~al., 2004]{predescu2004incomplete}
Predescu, C., Predescu, M., and Ciobanu, C.~V. (2004).
\newblock The incomplete beta function law for parallel tempering sampling of
  classical canonical systems.
\newblock {\em The Journal of Chemical Physics}, 120(9):4119--4128.

\bibitem[Ripley, 1987]{ripley1987stochastic}
Ripley, B. (1987).
\newblock {\em Stochastic simulation}.
\newblock John Wiley \& Sons, Ltd.

\bibitem[Robert and Casella, 2004]{robert2004monte}
Robert, C.~P. and Casella, G. (2004).
\newblock {\em {M}onte {C}arlo Statistical Methods}.
\newblock Springer Texts in Statistics. Springer, New York, 2nd edition.

\bibitem[Rosenthal, 2006]{rosenthal2006first}
Rosenthal, J.~S. (2006).
\newblock {\em A first look at rigorous probability theory}.
\newblock World Scientific, 2nd edition.

\bibitem[Rosta and Hummer, 2010]{rosta2010error}
Rosta, E. and Hummer, G. (2010).
\newblock {Error and efficiency of simulated tempering simulations}.
\newblock {\em The Journal of Chemical Physics}, 132(3):034102.

\bibitem[Sahu and Zhigljavsky, 2003]{sahu2003self}
Sahu, S.~K. and Zhigljavsky, A.~A. (2003).
\newblock Self-regenerative {M}arkov chain {M}onte {C}arlo with adaptation.
\newblock {\em Bernoulli}, 9(3):395--422.

\bibitem[Sakai and Hukushima, 2016a]{sakai2016irreversible}
Sakai, Y. and Hukushima, K. (2016a).
\newblock Irreversible simulated tempering.
\newblock {\em Journal of the Physical Society of Japan}, 85(10):104002.

\bibitem[Sakai and Hukushima, 2016b]{sakai2016learning}
Sakai, Y. and Hukushima, K. (2016b).
\newblock Irreversible simulated tempering algorithm with skew detailed balance
  conditions: a learning method of weight factors in simulated tempering.
\newblock In {\em Journal of Physics: Conference Series}, volume 750, page
  012013.

\bibitem[{Stan Development Team}, 2023]{stan2023stan}
{Stan Development Team} (2023).
\newblock Stan modeling language users guide and reference manual, version
  2.32.
\newblock \url{https://mc-stan.org}.
\newblock Accessed: 2023-07-05.

\bibitem[{Surjanovic} et~al., 2023]{surjanovic2023pigeons}
{Surjanovic}, N., {Biron-Lattes}, M., {Tiede}, P., {Syed}, S., {Campbell}, T.,
  and {Bouchard-C{\^o}t{\'e}}, A. (2023).
\newblock {Pigeons.jl: Distributed Sampling From Intractable Distributions}.
\newblock {\em arXiv:2308.09769}.

\bibitem[Surjanovic et~al., 2022]{surjanovic_parallel_2022}
Surjanovic, N., Syed, S., Bouchard-C\^{o}t\'{e}, A., and Campbell, T. (2022).
\newblock Parallel tempering with a variational reference.
\newblock In {\em Advances in Neural Information Processing Systems},
  volume~36, pages 565--577.

\bibitem[Syed et~al., 2022]{syed2022nrpt}
Syed, S., Bouchard-C{\^o}t{\'e}, A., Deligiannidis, G., and Doucet, A. (2022).
\newblock Non-reversible parallel tempering: {A} scalable highly parallel
  {MCMC} scheme.
\newblock {\em Journal of the Royal Statistical Society: Series B},
  84(2):321--350.

\bibitem[Tan, 2017]{tan2017optimally}
Tan, Z. (2017).
\newblock Optimally adjusted mixture sampling and locally weighted histogram
  analysis.
\newblock {\em Journal of Computational and Graphical Statistics},
  26(1):54--65.

\bibitem[{Tarek} et~al., 2020]{tarek2020dynamic}
{Tarek}, M., {Xu}, K., {Trapp}, M., {Ge}, H., and {Ghahramani}, Z. (2020).
\newblock {DynamicPPL}: {S}tan-like speed for dynamic probabilistic models.
\newblock {\em arXiv:2002.02702}.

\bibitem[Tierney, 1998]{tierney1998note}
Tierney, L. (1998).
\newblock {A note on Metropolis-Hastings kernels for general state spaces}.
\newblock {\em The Annals of Applied Probability}, 8(1):1--9.

\bibitem[Turitsyn et~al., 2011]{turitsyn2011irreversible}
Turitsyn, K.~S., Chertkov, M., and Vucelja, M. (2011).
\newblock Irreversible monte carlo algorithms for efficient sampling.
\newblock {\em Physica D: Nonlinear Phenomena}, 240(4-5):410--414.

\bibitem[{van de Meent} et~al., 2018]{vandemeent2021introduction}
{van de Meent}, J.-W., {Paige}, B., {Yang}, H., and {Wood}, F. (2018).
\newblock {An Introduction to Probabilistic Programming}.
\newblock {\em arXiv:1809.10756}.

\bibitem[{van der Vaart}, 1998]{vandervaart1998asymptotic}
{van der Vaart}, A.~W. (1998).
\newblock {\em Asymptotic statistics}.
\newblock Cambridge series in statistical and probabilistic mathematics.
  Cambridge University Press.

\bibitem[Wang et~al., 2021]{wang2021regeneration}
Wang, A.~Q., Pollock, M., Roberts, G.~O., and Steinsaltz, D. (2021).
\newblock {Regeneration-enriched Markov processes with application to Monte
  Carlo}.
\newblock {\em The Annals of Applied Probability}, 31(2):703 -- 735.

\bibitem[Wilkinson, 2005]{wilkinson2005parallel}
Wilkinson, D.~J. (2005).
\newblock Parallel {B}ayesian computation.
\newblock In Kontoghiorghes, E.~J., editor, {\em Handbook of Parallel Computing
  and Statistics}, pages 481--512. Marcel Dekker/CRC Press, New York.

\bibitem[Woodard et~al., 2009]{woodard2009conditions}
Woodard, D.~B., Schmidler, S.~C., and Huber, M. (2009).
\newblock {Conditions for rapid mixing of parallel and simulated tempering on
  multimodal distributions}.
\newblock {\em The Annals of Applied Probability}, 19(2):617--640.

\bibitem[Xie et~al., 2010]{xie2010improving}
Xie, W., Lewis, P.~O., Fan, Y., Kuo, L., and Chen, M.-H. (2010).
\newblock Improving marginal likelihood estimation for {B}ayesian phylogenetic
  model selection.
\newblock {\em Systematic Biology}, 60(2):150--160.

\bibitem[Yao et~al., 2022]{yao2022stacking}
Yao, Y., Vehtari, A., and Gelman, A. (2022).
\newblock Stacking for non-mixing bayesian computations: The curse and blessing
  of multimodal posteriors.
\newblock {\em Journal of Machine Learning Research}, 23(79):1--45.

\end{thebibliography}

\clearpage

\appendix

\section{Implementation details}\label{app:implementation}

\subsection{Efficient and stable computation of normalizing constants}\label{app:norm_constants}

Here we give a brief description of how the stepping-stone approach can be
used to estimate log-normalizing constants.
Consider the following importance sampling identity: for any $i<i'$,
\[
\frac{\mcZ(\beta_{i'})}{\mcZ(\beta_i)} = \EEi{i}\left[e^{-(\beta_{i'}-\beta_i)V}\right].
\]
Building the sample estimate of the above using the samples $\Vs$ (\cref{sec:adaptation}) yields
\[\label{eq:def_IS_estimator}
\widehat{\left(\frac{\mcZ(\beta_{i'})}{\mcZ(\beta_i)}\right)}_{\text{IS}} := 
\frac{1}{S_i}\sum_{n=1}^{S_i} e^{-(\beta_{i'}-\beta_i)V_n^{(i)}},
\]
which is a consistent estimator of the ratio of normalizing constants. 
Combining these using the telescoping identity
\[\label{eq:step_stone_telescopic}
\mcZ(\beta_i) = \prod_{j=1}^i \frac{\mcZ(\beta_j)}{\mcZ(\beta_{j-1})},
\]
where we used $\mcZ(0)=1$, yields the \emph{forward} stepping-stone estimator 
\[
\widehat{\mcZ(\beta_i)}_{\text{SS-F}} := 
\prod_{j=1}^i \widehat{\left(\frac{\mcZ(\beta_j)}{\mcZ(\beta_{j-1})}\right)}_{\text{IS}}.
\]
Similarly, the inverted importance sampling identity
\[
\frac{\mcZ(\beta_{i'})}{\mcZ(\beta_i)} = \EEi{i'}\left[e^{(\beta_{i'}-\beta_i)V}\right]^{-1},
\]
admits the estimator
\[\label{eq:def_IS-I_estimator}
\widehat{\left(\frac{\mcZ(\beta_{i'})}{\mcZ(\beta_i)}\right)}_{\text{IS-I}} := 
\left(\frac{1}{S_{i'}}\sum_{n=1}^{S_{i'}} e^{(\beta_{i'}-\beta_i)V_n^{(i')}} \right)^{-1}.
\]
Using this in \cref{eq:step_stone_telescopic} yields the \emph{backward}
stepping-stone estimator
\[
\widehat{\mcZ(\beta_i)}_{\text{SS-B}} := 
\prod_{j=1}^i \widehat{\left(\frac{\mcZ(\beta_j)}{\mcZ(\beta_{j-1})}\right)}_{\text{IS-I}}.
\]
Finally, to produce an estimator of the log-normalizing constants we take the
simple average of the logarithm of both stepping-stone estimators
\[\label{eq:def_log-norm_const_estimator}
\widehat{\log(\mcZ(\beta_i))} := \frac{1}{2}\left[ \log\left(\widehat{\mcZ(\beta_i)}_{\text{SS-F}}\right) + \log\left(\widehat{\mcZ(\beta_i)}_{\text{SS-B}}\right) \right].
\]

Two insights help in efficiently and accurately computing 
\cref{eq:def_log-norm_const_estimator}. The first is noting that the calculation
can be carried out recursively in one pass over the data $\Vs$, since
\[
\widehat{\log(\mcZ(\beta_i))} &= \widehat{\log(\mcZ(\beta_{i-1}))} +  \frac{1}{2}\left[\log\widehat{\left(\frac{\mcZ(\beta_i)}{\mcZ(\beta_{i-1})}\right)}_{\text{IS}} + \log\widehat{\left(\frac{\mcZ(\beta_i)}{\mcZ(\beta_{i-1})}\right)}_{\text{IS-I}} \right].
\]
The second insight is that the logarithms on the right-hand side above should
be computed using the log-sum-exp (LSE) trick 
\citep[][\S 3.5.3]{murphy2012machine}
\[
\log\widehat{\left(\frac{\mcZ(\beta_i)}{\mcZ(\beta_{i-1})}\right)}_{\text{IS}} &= 
-\log(S_{i-1}) + \LSE(-(\beta_i-\beta_{i-1})V^{(i-1)}) \\
\log\widehat{\left(\frac{\mcZ(\beta_i)}{\mcZ(\beta_{i-1})}\right)}_{\text{IS-I}} &= \log(S_{i}) - \LSE((\beta_i-\beta_{i-1})V^{(i)}).
\]

\subsection{Pseudo-code}\label{app:pseudo-code}

\begin{algorithm}
\caption{One step of NRST with fixed grid $\mcP$ and affinities $\mcC$.}
\begin{algorithmic}
\Function{NRSTStep}{$x,i,\eps$; $\mcP, \mcC$}
\vspace{.3em}
\State \texttt{\# Tempering step}
\State $\iprop \gets i+\eps$ \Comment{Propose move}
\If{$\iprop = N+1$} \Comment{Bounce above}
    \State $(i,\eps) \gets (N, -1)$
\ElsIf{$\iprop=-1$} \Comment{Bounce below}
    \State $(i,\eps) \gets (0, +1)$
\Else \Comment{Interior move}
	\State $A\sim \distBern(\alpha_{i,\iprop}(x))$ \Comment{Use \cref{eq:def_comm_step_acc_prob} with $\mcP$ and $\mcC$}
	\If{$A=1$}
	    \State $i \gets \iprop$
	\Else
	    \State $\epsilon \gets -\epsilon$
	\EndIf
\EndIf
\vspace{.3em}
\State \texttt{\# Exploration step}
\If{$i > 0$}
    \State $x \sim K_i(x,\cdot)$
\Else
    \State $x \sim \pi_0$
\EndIf
\State \Return $(x,i,\eps)$
\EndFunction
\end{algorithmic}
\label{alg:NRST_step}
\end{algorithm}

\begin{algorithm}
\caption{Running an NRST tour}
\begin{algorithmic}
\Function{NRSTTour}{$\mcP, \mcC$}
\vspace{.3em}
\State \texttt{\# Initialize from $\nu$}
\State $x_0\sim \pi_0$
\State $(x,i,\eps) \gets (x_0, 0, +1)$
\State $\mcT \gets \{(x,i,\eps)\}$ \Comment{Initialize trace}
\vspace{.3em}
\State \texttt{\# Run NRST until $\atom$ is reached}
\While {$i > 0$ or $\eps=+1$}
	\State $x,i,\eps \gets$ \textproc{NRSTStep}($x,i,\eps$; $\mcP, \mcC$) \Comment{\cref{alg:NRST_step}}
	\State $\mcT \gets \mcT \cup \{(x,i,\eps)\}$ \Comment{Update trace}
\EndWhile
\State \Return $\mcT$
\EndFunction
\end{algorithmic}
\label{alg:NRST_tour}
\end{algorithm}

\begin{algorithm}
\caption{Parallelized Regenerative NRST}
\begin{algorithmic}
\Function{ParallelNRST}{$\mcP, \mcC, \alpha, \delta, \widehat{\TE}$}
\State $K \gets  K_{\text{min}}(\alpha,\delta,\TEest)$ \Comment{\cref{eq:def_Kmin}}
\State $\mcS \gets \emptyset$ \Comment{Initialize storage}
\PFor {$k \in \{1,\dots,K\}$}
	\State $\mcT_k \gets $  \textproc{NRSTTour}($\mcP, \mcC$) \Comment{\cref{alg:NRST_tour}}
	\State $\mcS \gets \mcS \cup \{\mcT_k\}$ \Comment{Store tour trace}
\EndPFor
\State \Return $\mcS$
\EndFunction
\end{algorithmic}
\label{alg:NRST_par}
\end{algorithm}

\section{Additional experimental details}\label{app:additional_exps}

\subsection{Description of the models}\label{app:models_description}

In the following we use the convention that $\distNorm(\mu,\sigma^2)$ is a Gaussian 
distribution with mean $\mu$ and variance $\sigma^2$.

\paragraph{Toy multivariate Gaussian}
This is a toy $d$-dimensional Bayesian model specified by
\[
x &\sim \distNorm_d(\boldsymbol{0}, \sigma_0^2 I) \\
y|x &\sim \distNorm_d(x,I).
\]
By conjugacy, the posterior distribution for $x$ given $y$ is a multivariate 
Gaussian. In fact, the whole path of distributions $\beta\mapsto\pibeta$ is given by a 
family of Gaussians. To simplify matters, assume that $y=m\boldsymbol{1}$ for some 
$m\in\reals$. Then, for all $\beta\in[0,1]$,
\[
\pibeta = \distNorm_d\left(\mu(\beta)\boldsymbol{1},\sigma(\beta)^2 I\right),
\]
with
\[
\mu(\beta) := \beta m \sigma(\beta)^2, &&
\sigma(\beta)^2 := \left(\beta + \frac{1}{\sigma_0^2}\right)^{-1}.
\]
For \cref{fig:index_process} we use the parameters $d=3$, $m=2$, and $\sigma_0=2$.

\paragraph{Banana}
A $2$-dimensional target density in the shape of a banana, as described in Equation $1$ of
\citet{pagani2022banana}
\[
x_1 &\sim \distNorm(1,10) \\
x_2|x_1 &\sim \distNorm(x_1^2, 0.1^2).
\]
Using the fact that $\EE[x_2^2]=11$ under this target, we set the reference to be
\[
x_1 &\sim \distNorm(1, 10) \\
x_2 &\sim \distNorm(11, 10^2).
\]

\paragraph{Funnel}
A $20$-dimensional target density based on Neal's funnel \citep{neal2003slice}, using the 
parametrization in \citet{modi2021delayed}
\[
x_1 &\sim \distNorm(0,3^2) \\
x_i|x_1 &\sim \distNorm(0,e^{x_1}), \quad i\in\{2,\dots,20\}.
\]
We set the reference distribution to an isotropic Gaussian
\[
x_i &\distiid \distNorm(0,3^2), \quad i\in\{1,\dots,20\}.
\]

\paragraph{Hierarchical Model}
A slightly modified version of the hierarchical model described in 
\citet[][\S 6.3]{yao2022stacking}.
It is different only in that our model puts a Cauchy prior on $\mu$ instead of 
the (improper) uniform density on $\reals$. This allows us to obtain a proper 
prior distribution that can be used as reference
\[
\mu &\sim \distCauchy \\
\tau^2,\sigma^2 &\distiid \mathsf{InverseGamma}(0.1,0.1) \\
\theta_j|\mu,\tau^2 &\sim \distNorm(\mu,\tau^2), \quad j\in\{1,\dots,J\} \\
Y_{i,j}|\theta_j,\sigma^2 &\sim \distNorm(\theta_j,\sigma^2), \quad i\in\{1,\dots,M\}, j\in\{1,\dots,J\}.
\]
We generate a dataset $\mathbf{Y}$ with $J=8$ and $M=20$ by performing rejection 
sampling from this model, constraining the output to have a between-group 
variance $16$ times larger than the within-group variance.

\paragraph{mRNA Transfection}
This is a time series model taken from \citet{ballnus2017comprehensive}
\[
\log_{10}(t_0) &\sim \distUnif(-2,1) \\
\log_{10}(\kappa),\log_{10}(\beta),\log_{10}(\delta) &\distiid \distUnif(-5,5) \\
\log_{10}(\sigma) &\sim \distUnif(-2,5) \\
y_i|t_0,\kappa,\beta,\delta,\sigma &\sim \distNorm(\mu(t_i|t_0,\beta,\delta), \sigma^2), \quad i\in\{1,\dots,n_\text{obs}\},
\]
where
\[
\mu(t_i|t_0,\beta,\delta) := \frac{\kappa}{\delta-\beta}[\exp(-\beta(t_i-t_0)) - \exp(-\delta(t_i-t_0))].
\]
We use the dataset provided in \citet{surjanovic_parallel_2022}.

\paragraph{Threshold Weibull}
This is a $3$-parameter Weibull distribution that adds a threshold to the usual $2$-parameter
family. This creates a likelihood function that is unbounded from above, which usually
complicates standard frequentist inference. The likelihood model is taken from 
\citet[][Ch. 8]{cheng2017inference}, and we use the dataset ``SteenStickler0'' provided 
there. Finally, we put an uninformative prior to obtain a fully specified Bayesian model
\[
a &\sim \distUnif(0,200) \\
b &\sim \mathsf{InverseGamma}(0.1,0.1) \\
c &\sim \distUnif(0.1,10) \\
y_i|a,b,c &\sim \mathsf{Weibull}(a,b,c).
\]
where the $3$-parameter Weibull density is
\[
p(y|a,b,c) = \ind\{y>a\} \frac{c}{b}\left(\frac{y-a}{b}\right)^{c-1}\exp\left[-\left(\frac{y-a}{b}\right)^c \right].
\]

\paragraph{XY Model}
This is a classical $2$-dimensional lattice model of statistical mechanics. 
We focus on the 
nearest neighbor case with periodic boundaries, using a constant interaction 
$J=2$ and no external field. Thus, the potential function can be written as
\[
\forall x\in\reals^{n^2}:\ V(x) := -J\sum_{(i,j)\in E} \cos(x_i-x_j),
\]
where $E$ is the edge set of the lattice. In our simulations we use a square 
lattice with sides of length $n=8$. As reference, we use the uniform 
distribution on $[-\pi,\pi)^{n^2}$ (i.e., the microcanonical ensemble).
This model is interesting because it induces a 
phase transition in the path of distributions $\beta\mapsto\pibeta$
\citep[see e.g.][]{hasenbusch2005xymodel}.

\subsection{Experimentally choosing hyperparameters}\label{app:hyperparams}

Let $\mcM$ be the collection of models, and furthermore, let $\mcS$ be 
a finite set of pseudo-random number generator (PRNG) seeds. 
To select a \emph{good} 
default configuration, we first define a fixed sample quality 
$(\alpha=0.95,\delta=0.5)$ for \cref{alg:NRST_par}. Then, for each model 
$m\in\mcM$ we find the highest cost across seeds produced by every mix of 
hyper-parameters
\[\label{eq:def_worst_case_cost}
\text{WorstCost}_m(c', \gamma,\bkappa) := \max_{s\in\mcS} \text{Cost}(\text{model} = m, \text{seed} = s; c', \gamma,\bkappa),
\]
where the right-hand side corresponds to a cost measure of running 
\cref{alg:NRST_par} with the specified settings and sample quality controlled by
the pair $(\alpha,\delta)$. Also, let
\[
\forall m\in\mcM:\ (c'_m,\gamma_m,\bkappa_m) = \arg\min_{c',\gamma,\bkappa} \text{WorstCost}_m(c', \gamma,\bkappa)
\]
be the model-specific worst-case-optimal configuration. Then, we choose as 
default hyper-parameters the triplet that minimizes the highest cost increase 
across models relative to their own optimal configuration
\[
(c'_\star,\gamma_\star,\bkappa_\star) = \arg\min_{c',\gamma,\bkappa}\max_{m\in\mcM} \text{WorstCost}_m(c', \gamma,\bkappa) - \text{WorstCost}_m(c'_m,\gamma_m,\bkappa_m). 
\]

A visualization of the behavior of each model using every combination of the 
hyper-parameters is available in \cref{fig:hyperparams_all} at the end of this 
section. In the following, we give a high-level description of the findings.

\paragraph{Affinity slope function}

\begin{figure}[t]
\centering
\includegraphics[width=\plotwidth]{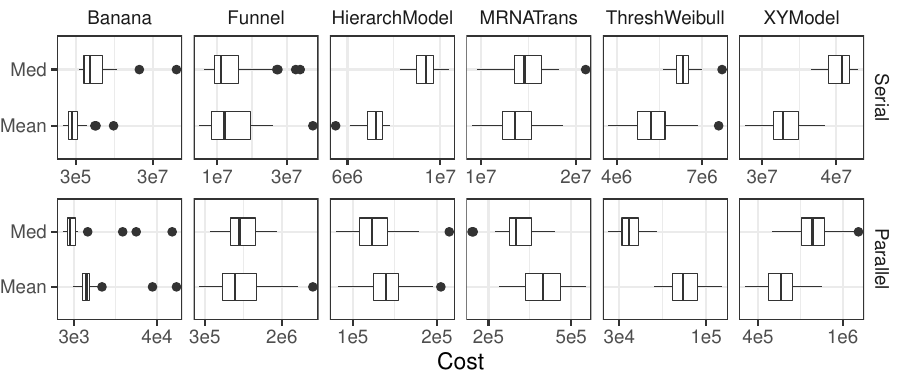}
\caption{Cost of running \cref{alg:NRST_par} versus the affinity tuning strategy, 
for each of $6$ selected models. Rows correspond to the two distinct elapsed-time 
metrics defined in \cref{sec:eval_ST_adaptation}. 
Each experiment is replicated $30$ times. 
The number of tours is determined in a quality-consistent approach 
using \cref{eq:def_Kmin} with  $\alpha=0.95$, $\delta=0.5$, and $\TEest$ 
estimated in a preliminary run.}
\label{fig:hyperparams_fun}
\end{figure}

\cref{fig:hyperparams_fun} shows the result of setting the slope function to
either of the two possible options---mean or median---while the other parameters
are held fixed at their default values. The experiments is run on $6$ different 
models and replicated for $30$ seeds. Additionally, the results are 
evaluated using the two varieties of elapsed-time discussed in 
\cref{sec:eval_ST_adaptation}: serial (equivalent to total evaluations of $V$) and
fully parallel (considering only the maximum number of evaluations in a tour). 

The most relevant aspect of \cref{fig:hyperparams_fun} is the reversal
of preferences depending on the cost measure used. Indeed, under the serial run time lens, the box-plots for the mean strategy are in most cases well below
the ones for the median. The worst-case cost reductions are in the order of 
$30\%$, except for the Banana model where it is orders of magnitude lower.
On the other hand, when the experiments are evaluated in terms of the running time
of a fully parallelized setting, the median strategy generally takes the 
lead. The largest reduction happens in the ``ThresholdWeibull'' model, where the 
median approach has one third of the parallel cost of the mean.

In order to keep the complexity of the rest of the paper manageable, in the
following we focus exclusively on the serial cost measure. 
Since this is an upper bound on the parallelized elapsed-time (for any number
of available workers), minimizing serial time still helps in controlling 
parallel cost. And since the mean-energy strategy performs better in the serial 
cost metric, we set it as default for the rest of this section.

\paragraph{Autocorrelation limit}

\begin{figure}[t]
\centering
\includegraphics[width=\plotwidth]{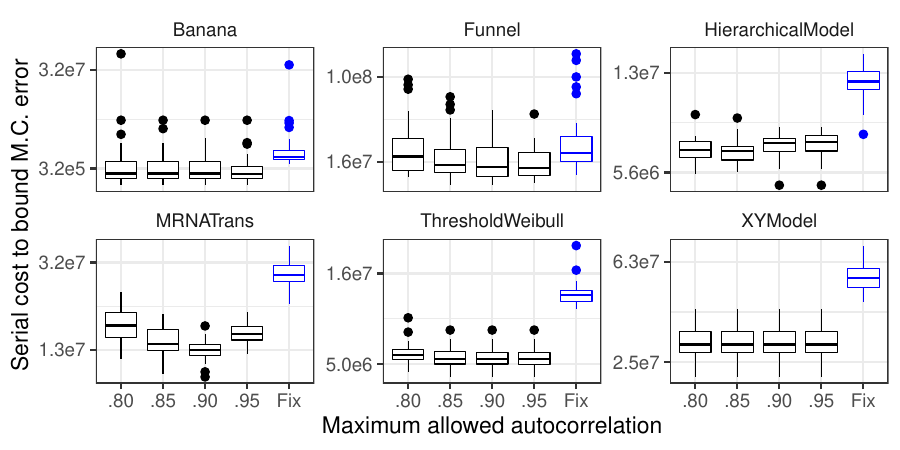}
\caption{Serial cost (as defined in \cref{sec:eval_ST_adaptation}) versus the
maximum allowed correlation parameter $\bkappa$, for each of $6$ selected models.
Each experiment is replicated $30$ times. The number of tours is
determined in a quality-consistent approach using \cref{eq:def_Kmin} with 
$\alpha=0.95$, $\delta=0.5$, and $\TEest$ estimated in a preliminary run.
``\textsf{Fix}'' means that exploration steps were not tuned but simply fixed to 
$3$ across all levels.}
\label{fig:hyperparams_cor}
\end{figure}

\cref{fig:hyperparams_cor} shows the summary of the experiments where $\bkappa$ is
varied, while leaving $\gamma=2$ fixed and using the mean strategy. We also show
a control case (``\textsf{Fix}'') where exploration steps are set to $3$ in all
levels regardless of autocorrelation. We chose this value because it is close to 
the average number exploration step obtained via autocorrelation tuning. 
In most cases, fixing is significantly more costly (around $2$ times) than tuning 
the number of exploration steps. The reason is that, in general, most 
intermediate distributions $\pibeta$ are easy to explore, so that even $1$ 
exploration step is sometimes enough to achieve a low-correlated sample $V(x)$.
Only few intermediate distributions are truly complex and thus require many more
than $3$ steps to achieve mixing in $V(x)$. Autocorrelation-based tuning can exploit this, concentrating efforts where they are most necessary.

The second thing to notice in \cref{fig:hyperparams_cor} is that varying $\bkappa$
does not generally produce significantly different cost distributions. The
paradigmatic example is the ``XYModel'', where using $1$ exploration step already gives auto-correlations lower than $0.8$. This insensitivity to $\bkappa$ means 
that the algorithm is robust to this choice.

\paragraph{Grid size correction}

\begin{figure}[t]
\centering
\includegraphics[width=\plotwidth]{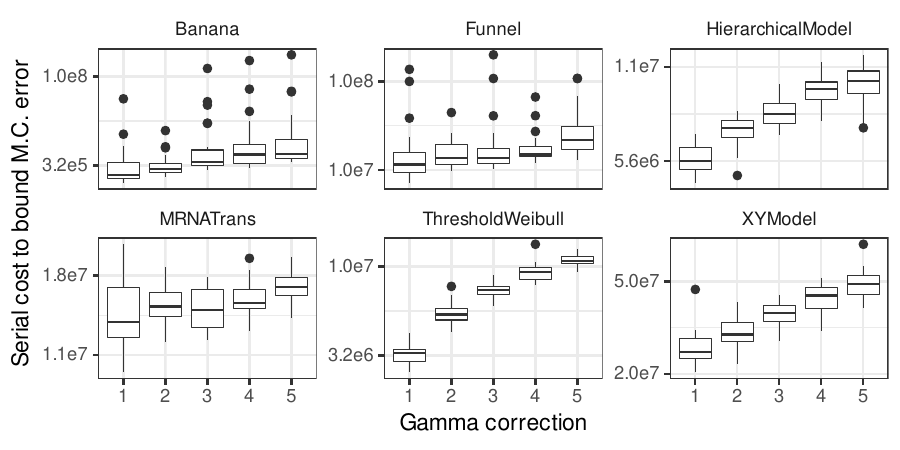}
\caption{Serial cost (as defined in \cref{sec:eval_ST_adaptation}) versus the
correction parameter $\gamma$ in \cref{sec:tune_N}, for each of $6$
selected models. Each experiment is replicated $30$ times. The number of tours is
determined in a quality-consistent approach using \cref{eq:def_Kmin} with 
$\alpha=0.95$, $\delta=0.5$, and $\TEest$ estimated in a preliminary run.}
\label{fig:hyperparams_gam}
\end{figure}

In \cref{sec:adaptation}, we derived a formula for the grid size $N$ that 
minimizes the expected number of steps taken in a run where the number of tours is
determined via $\Kmin(\alpha,\delta,\TE)$ (\cref{eq:def_Kmin}). This measure 
should be a good proxy of the serial running time when the exploration effort at 
each level is roughly the same.

\cref{fig:hyperparams_gam} shows the serial cost for different values of 
$\gamma$, while fixing $\bkappa=0.95$ and using the mean strategy.
We can see that the median cost does decreases towards $\gamma=1$, which supports
the derivations from \cref{sec:adaptation}. However, the worst-case cost does not
follow this pattern in at least half the models. The most likely reason is that
the quality of the tuning process for the grid and affinities deteriorates as $N$
decreases. Indeed, the interpolation used to set grid points and the estimators
used in setting affinities worsen as the grid becomes coarser. Thus, though
adaptation still works in most cases at $\gamma=1$, it becomes brittle and 
susceptible to yielding high-cost configurations. For this reason, we set 
$\gamma=2$ as the default value.

\begin{figure}[t]
\centering
\includegraphics[width=\plotwidth]{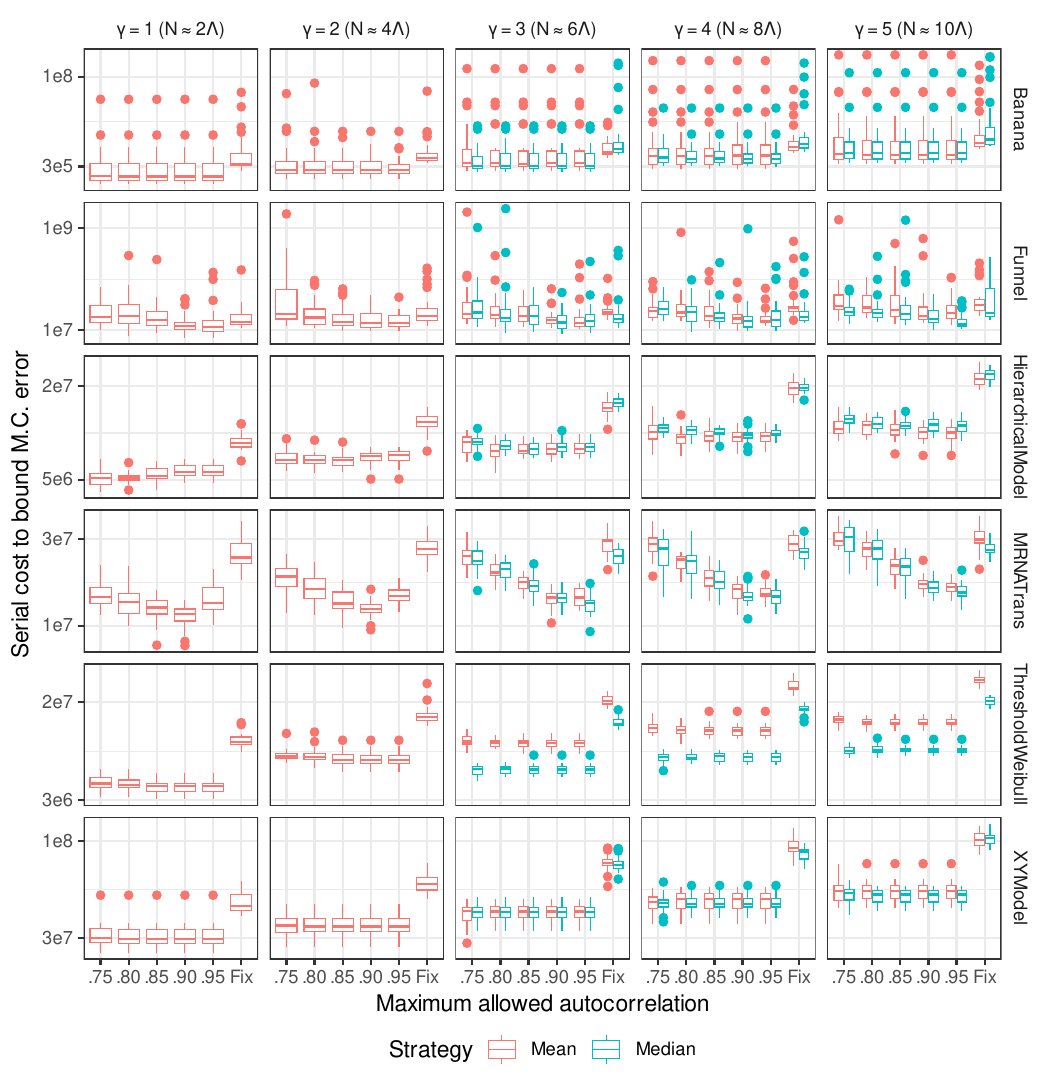}
\caption{Serial cost (as defined in \cref{sec:eval_ST_adaptation}) by affinity tuning strategy,
maximum allowed correlation parameter $\bkappa$, and grid size parameter $\gamma$,
for each of $6$ selected models. Each combination is replicated $30$ times.
The number of tours is determined in a quality-consistent approach using \cref{eq:def_Kmin} with 
$\alpha=0.95$, $\delta=0.5$, and $\TEest$ estimated in a preliminary run.
``\textsf{Fix}'' means that exploration steps were not tuned but simply fixed to 
$3$ across all levels.
}
\label{fig:hyperparams_all}
\end{figure}

\subsection{Tuning of competing ST algorithms}\label{app:tuning_details}

In order to quantify the effect of our tuning strategy, for the experiments in
\cref{sec:eval_ST_adaptation} we apply to all samplers both the approach outlined in 
\cref{sec:adaptation} and the tuning schemes recommended by their own authors. 

In particular, SH16 uses an uniformly-spaced grid, and estimates $\log(\mcZ(\beta))$
on-the-fly---i.e., running simulated tempering sequentially with arbitrary 
affinities and adjusting them at each step by approximating the thermodynamic 
identity using the trace of $V(x_n)$---as described in \citep{sakai2016learning}.
FBDR also employs a uniformly-spaced grid, but it uses a simpler approach for 
setting the affinities described in \citet{park2007choosing}. The idea is to 
collect $V$ samples at each level by running the explorers independently. The
log normalizing constants are then calculated using the thermodynamic identity.
Finally, \citet{geyer1995annealing} propose a more complex tuning method that 
iteratively adapts both the grid and the affinities. The former is adjusted to
promote equi-rejection, targeting an average rejection rate between $20\%$
and $40\%$. The log-normalizing constants, in turn, are adjusted using a 
stochastic approximation method. Sadly, this tuning strategy contains 
model-sensitive hyper-parameters and is not described in enough detail in
\citet{geyer1995annealing}. Hence, it was impossible for us to replicate 
robustly, and we opted for tuning GT95 using the parameters obtained by FBDR.

\section{Proofs}\label{app:proofs}

We begin by stating some basic assumptions, summarizing the notation used so far, and 
also introducing additional concepts necessary for this section. We assume that 
the base state space $\statespace$ is Polish, and let $\sigalg$ be its associated
Borel sigma algebra such that $\mespace$ is a measurable space. As defined in
\cref{sec:methodology}, for all $\beta\in[0,1]$ we let $\EEi{\beta}$ be the
expectation operator associated to the tempered distribution $\pibeta$ on $\mespace$.

The fact that $\statespace$ is Polish implies that $\prodss$---the product space 
for NRST defined in \cref{eq:def_prodss}---is also
Polish. Let $\prodsigalg$ be its Borel sigma algebra---which is identical to the
product sigma algebra---so that $\prodmespace$ is a measurable space. 
Recall that the lifted distribution $\lifted$ is defined on $\prodmespace$, and that
we use $\lifted(f)$ to denote the expectation under $\lifted$ of a measurable function
$f:\prodss\to\reals$.

Now, the kernel $\knrst$, together with an initial distribution $\mu$ on 
$\prodmespace$, define the NRST
Markov chain $\tdxpath=\{\tdx_n\}_{n=0}^\infty$ on the path space $\pathmespace$.
Let $\Pr_{\mu}$ denote the law of $\tdxpath$ and $\EE_{\mu}$ its corresponding 
expectation operator. In particular, as we did in \cref{sec:methodology}, we write
$\Pr_\nu$ for the law of the chain initialized from the renewal measure $\nu$ 
(\cref{eq:def_nu_atom_measure}) associated to the NRST atom $\atom$. When 
$\mu=\delta_{\tdx}$ for some $\tdx\in\prodss$, we simply 
write $\Pr_{\tdx}$ and $\EE_{\tdx}$.

We use $\{\mcF_n\}_{n=0}^\infty$ to denote the natural filtration of the Markov 
chain. For $j\in\nats$, let $\theta_j:\prodss^\infty\to\prodss^\infty$ denote the shift
operator defined by $\theta_j(\{\tdx_n\}_{n=0}^\infty)=\{\tdx_{j+n}\}_{n=0}^\infty$. 
For any $S\in\prodsigalg$, define the stopping times
\[\label{eq:def_first_hit_return_times}
T_{S} &:= \inf\{n\geq0: \tdx_n\in S\} \\
\sigma_{S} &:= \inf\{n\geq1: \tdx_n\in S\},
\]
known, respectively, as the first-hitting time and the first-return time of the 
set $S$. Furthermore, recursively define the sequences of
return times $\{T_S^{(k)}\}_{k\in\nats}$ and hitting times 
$\{\sigma_S^{(k)}\}_{k\in\nats}$ by setting $T_S^{(1)}=T_S$, 
$\sigma_S^{(1)}=\sigma_S$, and for all $k\in\nats$,
\[\label{eq:def_further_hit_return_times}
T_{S}^{(k+1)} &:= \inf\{n\geq T_S^{(k)}: \tdx_n\in S\} \\
\sigma_{S}^{(k+1)} &:= \inf\{n\geq\sigma_S^{(k)}: \tdx_n\in S\}.
\]
In particular, when $S=\atom$, we get that $T_\atom^{(k)}=T_k$ as defined in 
\cref{sec:NRST}.

Finally, in the following we assume that for all $i\in\indset$, the exploration 
kernel $K_i$ is $\pi_i$-irreducible in addition to $\pi_i$-invariant.

\subsection{Proofs for \cref{sec:NRST}}

\begin{proposition}\label{prop:lifted_invariance}
The Markov kernel $\knrst$ is $\lifted$-invariant.
\end{proposition}
\begin{proof}
Since the composition of $\lifted$-invariant kernels is $\lifted$-invariant, it 
suffices to show that both $\kexp$ and $\ktem$ satisfy this property. For the 
former, note that for all $A\in\sigalg$, $i\in\indset$ and $\eps \in \epsset$
\[
&[\lifted\kexp](A,i,\epsilon) \\
&=\frac{1}{2}\sum_{\epsilon'\in\{-1,+1\}} \sum_{j=0}^N p_j \int \pi_j(\dee x) K_j(x,A) \ind\{j=i\}\ind\{\epsilon'=\epsilon\} \\
&=\frac{1}{2}p_i\int  \pi_i(\dee x) K_i(x,A) \\
&=\lifted(A,i,\epsilon),
\]
due to $K_i$ being $\pi_i$-invariant.

The tempering step can be understood as the composition of two Metropolis 
steps. Since Metropolis kernels are $\tilde{\pi}$-reversible, they are
$\tilde{\pi}$-invariant. This shows that the tempering kernel is itself
invariant, since it is the composition of invariant kernels.
It remains to derive their acceptance probabilities. 
In the first step we deterministically propose
\[
(\xprop,\iprop, \epsprop) = (x,i+\epsilon,-\epsilon).
\]
Note that this move is an \emph{involution} \citep[][\S 2]{tierney1998note}; i.e., 
it is its own inverse, since applying it twice gives
\[
(i,\epsilon) \mapsto (i+\epsilon,-\epsilon) \mapsto (i+\epsilon-\epsilon,-(-\epsilon)) = (i,\epsilon).
\]
Then, the Metropolis ratio for the first step becomes
\[
\frac{\lifted(x,i+\epsilon,-\epsilon)}{\lifted(x,i,\epsilon)} = \frac{\pi_{i+\epsilon}(x)}{\pi_i(x)} \frac{p_{i+\epsilon}}{p_i}.
\]
The second move aims to correct the reversal of the velocity that occurs
when the tempering step is accepted, which would prevent the sampler from having 
momentum. To achieve this, we fix $x$ and $i$, and use the deterministic involutive 
proposal $\epsprop = -\epsilon$. By definition of $\lifted$, this proposal is always 
accepted, since
\[
\frac{\lifted(x,i,-\epsilon)}{\lifted(x,i,\epsilon)} = \frac{\frac{1}{2}\pi_i(x)p_i}{\frac{1}{2}\pi_i(x)p_i} = 1.
\]
It follows that the probability of jointly accepting both sub-steps---and thus the 
tempering proposal---is precisely the quantity in \cref{eq:def_comm_step_acc_prob}.
\end{proof}

\begin{lemma}\label{lemma:atom_recurrent_accessible}
\hfill
\benum
\item $\atom$ is recurrent.
\item $\atom$ is accessible; i.e., $\Pr_{\tdx_0}(\sigma_\atom<\infty)>0$ for $\lifted$-a.e.\ $\tdx_0$.
\eenum
\end{lemma}
\begin{proof}
\hfill\newline
\noindent 1. Since $\lifted(\atom)=p_0/2>0$, $\atom$ is recurrent by Prop. 6.2.8(i) in
\citet{douc2018markov}.

\noindent 2. Let $\tdx_0=(x_0,i_0,\eps_0)\in\prodss$. Note that if 
$\tdx_0\in\atom$, we know that $\Pr_{\tdx_0}(\tdx_1\notin\atom)=1$, so by the Markov property
\[
\forall\tdx_0\in\atom:\ \Pr_{\tdx_0}(\sigma_\atom<\infty) = \EE_{\tdx_0}[\Pr_{\tdx_1}(\sigma_\atom<\infty)].
\]
The right-hand side is positive if accessibility holds for $\tdx_0\notin\atom$. 
Let us focus on this case. In particular, when $\eps_0=-1$,
\[
\Pr_{\tdx_0}(\sigma_\atom<\infty) &= \Pr_{\tdx_0}\left(\bigcup_{n\in\nats} \{\tdx_n \in \atom\} \right) \geq \Pr_{\tdx_0}(\tdx_{i_0} \in \atom).
\]
The bound on the right-hand side is the probability of a perfect
return to the atom in exactly $i_0$ steps, which is positive because
\[
\Pr_{\tdx_0}(\tdx_{i_0} \in \atom) &= \Pr_{\tdx_0}(i_1=i_0-1,i_2=i_0-2,\dots,i_{i_0}=0) \\
&= \acc{i_0,i_0-1}(x_0) \prod_{j=1}^{i_0}  \int K_{i_0-j}(x_{j-1},\dee x_j) \acc{i_0-j,i_0-j-1}(x_j) \\
&>0.
\]
The inequality follows from 1) the fact that for all $i\in\indsetone$, 
$\acc{i,i-1}(x)>0$ $\pi$-a.s.\---because $V(x)<\infty$ $\pi$-a.s.\---and 2) 
the assumption that the exploration kernels $\{K_i\}_{i=1}^N$ are 
$\pi_i$-invariant and $\pi_i$-irreducible. 

Finally, to tackle the remaining case $\eps_0=+1$, note that the set 
\[
\mcN:=\statespace\times\indset\times\{-1\}
\]
has $\Pr_{\tdx_0}(T_{\mcN}<\infty)=1$ due to the boundary at $i=N$. Hence
\[
\Pr_{\tdx_0}(\sigma_\atom<\infty) = \EE_{\tdx_0}[\Pr_{\tdx_{T_{\mcN}}}(\sigma_\atom<\infty)] > 0,
\]
because $\tdx_{T_{\mcN}}=(x_{T_{\mcN}},i_{T_{\mcN}},-1)$ and thus the bound 
follows from the case $\eps_0=-1$.
\end{proof}

\cref{lemma:kth_visit_return,lemma:translation_Douc_regeneration} below have the
sole purpose of translating results between two conventions: from 
the atomic Markov chain theory---where tours are started at the atom---to the 
regenerative simulation approach---where tours begin with a regeneration. With 
these tools, we can then simply leverage existing results for atomic Markov chains
in \cref{lemma:tour_sums_iid} and \cref{thm:LLN}.

\begin{lemma}\label{lemma:kth_visit_return}
For all $\mcS\in\prodsigalg$, $k\in\nats$, and $\omega\in\pathspace$: $\sigma_\mcS^{(k)}(\omega) = 1 + [T_\mcS^{(k)}\circ\theta_1](\omega)$.
\end{lemma}
\begin{proof}
We make the dependence on $\omega$ implicit to avoid cluttering the notation. 
Let us proceed by induction. First, note that
\[
\sigma_\mcS &= \inf\{n>0: \tdx_n\in\mcS\} = 1+\inf\{n\geq0: \tdx_{n+1}\in\mcS\} = 1+T_1\circ\theta_1.
\]
Now, assume $\sigma_\mcS^{(k-1)} = 1+T_\mcS^{(k-1)}\circ\theta_1$. Then
\[
\sigma_\mcS^{(k)} &= \inf\{n>\sigma_\mcS^{(k-1)}: \tdx_n\in\mcS\} \\
&= \inf\{n>1+T_\mcS^{(k-1)}\circ\theta_1: \tdx_n\in\mcS\} \\
&= 1+\inf\{n>T_\mcS^{(k-1)}\circ\theta_1: \tdx_n\circ\theta_1\in\mcS\} \\
&= 1+\inf\{n>T_\mcS^{(k-1)}: \tdx_{n}\in\mcS\}\circ\theta_1 \\
&= 1+T_\mcS^{(k)}\circ\theta_1,
\]
as needed.
\end{proof}

\begin{lemma}\label{lemma:translation_Douc_regeneration}
Let $f:\statespace\to\reals$ be a $\prodsigalg$-measurable function. Define
\[\label{eq:def_index1_toursums}
z_1(f) := \sum_{n=1}^{\sigma_\atom} f(\tdx_n), \qquad \forall k\geq 2:z_k(f) &:= \sum_{n=\sigma_\atom^{(k-1)}+1}^{\sigma_\atom^{(k)}} f(\tdx_n).
\]
Then,
\benum
\item for all $k\in\nats$ and $\omega\in\pathspace$, $[z_k(f)](\omega) = [s_k(f)\circ\theta_1](\omega)$, and
\item under $\Pr_{\tdx}$ for $\tdx\in\atom$, the real-valued stochastic process 
$\{z_k(f)\}_{k\in\nats}$ has the same finite dimensional distributions as 
$\{s_k(f)\}_{k\in\nats}$ under $\Pr_\nu$,
\eenum
where $\{s_k(f)\}_{k\in\nats}$ are the tour sums of $f$ (see \cref{sec:NRST}).
\end{lemma}
\begin{proof}
\hfill\newline
\noindent1. As in previous results, we make the dependence on $\omega$ implicit 
to avoid cluttering the notation. Using \cref{lemma:kth_visit_return}, we see that
\[
\sum_{n=1}^{\sigma_\atom} f(\tdx_n) = \sum_{n=1}^{1+T_1\circ\theta_1} f(\tdx_n) = \sum_{n=0}^{T_1\circ\theta_1} f(\tdx_n\circ\theta_1) = s_1(f)\circ\theta_1.
\]
Similarly, for $k\geq 2$,
\[
\sum_{n=\sigma_\atom^{(k-1)}+1}^{\sigma_\atom^{(k)}} f(\tdx_n) = \sum_{n=(1+T_{k-1}\circ\theta_1)+1}^{1+T_k\circ\theta_1} f(\tdx_n) = \sum_{n=T_{k-1}\circ\theta_1+1}^{T_k\circ\theta_1} f(\tdx_n\circ\theta_1) = s_k(f)\circ\theta_1.
\]

\noindent 2. It suffices to show that for all $\tdx\in\atom$, $K\in\nats$, and any 
collection of Borel sets $\{A_k\}_{k=1}^K$,
\[
\Pr_{\tdx}\left(\bigcap_{k=1}^K z_k(f) \in A_k \right) = \Pr_\nu\left(\bigcap_{k=1}^K s_k(f) \in A_k \right).
\]
Indeed, using part 1., we have by the Markov property
\[
\Pr_{\tdx}\left(\bigcap_{k=1}^K s_k(f)\circ\theta_1\in A_k \right) &= \EE_{\tdx}\left[\EE_{\tdx}\left[\left(\prod_{k=1}^K \ind\{s_k(f)\in A_k\} \right) \circ \theta_1 \middle| \mcF_1 \right] \right] \\
&=\EE_{\tdx}\left[\EE_{\tdx_1}\left[\prod_{k=1}^K \ind\{s_k(f)\in A_k\} \right] \right] \\
&=\Pr_\nu\left(\bigcap_{k=1}^K s_k(f) \in A_k \right).
\]
\end{proof}

\begin{lemma}\label{lemma:tour_sums_iid}
For any $\prodsigalg$-measurable function $f$, its tour sums 
$\{s_k(f)\}_{k\in\nats}$---as defined in \cref{sec:NRST}---are an \iid collection
under $\Pr_\nu$.
\end{lemma}

\begin{proof}
Since $\atom$ is recurrent by \cref{lemma:atom_recurrent_accessible},
Corollary 6.5.2 in \citet{douc2018markov} shows that the random variables
$\{z_k(f)\}_{k\in\nats}$ defined in \cref{eq:def_index1_toursums}
are \iid under $\Pr_{\tdx}$ for all $\tdx\in\atom$. This means that for any 
$K$ and any collection of Borel sets $\{A_k\}_{k=1}^K$
\[\label{eq:zk_iid_cond}
\forall \tdx\in\atom:\ \Pr_{\tdx}\left(\bigcap_{k=1}^K z_k(f)\in A_k \right) = \prod_{k=1}^K \Pr_{\tdx}(z_1(f)\in A_k).
\]
The above, together with \cref{lemma:translation_Douc_regeneration}, imply that
\[
\Pr_\nu\left(\bigcap_{k=1}^K s_k(f) \in A_k \right)=\Pr_\nu(s_1(f)\in A_k).
\]
We conclude that $\{s_k(f)\}_{n\in\nats}$ is \iid under $\Pr_\nu$.
\end{proof}

\begin{proof}[\textbf{Proof of \cref{thm:LLN}}]

By \cref{lemma:atom_recurrent_accessible}, $\atom$ is recurrent and accessible.
Therefore, Theorem 6.4.2(iii) in \citet{douc2018markov} shows that $\lifted$ is
proportional to the finite measure $\lambda_\atom$, defined by its action on
measurable functions $f:\prodss\to\reals$
\[
\lambda_\atom(f) := \EE_{\tdx}[z_1(f)],
\]
where $\tdx\in\atom$, and $\{z_k(f)\}_{k\in\nats}$ are defined in \cref{eq:def_index1_toursums}. By 
\cref{lemma:translation_Douc_regeneration}, the distribution of $z_1(f)$ under
$\Pr_{\tdx}$ when $\tdx\in\atom$ is the same as the distribution of $s_1(f)$ under
$\Pr_\nu$. Hence,
\[
\lambda_\atom(f) = \EE_\nu[s_1(f)].
\]
Its normalizing constant can be recovered using $f\equiv1$, so that
\[
\lifted(f) = \frac{\lambda_\atom(f)}{\lambda_\atom(1)} = \frac{\EE_\nu[s_1(f)]}{\EE_\nu[\tau_1]}.
\]
Thus,
\[
\EE_\nu[s_1(f)] = \lifted[f]\EE_\nu\left[\tau_1\right].
\]
It follows that $\{s_k\}_{k\in\nats}$ is a collection of \iid random variables
(\cref{lemma:tour_sums_iid}) with common expectation given by the expression above.
In particular, using $f=\ind_\atom$, and noting that $s_1(f)=1$ $\Pr_\nu$-a.s.,
gives
\[
\EE_\nu\left[\tau\right] = \frac{1}{\lifted(\atom)} = \frac{2}{p_0}.
\]
Finally, the strong law of large numbers for \iid random variables shows that
\[
\lim_{k\to\infty} \frac{\sum_{n=0}^{T_k} f(\tdx_n)}{\sum_{n=0}^{T_k} g(\tdx_n)} = \lim_{k\to\infty} \frac{\frac{1}{k}\sum_{j=1}^k s_j(f)}{\frac{1}{k}\sum_{j=1}^k s_j(g)} = \frac{\EE_\nu[s(f)]}{\EE_\nu[s(g)]} = \frac{\lifted[f]}{\lifted[g]}, \;\; \Pr_{\nu}\text{-a.s.}
\]

\end{proof}

\begin{proof}[\textbf{Proof of \cref{thm:CLT}}]
The proof draws from the informal exposition in
\citet[][Ch. 6.4]{ripley1987stochastic}. To de-clutter the notation, let
\[
\sigma_d^2 := \EE_\nu[s(\ind\{i=N\}(h(x)-\target(h)))^2],
\]
and
\[
\hslash_k := s_k(h(x)\ind\{i=N\}), \qquad k\in\nats.
\]
By \cref{thm:LLN}, for all $k\in\nats$
\[
\EE_\nu[\hslash_k] = \EE_\nu[v] \target(h), \qquad \EE_\nu[v] = \frac{2p_N}{p_0}.
\]
On the other hand,
\[
\rest_k(h) - \target(h) &= \frac{\sum_{j=1}^{k} \hslash_j}{\sum_{j=1}^{k} v_j}  - \target(h) = \frac{\sum_{j=1}^{k} (\hslash_j - \target(h) v_j)}{\sum_{j=1}^{k} v_j}.
\]
Let 
\[
d_k := \hslash_k - \target(h) v_k = s_k(\ind\{i=N\}(h(x) - \target(h))), \qquad k\in\nats.
\]
By \cref{lemma:tour_sums_iid}, $\{d_k\}_{k\in\nats}$ is a collection of \iid 
random variables under $\Pr_\nu$, with zero mean and variance 
$\sigma_d^2$. Thus, by the Central Limit Theorem
\[
k^{-1/2}\sum_{j=1}^{k} d_j = k^{1/2}\left(\frac{1}{k}\sum_{j=1}^{k} d_j\right) \convd \distNorm(0,\sigma_d^2).
\]
Now let us write
\[
k^{1/2}(\rest_k(h) - \target(h)) = \frac{k^{-1/2}\sum_{j=1}^{k} d_j}{k^{-1}\sum_{j=1}^{k} v_j} = \frac{a_k}{b_k}.
\]
We showed that $a_k\convd Z$ where $Z\sim\distNorm(0,\sigma_d^2)$, 
while $b_k\convas\EE_\nu[v]$ by \cref{thm:LLN}. Using Theorem 2.7(v) in 
\citet{vandervaart1998asymptotic}, we have that $(a_k,b_k)\convd(Z,\EE_\nu[v])$.
Since $\EE_\nu[v]=2p_0/p_N>0$, by the continuous mapping theorem 
\citep[Thm.\ 2.3]{vandervaart1998asymptotic} we conclude that
\[
k^{1/2}(\rest_k(h) - \target(h)) \convd \distNorm\left(0,\left[\frac{\sigma_d}{\EE_\nu[v]}\right]^2\right).
\]
\end{proof}

\subsection{Proofs for \cref{sec:index_process}}

\begin{lemma}\label{lemma:kexp_V_marginal_is_target_marginal}
Under \cref{assu:ELE}, for all $i\in\indset$ and $\target$-almost every $x$,
\[
K_i(x,\dee v') = \targetv_i(\dee v'),
\]
where the right-hand side is the push-forward of $\target_i$ under $V$. 
\end{lemma}
\begin{proof}
For $i\in\indsetone$, let
\[
K_i(x, \dee v') = \int_{x'} K_i(x,\dee x')\delta_{V(x')}(\dee v'),
\]
be the push-forward of $K_i(x,\dee x')$ under $V$. By \cref{assu:ELE},
$K_i(x, \dee v') = q_i(\dee v')$ for $\target$-almost every $x$.
On the other hand, since 
$K_i$ is $\target_i$ invariant for every $i\in\indsetone$, it follows that
$x'\sim\target_i$ whenever $x\sim\target_i$. In particular, this means that
$v'\sim\targetv_i$ whenever $x\sim\target_i$. The only possible choice for $q_i$
that makes these statements consistent is $q_i=\targetv_i$, for all 
$i\in\indsetone$.
\end{proof}

\begin{proof}[\textbf{Proof of \cref{prop:marginal_chain}}]
We concentrate on NRST because the proof for ST is identical.
Since $\statespace\times\reals$ is a Polish space, for $\target$-almost every $x$, 
there exists a regular conditional probability $K_i((x,v'),\dee x')$ 
\citep[Theorem 10.2.2]{dudley2002real} such that
the joint distribution of $(x',v'=V(x'))$ conditional on $x$ can be written as
\[
K_i(x,\dee x')\delta_{V(x')}(\dee v') = K_i(x,\dee v')K_i((x,v'),\dee x'),
\]
where $K_i(x,\dee v')$ is the push-forward of $K_i(x,\dee x')$ under $V$. By
\cref{assu:ELE} and \cref{lemma:kexp_V_marginal_is_target_marginal},
\[\label{eq:RCD_kexp_and_ELE}
K_i(x,\dee x')\delta_{V(x')}(\dee v') = \targetv_i(\dee v')K_i((x,v'),\dee x').
\] 
Hence, for $\target$-almost every $x$ and all $(i,\eps)$, 
the marginal distribution of $(i',\eps')$ conditional on $(x,i,\eps)$ after a 
full NRST step is
\[
&\knrst((x,i,\eps), (\statespace, i', \eps')) \\
&= \int_{x^\star}\sum_{i^\star,\eps^\star} \kexp((x,i,\eps),(\dee x^\star, i^\star, \eps^\star)) \ktem((x^\star,i^\star,\eps^\star),(\statespace,i',\eps')) \\
&= \int_{x^\star} K_i(x,\dee x^\star) \ktem((x^\star,i,\eps),(\statespace,i',\eps')) \\
&(\kexp \text{ only affects the }x \text{ component}) \\
&= \int_{x^\star} K_i(x,\dee x^\star)[\acc{i,i+\eps}(V(x^\star))\ind\{i'=i+\eps, \eps'=\eps\} + \rej{i,i+\eps}(V(x^\star))\ind\{i'=i, \eps'=-\eps\}] \\
&= \int_{x^\star} K_i(x,\dee x^\star) \int_{v^\star}\delta_{V(x^\star)}(\dee v^\star)[\acc{i,i+\eps}(v^\star)\ind\{i'=i+\eps, \eps'=\eps\} + \rej{i,i+\eps}(v^\star)\ind\{i'=i, \eps'=-\eps\}] \\
&(\text{integrate with respect to point mass at } V(x^\star)) \\
&= \int_{v^\star} \targetv_i(\dee v^\star) [\acc{i,i+\eps}(v^\star)\ind\{i'=i+\eps, \eps'=\eps\} + \rej{i,i+\eps}(v^\star)\ind\{i'=i, \eps'=-\eps\}]\int_{x^\star} K_i((x,v^\star), \dee x^\star) \\
&(\text{\cref{eq:RCD_kexp_and_ELE} and Fubini's}) \\
&= \int_{v^\star} \targetv_i(\dee v^\star) [\acc{i,i+\eps}(v^\star)\ind\{i'=i+\eps, \eps'=\eps\} + \rej{i,i+\eps}(v^\star)\ind\{i'=i, \eps'=-\eps\}] \\
&= \acc{i,i+\eps}\ind\{i'=i+\eps, \eps'=\eps\} + \rej{i,i+\eps}\ind\{i'=i, \eps'=-\eps\} \\
&=\mknrst((i,\eps),(i',\eps')).
\]
Since the right-hand side does not depend on $x$, we conclude that, for any initial
distribution $\mu$ on $\prodmespace$, the index process of the NRST chain 
initialized with $\mu$ follows marginally a finite state Markov chain on 
$\indset\times\epsset$ with transition kernel $\mknrst$. Its initial distribution 
is given by the push-forward of $\mu$ under the canonical projection 
$(x,i,\eps)\mapsto(i,\eps)$. Moreover, the kernel $\mknrst$ is invariant with 
respect to $\lifted(i,\eps)$ because the NRST kernel $\knrst$ is
$\lifted$-invariant (\cref{prop:lifted_invariance}) and $\lifted(i,\eps)$ is its 
marginal.
\end{proof}

\subsection{Proofs for \cref{sec:toureff}}

Before we begin, let us introduce notation specific for the index process. We
identify the atom set $\atom$ in product space $\prodss$ with the index process
state $(0,-1)$. Thus, we overload the notation for the hitting times of the atom 
$\{T_k\}_{k\in\nats}$ to mean visits of the index process to the state $(0,-1)$.
Moreover, define
\[\label{eq:def_top_level_aiming_up}
\mcU:=(N,+1),
\]
which is the state at the top level of the index process with momentum aiming up.
Let us write $T_\mcU$ for the first hitting time 
(\cref{eq:def_first_hit_return_times}) of the singleton $\{\mcU\}$. Then, we can
define the \emph{roundtrip time} as
\[\label{eq:def_RTT}
\rtt := \inf\{n>T_\mcU: (i,\eps)=(0,-1)\} = T_1\circ\theta_{T_\mcU}.
\]

Recall that, under \cref{assu:ELE}, the index process is a finite state Markov 
chain on $\indset\times\epsset$ (\cref{prop:marginal_chain}). 
We use $\Pr_{(i,\eps)}$ to denote the law corresponding to initializing the index process at $(i,\eps)$. To simplify the notation further, we write $\Pr_{(i_0,+)}$
if $\eps_0=+1$, and $\Pr_{(i_0,-)}$ otherwise. Note
that initializing NRST from the regenerative measure $\nu$ implies starting
the index process from $(0,+1)$.

\begin{lemma}\label{lemma:unif_pprior_sym_rejs}
If $\{p_i\}_{i=0}^N$ is uniform, then $\rej{i-1,i}=\rej{i,i-1}$ for all
$i\in\{1,\dots,N\}$.
\end{lemma}
\begin{proof}
Since $\lifted(i,\eps)$ is constant and non-zero, by the $\lifted$-invariance 
of the index process
\[
&\lifted(i,+1) = \lifted(i,-1)\rej{i,i-1} + \lifted(i-1,+1)(1-\rej{i-1,i}) \\
&\iff 1 = \rej{i,i-1} + 1-\rej{i-1,i} \\
&\iff \rej{i,i-1} = \rej{i-1,i}.
\]
\end{proof}

The following result analyzes the probabilities that the index process started
at either end of $\indset$ is able to reach the opposite end before returning to
the initial state. As we shall see in the proof of \cref{thm:TE_formula}, these are
key ingredients for the distribution of the number of visits to level $N$ within
a tour---which in turn determines TE (\cref{sec:toureff}).

\begin{lemma}\label{lemma:prob_visit_top_uniform}
For the index process with uniform $\{p_i\}_{i=0}^N$,
\[
\Pr_{(0,+)}(T_\mcU < T_1) = \Pr_{(N,-)}(T_\mcU > T_1) = \frac{\EE_{(0,+)}[\tau]}{\EE_{(0,-)}[\rtt]},
\]
where $\rtt$ is the roundtrip time (\cref{eq:def_RTT}) and $\tau$ is the length of
a typical tour (defined in \cref{sec:NRST}). Consequently, for reversible ST
\[
\Pr_{(0,+)}(T_\mcU < T_1) = \Pr_{(N,-)}(T_\mcU > T_1) = \frac{1}{2N + 2\sum_{i=1}^N \frac{\rho_i}{\alpha_i}},
\]
whereas for NRST
\[
\Pr_{(0,+)}(T_\mcU < T_1) = \Pr_{(N,-)}(T_\mcU > T_1) = \frac{1}{1 + \sum_{i=1}^N \frac{\rho_i}{\alpha_i}}.
\]
\end{lemma}

\begin{proof}
Define for all $k\in\nats$
\[
m_k := \ind\{0 < s_k(\ind\{(i_n,\eps_n)=\mcU\})\} = \ind\left\{0 < \sum_{n=T_{k-1}+1}^{T_k} \ind\{(i_n,\eps_n)=\mcU\}\right\},
\]
which indicates whether the process hit $\mcU$ at least once during the $k$-th 
tour or not. In other words, $m_k$ indicates if the $k$-th tour is a 
roundtrip. Note that
\[\label{eq:prob_visit_top_uniform_sum}
\Pr_{(0,+)}\left(T_\mcU < T_1\right) = \EE_{(0,+)}[m_1],
\]
Since the collection $\{m_k\}_{k\in\nats}$ is \iid, it follows from the strong law of 
large numbers that
\[
\lim_{k\to\infty} \frac{1}{k}\sum_{j=1}^k m_j =  \Pr_{(0,+)}(T_\mcU < T_1), \qquad \Pr_{(0,+)}-\text{a.s.}
\]
Now, let $k(n)$ be the index of the tour occurring at step $n$; i.e.,
\[
\forall n\in\posInts:\ k(n) := 1+\sum_{m=0}^n \ind\{(i_m,\eps_m)=(0,-1)\}.
\]
Then, \cref{thm:LLN} shows that
\[
\lim_{n\to\infty} \frac{n}{k(n)} = \lim_{k\to\infty} \frac{T_k}{k}= \lim_{k\to\infty} \frac{1}{k}\sum_{j=1}^k\tau_j = \EE_{(0,+)}[\tau], \quad \Pr_{(0,+)}-\text{a.s.}
\]
Under a uniform distribution over levels, \cref{eq:LLN_toursums_identity}
in \cref{thm:LLN} shows that $\EE_{(0,+)}[\tau]=2(N+1)$ for NRST. A similar
identity can be obtained for ST by using the fact that it admits the atom 
$\atom_\text{ST}:=\statespace\times\{0\}\times\epsset$. This means that for ST,
$\EE_{(0,+)}[\tau]=N+1$.

On the other hand, let $R(n)$ be the number of roundtrips up to step $n$; i.e.,
\[
R(n) := \sum_{j=1}^{k(n)}m_j.
\]
$\{R_n\}_{n\in\posInts}$ is a renewal process with \iid 
holding times that have expectation $\EE_{(0,-)}[\rtt]$. By the law of large 
numbers for renewal processes \citep[Prop. 9.1.25]{cinlar2013introduction},
\[
\lim_{n\to\infty} \frac{R(n)}{n} = \frac{1}{\EE_{(0,-)}[\rtt]},
\]
almost surely. Combining the above two facts gives
\[
\Pr_{(0,+)}(T_\mcU < T_1) &= \lim_{k\to\infty} \frac{1}{k}\sum_{j=1}^k m_j \\
&= \lim_{n\to\infty} \frac{n}{k(n)}\frac{R(n)}{n} \\
&= \left[\lim_{n\to\infty} \frac{n}{k(n)}\right] \left[\lim_{n\to\infty}\frac{R(n)}{n}\right] \\
&= \frac{\EE_{(0,+)}[\tau]}{\EE_{(0,-)}[\rtt]}.
\]
Using the expressions from Theorem 1 in \citet{syed2022nrpt} for the expected
roundtrip times, together with the values for the expected tour length discussed
above, we see that for ST,
\[
\Pr_{(0,+)}(T_\mcU < T_1) &= \frac{1}{2N + 2\sum_{i=1}^N \frac{\rho_i}{\alpha_i}},
\]
whereas for NRST
\[
\Pr_{(0,+)}(T_\mcU < T_1) &= \frac{1}{1 + \sum_{i=1}^N \frac{\rho_i}{\alpha_i}}.
\]

Finally, to see that
\[
\Pr_{(0,+)}(T_\mcU < T_1) = \Pr_{(N,-)}(T_\mcU > T_1),
\]
simply note that we can obtain the right-hand side by re-labelling $i\gets N-i$.
But since the quantity $\sum_{i=1}^N \frac{\rho_i}{\alpha_i}$ is invariant under
permutations of the indices, the result is identical.
\end{proof}

\begin{proof}[\textbf{Proof of \cref{thm:TE_formula}}]
Following the definition of $\TE$, we proceed by describing the distribution of 
the number of visits to level $N$ within a tour. Under \cref{assu:ELE}, 
\cref{prop:marginal_chain} lets us analyze this distribution using only the index
process. Now, the probability that the chain hits level $N$ before returning to 
the atom is $\Pr_{(0,+)}(T_\mcU < T_1)$.
From there, by the strong Markov property, a second visit will occur with probability 
$\Pr_{(N,-)}(T_\mcU < T_1)$. 
If that happens, then---again by the Markov property---a third visit will occur with 
the same probability. 
Iterating this argument shows that
\[\label{eq:visits_decomposition}
v \overset{d}{=} RUF
\]
where $R$ and $U$ are independent random variables with marginal distributions
\[
R&\sim \distBern\left(\Pr_{(0,+)}(T_\mcU < T_1)\right) \\
U&\sim \distGeom_1\left(\Pr_{(N,-)}(T_\mcU > T_1)\right).
\]
$\distGeom_1(p)$ refers to a geometric distribution supported on
$\nats=\{1,2,\dots\}$, while the constant $F$ is related to the delay in turning
back after hitting the boundary at $N$. For ST, $F=1$ because this process 
immediately forgets its current direction. On the other hand, NRST needs $F=2$ 
because, after hitting $(N,+)$, it takes an additional step to flip
(deterministically) to $(N,-)$.
The result now follows directly from \cref{lemma:prob_visit_top_uniform}. Indeed, let 
$p=\Pr_{(0,+)}(T_\mcU < T_1)$. Then, for NRST we have
\[
\EE_\nu[v^2] = 4p\left[\frac{1}{p^2} + \frac{1-p}{p^2}\right] = 4\frac{2-p}{p}.
\]
Using this together with $\EE_\nu[v]=2$ (via \cref{eq:visits_decomposition} or
\cref{eq:LLN_toursums_identity} with $f=\ind\{i=N\}$) yields
\[
\TE &= \frac{p}{2-p} = \frac{1}{2/p-1} = \frac{1}{1 + 2\sum_{i=1}^N \frac{\rho_i}{\alpha_i}}.
\]
Likewise for ST, \cref{eq:visits_decomposition} gives $\EE_\nu[v]=1$---we can
verify this by modifying \cref{eq:LLN_toursums_identity} for (reversible) ST
as discussed in the proof of \cref{lemma:prob_visit_top_uniform}. Finally,
\[
\TE = \frac{1}{2/p-1} = \frac{1}{4N-1 + 4\sum_{i=1}^N \frac{\rho_i}{\alpha_i}}.
\]
\end{proof}

\begin{proof}[\textbf{Proof of \cref{thm:CLT_bdd_h}}]
Applying the triangle inequality twice,
\[
|s(\ind\{i=N\}[h(x)-\pi(h)])| &\leq  \sum_{n=0}^{T_{1}} \ind\{i_n=N\}|h(x_n)-\pi(h)| \leq  2v.
\]
Therefore,
\[
\EE_\nu[s(\ind\{i=N\}[h(x)-\pi(h)])^2] \leq 4\EE_\nu[v^2].
\]
Using the definition of $\sigma(h)^2$ and applying the above bound gives
\[
\sigma(h)^2 \leq 4\frac{\EE_\nu[v^2]}{\EE_\nu[v]^2} = \frac{4}{\TE}.
\]
To check that $\TE>0$, first recall that $\EE_\nu[v]=2p_N/p_0>0$ (\cref{thm:LLN}).
Secondly, we know that under \cref{assu:ELE} the
index process follows a finite Markov chain on $\indset\times\epsset$ 
(\cref{prop:marginal_chain}). Therefore, the number of visits to $i=N$ before 
absorption into the atom $\atom$---i.e., $v$---has finite second moment
\citep[see e.g.][Theorem 3.3.3]{kemeny1960finite}.
\end{proof}

\subsection{Proofs for \cref{sec:limit_grid}}

\begin{lemma}\label{lemma:finite_norm_const}
$\mcZ(\beta) < \infty$ for all $\beta\in[0,1]$.
\end{lemma}
\begin{proof}
\[
\mcZ(\beta) =\EEi{0}[e^{-\beta V(x)}] \leq \EEi{0}[e^{-V(x)}]^{\beta} <\infty.
\]
The first bound follows from Jensen's since $u\mapsto u^\beta$ is concave in 
$[0,\infty)$ when $\beta\in[0,1]$, and the last one is due to the requirement
that $\mcZ=\mcZ(1)<\infty$ in \cref{eq:def_target_dist}.
\end{proof}

\begin{lemma}\label{lemma:MGF}
For all $\beta\in[0,1]$ and $\xi \in [-(1-\beta),\beta]$,
\[
\mcM^{(\beta)}(\xi) := \EEi{\beta}[e^{\xi V(x)}]<\infty.
\]
It follows that for all $\beta\in(0,1)$, $\mcM^{(\beta)}(\xi)$ is finite on an
open interval around $0$. In particular, this implies that for all $a>0$,
\[
\EEi{\beta}[|V(x)|^a]<\infty.
\]
\end{lemma}
\begin{proof}
Note that
\[
\mcM^{(\beta)}(\xi) &= \EEi{\beta}[e^{\xi V(x)}] = \frac{1}{\mcZ(\beta)}\EEi{0}[e^{- (\beta-\xi) V(x)}] = \frac{\mcZ(\beta-\xi)}{\mcZ(\beta)}.
\]
By \cref{lemma:finite_norm_const}, the right-hand side is finite for all
$\beta\in[0,1]$ whenever
\[
0 \leq \beta-\xi \leq 1 \iff -(1-\beta) \leq \xi \leq \beta.
\]
For the second statement, note that for $\beta\in(0,1)$, $\mcM^{(\beta)}(\xi)$
is finite for all $\xi\in (-(1-\beta),\beta)$, which is an open interval 
around $0$. The result follows for $a\in\nats$ by standard results for 
moment-generating functions 
\citep[see e.g.][Theorem 9.3.3]{rosenthal2006first}. To extend to $a>0$, let
\[
\ceil{a} := \min\{n\in\nats:a\leq n\}.
\]
Then,
\[
\EEi{\beta}[|V(x)|^a] = \EEi{\beta}\left[ \left( |V(x)|^{\ceil{a}} \right)^{\frac{a}{\ceil{a}}} \right] \leq \EEi{\beta}\left[ |V(x)|^{\ceil{a}} \right]^{\frac{a}{\ceil{a}}} <\infty,
\]
by Jensen's inequality since $z\mapsto z^{\frac{a}{\ceil{a}}}$ is concave.
\end{proof}

\begin{lemma}\label{lemma:norm_const_derivs}
$\mcZ(\beta)\in\mcC^\infty((0,1))$, and for all $k\in\nats$ and $\beta\in(0,1)$,
\[
\frac{\mathrm{d}^k}{\mathrm{d}\beta^k}\mcZ(\beta) = (-1)^k \mcZ(\beta) \EEi{\beta}[V^k].
\]
\end{lemma}
\begin{proof}
This is a special case of Lemma 4.5 in \citet{vandervaart1998asymptotic}.
\end{proof}

\begin{corollary}\label{cor:mean_V_properties}
The function $\beta\mapsto\EEi{\beta}[V]$ is $C^\infty((0,1))$ and monotone 
non-increasing, with
\[
\der{\EEi{\beta}[V]}{\beta} = -\var_\beta(V).
\]
Additionally,
\[
\forall \beta\in[0,1]:\ \log(\mcZ(\beta)) = -\int_0^\beta \EEi{\beta'}[V] \dee \beta'.
\]
\end{corollary}
\begin{proof}
The case $k=1$ in \cref{lemma:norm_const_derivs} shows that for all 
$\beta\in(0,1)$,
\[\label{eq:mean_V_as_Z_ratio}
\EEi{\beta}[V] = - \frac{\mcZ'(\beta)}{\mcZ(\beta)}.
\]
Since $\mcZ\in\mcC^\infty((0,1))$, $\EEi{\cdot}[V]$ is the ratio of two smooth,
strictly positive functions. Therefore, $\EEi{\cdot}[V]\in\mcC^\infty((0,1))$ too. 
Moreover, differentiating \cref{eq:mean_V_as_Z_ratio} and using 
\cref{lemma:norm_const_derivs} for $k\in\{1,2\}$ shows that
\[
\der{\EEi{\beta}[V]}{\beta} = -\frac{\mcZ''(\beta)}{\mcZ(\beta)} + \frac{\mcZ'(\beta)^2}{\mcZ(\beta)^2} = -\EEi{\beta}[V^2] + \EEi{\beta}[V]^2 = -\var_\beta(V) \leq 0.
\]
On the other hand, we can rewrite the right-hand side of 
\cref{eq:mean_V_as_Z_ratio} as
\[
\EEi{\beta}[V] = -\der{\log(\mcZ(\beta))}{\beta}.
\]
Integrating on both sides and using $\mcZ(0)=1$ gives the last result.
\end{proof}

\begin{lemma}\label{lemma:acc_prob_as_ratio}
If $c$ is differentiable, then for all $\beta,\beta'\in(0,1)$,
\[\label{eq:acc_prob_as_ratio}
\acc{\beta,\beta'} = \frac{\EEi{\bar\beta}[\exp(-|\dbeta||c'(\check\beta) - V|/2)]}{\EEi{\bar\beta}[\exp(-\dbeta (c'(\check\beta)- V)/2)]},
\]
where $\dbeta=\beta'-\beta$, $\bar\beta:=(\beta+\beta')/2$, and 
$\check\beta\in(\min\{\beta,\beta'\}, \max\{\beta,\beta'\})$ satisfies the mean
value theorem equation
\[\label{eq:def_c_check_beta}
c'(\check{\beta}) = \frac{1}{\dbeta}\left(c(\bar\beta+\dbeta/2) - c(\bar\beta-\dbeta/2)\right).
\]
\end{lemma}
\begin{proof}
The following is an adaptation of the argument in \citet{predescu2004incomplete} 
to the simulated tempering case. Consider first the case $\dbeta\geq0$, 
and note that
\[
\acc{\beta,\beta + \dbeta}(x) &= \exp(-\max\{0, \dbeta V(x) - [c(\beta+\dbeta)-c(\beta)]\}) \\
&= \exp(-\dbeta\max\{0, V(x) - [c(\beta+\dbeta)-c(\beta)]/\dbeta \}) \\
&= \exp(-\dbeta\max\{0, V(x) - c'(\check\beta)\}) && (\text{MVT: } \check\beta \in (\beta,\beta'))\\
&= \exp(\dbeta\min\{0, c'(\check\beta) - V(x)\}) \\
&= \exp(\dbeta[c'(\check\beta) - V(x)]/2) \exp(-\dbeta|c'(\check\beta) - V(x)|/2),
\]
where MVT stands for the mean value theorem. Integrating with respect to 
$\pibeta$,
\[
\acc{\beta,\beta+\dbeta} &= \EEi{\beta}[\acc{\beta,\beta+\dbeta}(x)] \\
&= \exp(\dbeta c'(\check\beta)/2)\EEi{\beta}[\exp(-\dbeta V/2) \exp(-\dbeta|c'(\check\beta) - V|/2)] \\
&= \frac{\exp(\dbeta c'(\check\beta)/2)}{\mcZ(\beta)}\EEi{0}[\exp(-\underbrace{(\beta + \dbeta/2)}_{\bar\beta} V) \exp(-\dbeta|c'(\check\beta) - V|/2)] \\
&= \frac{\mcZ(\bar\beta)\exp(\dbeta c'(\check\beta)/2)}{\mcZ(\beta)} \EEi{\bar\beta} [\exp(-\dbeta|c'(\check\beta) - V|/2)].
\]
On the other hand,
\[\label{eq:acc_prob_as_ratio_denom}
\mcZ(\beta) &= \EEi{0}[e^{-\beta V}] = \mcZ(\bar\beta)\EEi{\bar\beta}[e^{\dbeta V/2}].
\]
Hence,
\[
\acc{\beta,\beta+\dbeta} &= \frac{\EEi{\bar\beta} [\exp(-\dbeta|c'(\check\beta) - V|/2)]}{\EEi{\bar\beta}[\exp(-\dbeta (c'(\check\beta) - V)/2)]}.
\]

Similarly for $\dbeta<0$,
\[
\acc{\beta,\beta -|\dbeta|}(x) &= \exp(-\max\{0, -|\dbeta| V - [c(\beta-|\dbeta|)-c(\beta)]\}) \\
&= \exp(|\dbeta|\min\{0,  V - [c(\beta)-c(\beta-|\dbeta|)]/|\dbeta| \}) \\
&= \exp(|\dbeta|\min\{0,  V - c'(\check\beta) \}) && (\text{MVT: }\check\beta\in(\beta',\beta)) \\
&= \exp(-|\dbeta|[c'(\check\beta) - V]/2) \exp(-|\dbeta||c'(\check\beta) - V|/2)
\]
Now
\[
\acc{\beta,\beta-|\dbeta|} &= \EEi{\beta}[\acc{\beta,\beta-|\dbeta|}(x)] \\
&= \exp(-|\dbeta|c'(\check\beta)) \EEi{\beta}[\exp(|\dbeta|V/2) \exp(-|\dbeta||c'(\check\beta) - V|/2)] \\
&= \frac{\exp(-|\dbeta|c'(\check\beta))}{\mcZ(\beta)} \EEi{0}[\exp(-\underbrace{(\beta - |\dbeta|/2)}_{\bar\beta}V) \exp(-|\dbeta||c'(\check\beta) - V|/2)] \\
&= \frac{\mcZ(\bar\beta)\exp(-|\dbeta|c'(\check\beta))}{\mcZ(\beta)} \EEi{\bar\beta}[\exp(-|\dbeta||c'(\check\beta) - V|/2)].
\]
But since $\mcZ(\beta)=\mcZ(\bar\beta)\EEi{\bar\beta}[e^{-|\dbeta| V/2}]$,
we get
\[
\acc{\beta,\beta-|\dbeta|} = \frac{\EEi{\bar\beta}[\exp(-|\dbeta||c'(\check\beta) - V|/2)]}{\EEi{\bar\beta}[\exp(-\dbeta (c'(\check\beta)- V)/2)]}
\]
as necessary. Note that the right-hand side is also valid for $\dbeta\geq0$. 
Finally, we can verify that $\acc{\beta,\beta'}\leq 1$ because for any 
$u\in\reals$, $e^{-|u|} \leq e^{-u}$.
\end{proof}

In the last Lemma, the variable $\check{\beta}$ depends on $\beta$ and $\beta'$.
Ultimately, we will keep one of the end points, say $\beta$, fixed and vary the other end point $\beta'$. As a result, in the following we write $\check{\beta}(\dbeta)$ to emphasize this dependency.

\begin{lemma}\label{lemma:c_check_beta_C2}
Under \cref{assu:c_smooth}, for each $\bar\beta\in(0,1)$, the function
$\dbeta \mapsto c'(\check{\beta}(\dbeta))$ defined by \cref{eq:def_c_check_beta}
(\cref{lemma:acc_prob_as_ratio}) is $\mcC^2((-r(\bar\beta),r(\bar\beta))))$, with
\[\label{eq:def_rbarbeta}
r(\bar\beta) := 2\min\{\bar\beta, 1-\bar\beta\},
\]
and admits the following second-order Taylor expansion around $\dbeta=0$ with
Peano remainder
\[\label{eq:c_check_beta_C2_zero_deriv}
c'(\check{\beta}(\dbeta)) = c'(\bar\beta) + 0\cdot\dbeta + \frac{1}{24}c'''(\bar\beta)\dbeta^2 + o(\dbeta^2).
\]
\end{lemma}
\begin{proof}
Since $c'(\check{\beta}(\cdot))$ is a linear combination of $c$ evaluated at two 
points, the function is as smooth as $c$ for $\dbeta\neq0$. To understand its 
behavior at $\dbeta=0$, consider a third order expansion of $c$ around $\bar\beta$ 
with Peano remainder
\[
c(\bar\beta+\delta) = c(\bar\beta) + c'(\bar{\beta})\delta + \frac{1}{2}c''(\bar\beta)\delta^2 +\frac{1}{6}c'''(\bar\beta)\delta^3 + h(\delta)\delta^3,
\]
with $h$ a function such that $\lim_{\delta\to0}h(\delta)=0$. 
Substituting $\delta$ with $\pm\dbeta/2$ and using these approximations in 
\cref{eq:def_c_check_beta} gives
\[
&c'(\check{\beta}(\dbeta)) \\
&= \frac{1}{\dbeta}\left(c(\bar\beta+\dbeta/2) - c(\bar\beta-\dbeta/2)\right)\\
&= \frac{1}{\dbeta}\left(\frac{1}{2}c'(\bar\beta)[\dbeta+\dbeta] + \frac{1}{48}c'''(\bar\beta)[\dbeta^3+\dbeta^3] +  \underbrace{\frac{1}{8}[h(\dbeta/2)-h(-\dbeta/2)]}_{:=\tilde{h}(\dbeta)}\dbeta^3\right)\\
&= c'(\bar\beta) + 0\cdot\dbeta + \frac{1}{24}c'''(\bar\beta)\dbeta^2 + \tilde{h}(\dbeta)\dbeta^2.
\]
This is now a second-order Taylor expansion with Peano remainder for 
$c'(\check{\beta}(\cdot))$. Since $c\in\mcC^3((0,1))$, by Theorem 3 in
\citet{oliver1954exact} we conclude that
$c'(\check{\beta}(\cdot))\in\mcC^2((-r(\bar\beta),r(\bar\beta)))$, as required.
\end{proof}

\begin{lemma}\label{lemma:deriv_expectations}
Suppose \cref{assu:c_smooth} holds. Fix $\bar\beta\in(0,1)$ and let $r(\bar\beta)$
be as in \cref{eq:def_rbarbeta}. Then, for $\phi\in\{-1,1\}$ 
and $|\dbeta|<r(\bar\beta)$,
\[
&\der{}{\dbeta}\EEi{\bar\beta}[\exp(\phi\dbeta (c(\check{\beta}(\dbeta))-V)/2)] \\
&\phantom{M}= \frac{\phi}{2}\EEi{\bar\beta}[\exp(\phi \dbeta(\gamma(\dbeta)-V)/2)(\gamma(\dbeta)-V + \dbeta\gamma'(\dbeta))], \\
&\dder{}{\dbeta}\EEi{\bar\beta}[\exp(\phi\dbeta(c(\check{\beta}(\dbeta))-V)/2)] \\
&\phantom{M}=\frac{\phi}{2}\EEi{\bar\beta}[\exp(\phi \dbeta(\gamma(\dbeta)-V)/2)(2\gamma'(\dbeta) + \dbeta\gamma''(\dbeta))] \\
&\phantom{M='}+\frac{1}{4}\EEi{\bar\beta}[\exp(\phi \dbeta(\gamma(\dbeta)-V)/2)(\gamma(\dbeta)-V + \dbeta\gamma'(\dbeta))^2],
\]
where $\gamma(\dbeta):=c'(\check{\beta}(\dbeta))$.
Furthermore, if \cref{assu:nullset} holds too, then for all
$|\dbeta|<r(\bar\beta)$,
\[
&\der{}{\dbeta}\EEi{\bar\beta}[\exp(- \dbeta |c(\check{\beta}(\dbeta))-V|/2)] \\
&= -\frac{1}{2}\EEi{\bar\beta}[\exp(- \dbeta|\gamma(\dbeta)-V|/2)(|\gamma(\dbeta)-V| + \dbeta\gamma'(\dbeta)\sgn(\gamma(\dbeta)-V))] \\
&\dder{}{\dbeta}\EEi{\bar\beta}[\exp(- \dbeta|c(\check{\beta}(\dbeta))-V|/2)] \\
&= -\frac{1}{2}\EEi{\bar\beta}[\exp(- \dbeta|\gamma(\dbeta)-V|/2)\sgn(\gamma(\dbeta)-V)(2\gamma'(\dbeta) + \dbeta\gamma''(\dbeta))] \\
&\phantom{='}+\frac{1}{4}\EEi{\bar\beta}[\exp(- \dbeta|\gamma(\dbeta)-V|/2)(|\gamma(\dbeta)-V| + \dbeta\gamma'(\dbeta)\sgn(\gamma(\dbeta)-V))^2].
\]
\end{lemma}
\begin{proof}
We focus on the case with absolute value and leave the other to the reader
since it requires a similar but simpler argument. To simplify the notation, 
in the following we write $z$ instead of $\dbeta$. Now, define
\[
F(v,z) := \exp(- z|\gamma(z)-v|/2),
\]
and note that $z\mapsto\EEi{\bar\beta}[F(V,z)]$ is our function of interest. From
\cref{lemma:c_check_beta_C2}, we know that
$\gamma\in\mcC^2((-r(\bar\beta),r(\bar\beta)))$ under \cref{assu:c_smooth}. It follows that for each $v$, the function 
$F(v,\cdot)$ is absolutely continuous and thus almost everywhere differentiable
on $(-r(\bar\beta),r(\bar\beta))$. Therefore, $F$ admits a well-defined 
partial derivative on the set
\[
A := \{(v,z): |z|< r(\bar\beta), \gamma(z)\neq v\}.
\]
Indeed, for $(v,z)\in A$,
\[
\D{F}{z}(v,z) = -\frac{1}{2}\exp(- z|\gamma(z)-v|/2)[|\gamma(z)-v| + z\gamma'(z)\sgn(\gamma(z)-v)].
\]
Furthermore, for each $v$, the function $z\mapsto\D{F}{z}$ is of bounded 
variation. Thus, it is almost everywhere differentiable on
$(-r(\bar\beta),r(\bar\beta))$. It follows that $\DD{F}{z}$ exists on $A$ too, with
\[
\DD{F}{z}(v,z) &= -\frac{1}{2}\exp(- z|\gamma(z)-v|/2)\D{}{z}[|\gamma(z)-v| + z\gamma'(z)\sgn(\gamma(z)-v)] \\
&\phantom{='}-[|\gamma(z)-v| + z\gamma'(z)\sgn(\gamma(z)-v)]\D{}{z}\exp(- z|\gamma(z)-v|/2) \\
&= -\frac{1}{2}\exp(- z|\gamma(z)-v|/2)\sgn(\gamma(z)-v)[2\gamma'(z) + z\gamma''(z)] \\
&\phantom{='}+\frac{1}{4}\exp(- z|\gamma(z)-v|/2)[|\gamma(z)-v| + z\gamma'(z)\sgn(\gamma(z)-v)]^2.
\]

Now, our aim is to show that for $i\in\{1,2\}$,
\[
\frac{\mathrm{d}^i}{\mathrm{d}z^i}\EEi{\bar\beta}[F(V,z)] = \EEi{\bar\beta}\left[\Di{F}{z}(V,z)\right].
\]
We will do this only for the second derivative since the first one requires 
a similar argument. Fix $(v,z)\in A$. Since 
$\gamma\in\mcC^2((-r(\bar\beta),r(\bar\beta)))$, there exists $\delta=\delta(v,z)$ 
with $0<\delta<r(\bar\beta)$ such that for all $z'$ with 
$|z'-z|< \delta$, we have $(v,z')\in A$. Hence, $z'\mapsto \D{F}{z}(v,z')$ is 
$\mcC^1((z-\delta,z+\delta))$ and the MVT implies that
\[\label{eq:MVT_second_deriv}
\frac{\D{F}{z}(v,z')-\D{F}{z}(v,z)}{z'-z} = \DD{F}{z}(v,\check z),
\]
for all $z'\in(z-\delta,z+\delta)\setminus\{z\}$ and some
$\check{z}\in(\min\{z,z'\},\max\{z,z'\})$. On the other hand, for all $(v,z)\in A$,
\[\label{eq:abs_val_second_deriv_bound}
\left|\DD{F}{z}(v,z)\right| &\leq \frac{1}{2}\exp(- z|\gamma(z)-v|/2)|2\gamma'(z) + z\gamma''(z)| \\
&\phantom{='}+\frac{1}{4}\exp(- z|\gamma(z)-v|/2)[|\gamma(z)-v| + z\gamma'(z)\sgn(\gamma(z)-v)]^2 \\
&\leq \frac{1}{2}\exp(- z|\gamma(z)-v|/2)|2\gamma'(z) + z\gamma''(z)| \\
&\phantom{='} + \frac{1}{2}\exp(- z|\gamma(z)-v|/2)[(\gamma(z)-v)^2 + (z\gamma'(z))^2].
\]
Let us call the right-hand side of the above inequality
\[
G(v,z) &:= \frac{1}{2}\exp(- z|\gamma(z)-v|/2)|2\gamma'(z) + z\gamma''(z)| \\
&\phantom{='} + \frac{1}{2}\exp(- z|\gamma(z)-v|/2)[(\gamma(z)-v)^2 + (z\gamma'(z))^2].
\]
Note that, for each $v$, the function $G(v,\cdot)$ is continuous for all 
$|z|<r(\bar\beta)$. Therefore, $G(v,\cdot)$ achieves its maximum on any compact
subset of $(-r(\bar\beta),r(\bar\beta))$. In particular, for any 
$0<\varepsilon<r(\bar\beta)$,
\[
z^\star(v) := \argmax_{z\in[-r(\bar\beta)+\varepsilon,r(\bar\beta)-\varepsilon]} G(v,z).
\]
Let $G(v):=G(v,z^\star(v))$. We want to show that this function is integrable. 
We shall see that $G(\cdot,z)$ is in fact integrable for every 
$|z|<r(\bar\beta)$. Indeed, using the bound $e^{|u|}\leq e^u+e^{-u}$, we see that
\[
\exp(-z|\gamma(z)-v|/2) &\leq e^{-z\gamma(z)/2}e^{zv/2} + e^{z\gamma(z)/2}e^{-zv/2}.
\]
It follows that for $p\geq 0$,
\[\label{eq:derivative_kernel_expectations}
&\EEi{\bar\beta}[|V|^p\exp(- z|\gamma(z)-V|/2)] \\
&\leq e^{z\gamma(z)/2}\EEi{\bar\beta}[|V|^pe^{-zV/2}] + e^{-z\gamma(z)/2}\EEi{\bar\beta}[|V|^pe^{zV/2}] \\
&= \frac{e^{z\gamma(z)/2}\mcZ(\bar\beta-z/2)}{\mcZ(\bar\beta)}\EEi{\bar\beta-z/2}[|V|^p] + \frac{e^{-z\gamma(z)/2}\mcZ(\bar\beta+z/2)}{\mcZ(\bar\beta)}\EEi{\bar\beta+z/2}[|V|^p],
\]
and the right-hand side is finite for $|z|<r(\bar\beta)$ by 
\cref{lemma:finite_norm_const,lemma:MGF}. Hence, for $|z|<r(\bar\beta)$,
\[
\EEi{\bar\beta}[G(V,z)] &\leq \frac{1}{2}(|2\gamma'(z) + z\gamma''(z)|+(z\gamma'(z))^2)\EEi{\bar\beta}[\exp(- z|\gamma(z)-V|/2)] \\
&\phantom{='} + \frac{1}{2}\EEi{\bar\beta}[\exp(- z|\gamma(z)-V|/2)(\gamma(z)-V)^2] \\
&\leq \frac{1}{2}(|2\gamma'(z) + z\gamma''(z)|+(z\gamma'(z))^2 + 2\gamma(z)^2)\EEi{\bar\beta}[\exp(- z|\gamma(z)-V|/2)] \\
&\phantom{='} + \EEi{\bar\beta}[V^2\exp(- z|\gamma(z)-V|/2)],
\]
and the right-hand side is finite by \cref{eq:derivative_kernel_expectations}.
Since $|z^*(v)|<r(\bar\beta)$, we conclude that $G(v):=G(v,z^\star(v))$ is 
integrable. And since this function maximizes the right-hand side of 
\cref{eq:MVT_second_deriv}, we see that for all $v$ and 
$|z|<r(\bar\beta)-\varepsilon$,
\[
\left|\frac{\ind\{\gamma(z')\neq v\}\D{F}{z}(v,z')-\ind\{\gamma(z)\neq v\}\D{F}{z}(v,z)}{z'-z} \right| \leq G(v),
\]
Furthermore, since $(v,z)\in A$ implies $(v,z')\in A$ when $|z'-z|<\delta$,
\[
&\lim_{z'\to z} \frac{\ind\{\gamma(z')\neq v\}\D{F}{z}(v,z')-\ind\{\gamma(z)\neq v\}\D{F}{z}(v,z)}{z'-z} \\
&=\lim_{z'\to z} \frac{\ind\{\gamma(z)\neq v\}\D{F}{z}(v,z')-\ind\{\gamma(z)\neq v\}\D{F}{z}(v,z)}{z'-z} \\
&= \ind\{\gamma(z)\neq v\}\lim_{z'\to z} \frac{\D{F}{z}(v,z')-\D{F}{z}(v,z)}{z'-z}\\
&= \ind\{\gamma(z)\neq v\}\DD{F}{z}(v,z).
\]
Finally, since for each $z$, the set $\{v:\gamma(z)=v\}$ is $\pi_0$-null by 
\cref{assu:nullset}---and therefore it is $\target^{(\bar\beta)}$-null---the
dominated convergence theorem shows that for all $|z|<r(\bar\beta)-\varepsilon$,
\[
\dder{}{z}\EEi{\bar\beta}\left[F(V,z)\right] &= \der{}{z}\EEi{\bar\beta}\left[\D{F}{z}(V,z)\right] \\
&= \lim_{z'\to z}\frac{\EEi{\bar\beta}\left[\D{F}{z}(V,z')-\D{F}{z}(V,z)\right]}{z'-z} \\
&= \lim_{z'\to z}\frac{\EEi{\bar\beta}\left[\ind\{\gamma(z')\neq V\}\D{F}{z}(V,z')-\ind\{\gamma(z)\neq V\}\D{F}{z}(V,z)\right]}{z'-z} \\
&= \lim_{z'\to z}\EEi{\bar\beta}\left[\frac{\ind\{\gamma(z')\neq V\}\D{F}{z}(V,z')-\ind\{\gamma(z)\neq V\}\D{F}{z}(V,z)}{z'-z}\right] \\
&= \EEi{\bar\beta}\left[\lim_{z'\to z}\frac{\ind\{\gamma(z')\neq V\}\D{F}{z}(V,z')-\ind\{\gamma(z)\neq V\}\D{F}{z}(V,z)}{z'-z}\right] \\
&= \EEi{\bar\beta}\left[\ind\{\gamma(z)\neq V\}\DD{F}{z}(V,z)\right] \\
&= \EEi{\bar\beta}\left[\DD{F}{z}(V,z)\right],
\]
as required. And since the result holds for any $\varepsilon>0$, we conclude that 
the derivative exists for all $|z|<r(\bar\beta)$.
\end{proof}

\begin{lemma}\label{lemma:ratio_expansion}
Let $x:\mcO\to\reals_+$, $y:\mcO\to(0,\infty)$, with $x,y\in\mcC^2(\mcO)$ for some
open set $\mcO\subset\reals$. Consider
\[
f(z) := \frac{x(z)}{y(z)}.
\]
Then,
\[
\der{f}{z}(z) &= \frac{1}{y(z)^2}(y(z) x'(z) - x(z) y'(z)) \\
\dder{f}{z}(z) &= \frac{1}{y(z)^3} (y(z)^2 x''(z) - y(z) (2 x'(z) y'(z) + x(z) y''(z)) + 2 x(z) y'(z)^2).
\]
Furthermore, if $z_0\in\mcO$ is such that $x(z_0)=y(z_0)=1$ and 
$x''(z_0)=y''(z_0)$, then
\[
f(z) &= 1 + (x'(z_0) - y'(z_0))(z - z_0) - y'(z_0) (x'(z_0)-y'(z_0))(z - z_0)^2 \\
&\phantom{=1'}+ o((z-z_0)^2).
\]
\end{lemma}
\begin{proof}
The expressions for the derivatives can be directly obtained via standard calculus.
Now, a second order Taylor expansion of $f$ around $z_0$ with Peano remainder gives
\[
f(z) &= \frac{x(z_0)}{y(z_0)} + \frac{1}{y(z_0)^2}(y(z_0) x'(z_0) - x(z_0) y'(z_0)) (z - z_0) \\
&\phantom{='}+ \frac{1}{2 y(z_0)^3} (y(z_0)^2 x''(z_0) - y(z_0) (2 x'(z_0) y'(z_0) + x(z_0) y''(z_0)) + 2 x(z_0) y'(z_0)^2)(z - z_0)^2 \\
&\phantom{='}+ o((z-z_0)^2) \\
&= 1 + (x'(z_0) - y'(z_0))(z - z_0) \\
&\phantom{='}+\frac{1}{2} (x''(z_0) - 2 x'(z_0) y'(z_0) - y''(z_0) + 2 y'(z_0)^2)(z - z_0)^2 + o((z-z_0)^2) \\
&= 1 + (x'(z_0) - y'(z_0))(z - z_0) - y'(z_0) (x'(z_0)-y'(z_0))(z - z_0)^2 + o((z-z_0)^2),
\]
as required.
\end{proof}

\begin{proof}[\textbf{Proof of \cref{thm:rejection_expansions}}]
The idea is to re-parametrize $(\beta,\beta')\mapsto (\bar\beta, \dbeta)$,
then fix $\bar\beta$ and analyze the asymptotics on $\dbeta$. In turn, we do this
by applying \cref{lemma:ratio_expansion} to the ratio expressions for the jump
probabilities in 
\cref{lemma:acc_prob_as_ratio}. When separating upwards and downwards jump, these
expressions become
\[\label{eq:acc_prob_as_ratio_reparam}
\acc{\beta,\beta'} &= \frac{\EEi{\bar\beta}[\exp(-\dbeta|c'(\check\beta) - V|/2)]}{\EEi{\bar\beta}[\exp(-\dbeta (c'(\check\beta)- V)/2)]} \\
\acc{\beta',\beta} &= \frac{\EEi{\bar\beta}[\exp(-\dbeta|c'(\check\beta) - V|/2)]}{\EEi{\bar\beta}[\exp(\dbeta (c'(\check\beta)- V)/2)]}.
\]
Now, applying \cref{lemma:ratio_expansion} requires obtaining the derivatives of 
numerators and denominators in \cref{eq:acc_prob_as_ratio_reparam} evaluated at
$\dbeta=0$. Under the assumptions considered, \cref{lemma:deriv_expectations}
shows that the common numerator satisfies,
\[
\der{}{\dbeta}\EEi{\bar\beta}[\exp(-\dbeta |c'(\check\beta) - V|/2)]_{\dbeta=0} &= -\frac{1}{2}\EEi{\bar\beta}[|c'(\bar\beta) - V|] \\
\dder{}{\dbeta}\EEi{\bar\beta}[\exp(-\dbeta |c'(\check\beta) - V|/2)]_{\dbeta=0} &= \frac{1}{4}\EEi{\bar\beta}[(c'(\bar\beta) - V)^2],
\]
where we used the fact that the derivative of $\dbeta\mapsto c'(\check\beta(\dbeta))$ is zero at $\dbeta=0$ (\cref{lemma:c_check_beta_C2}).
Similarly, for the denominator in the upward jump probability we get
\[
\der{}{\dbeta}\EEi{\bar\beta}[\exp(-\dbeta (c'(\check\beta) - V)/2)]_{\dbeta=0} &= -\frac{1}{2}\EEi{\bar\beta}[c'(\bar\beta) - V] \\
\dder{}{\dbeta}\EEi{\bar\beta}[\exp(-\dbeta (c'(\check\beta) - V)/2)]_{\dbeta=0} &= \frac{1}{4}\EEi{\bar\beta}[(c'(\bar\beta) - V)^2],
\]
while the denominator for the downward jump yields
\[
\der{}{\dbeta}\EEi{\bar\beta}[\exp(\dbeta (c'(\check\beta) - V)/2)]_{\dbeta=0} &= \frac{1}{2}\EEi{\bar\beta}[c'(\bar\beta) - V] \\
\dder{}{\dbeta}\EEi{\bar\beta}[\exp(\dbeta (c'(\check\beta) - V)/2)]_{\dbeta=0} &= \frac{1}{4}\EEi{\bar\beta}[(c'(\bar\beta) - V)^2].
\]
Since the second derivatives match for numerators and denominators on both 
expressions in \cref{eq:acc_prob_as_ratio_reparam}, we can now apply 
\cref{lemma:ratio_expansion}. The upward jump probability can be expanded as
\[
\acc{\beta,\beta'} &= 1 -\frac{1}{2}\EEi{\bar\beta}[|c'(\bar\beta) - V| - (c'(\bar\beta) - V)]\dbeta \\
&\phantom{='}-\frac{1}{4} \EEi{\bar\beta}[c'(\bar\beta) - V] \EEi{\bar\beta}[|c'(\bar\beta) - V| - (c'(\bar\beta) - V)]\dbeta^2 + O(\dbeta^3) \\
&= 1 - \EEi{\bar\beta}[\max\{0,V-c'(\bar\beta)\}]\dbeta \\
&\phantom{='} + \frac{1}{2} \EEi{\bar\beta}[V-c'(\bar\beta)] \EEi{\bar\beta}[\max\{0,V-c'(\bar\beta)\}]\dbeta^2 + o(\dbeta^2).
\]
The downward jump probability can be similarly expanded as
\[
\acc{\beta',\beta} &= 1 -\frac{1}{2}\EEi{\bar\beta}[|c'(\bar\beta) - V| + (c'(\bar\beta) - V)]\dbeta \\
&\phantom{='}+\frac{1}{4} \EEi{\bar\beta}[c'(\bar\beta) - V]\EEi{\bar\beta}[|c'(\bar\beta) - V| + (c'(\bar\beta) - V)]\dbeta^2 + o(\dbeta^2) \\
&= 1 +\EEi{\bar\beta}[\min\{0,V-c'(\bar\beta)\}]\dbeta \\
&\phantom{='}+\frac{1}{2} \EEi{\bar\beta}[V-c'(\bar\beta)]\EEi{\bar\beta}[\min\{0,V-c'(\bar\beta)\} ] \dbeta^2 + o(\dbeta^2).
\]
The expressions for the rejection probabilities then follow immediately.
\end{proof}

\begin{proof}[\textbf{Proof of \cref{cor:rhoprime_eqs}}]
This is direct from \cref{thm:rejection_expansions}, since
\[
\frac{\rej{\beta,\beta'}}{\dbeta} &= \EEi{\bar\beta}[\max\{0,V-c'(\bar\beta)\}] + O(\dbeta) \\
\frac{\rej{\beta',\beta}}{\dbeta} &= -\EEi{\bar\beta}[\min\{0,V-c'(\bar\beta)\}] + O(\dbeta).
\]
Taking the limit $\dbeta\to0$ gives the expressions for the directional rates,
since $\beta=\bar\beta=\beta'$ in the limit. The expressions in
\cref{eq:rhoprime_diff_abs_vals,eq:rhoprime_sum_abs_vals} can be obtained by 
applying the identities
\[
\max\{0,v\} = \frac{v + |v|}{2}, \qquad \min\{0,v\} = \frac{v - |v|}{2}.
\]
\end{proof}

\begin{lemma}\label{lemma:mad_bounds}
For all $\beta\in[0,1]$,
\[
\EEi{\beta}[|V - \EEi{\beta}[V]|] \leq \EEi{\beta}[|V - V'|] \leq 2\EEi{\beta}[|V - \EEi{\beta}[V]|].
\]
\end{lemma}
\begin{proof}
Since $V'\eqd V$, it holds that
\[
\EEi{\beta}[|V - \EEi{\beta}[V]|] &= \int \pi^{(\beta)}(\dee v) \left| v - \int\pi^{(\beta)}(\dee v') v' \right| \\
&= \int \pi^{(\beta)}(\dee v) \left| \int\pi^{(\beta)}(\dee v')(v - v') \right| \\
&\leq \int \pi^{(\beta)}(\dee v)  \int\pi^{(\beta)}(\dee v')|v - v'| && (\text{Triangle ineq.})\\
&= \int \int\pi^{(\beta)}(\dee v)\pi^{(\beta)}(\dee v')|v - v'| && (\text{Fubini's})\\
&=\EEi{\beta}[|V - V'|].
\]
On the other hand,
\[
\EEi{\beta}[|V - V'|] &= \EEi{\beta}[|(V - \EEi{\beta}[V]) - (V'-\EEi{\beta}[V])|] \\
&\leq 2\EEi{\beta}[|V - \EEi{\beta}[V]|].
\]
\end{proof}

\begin{lemma}\label{lemma:unif_bound_V_moments_with_extremes}
Let $f:\reals\to\reals_+$ be a non-negative function such that $f(V)$ is integrable
under both $\pi_0$ and $\target$. Then, for all $\beta\in[0,1]$,
\[\label{eq:unif_upper_bound_functions}
\mcZ(\beta)\EEi{\beta}[f(V)] \leq \beta\mcZ(1) \EEi{1}[f(V)] + (1-\beta)\EEi{0}[f(V)].
\]
Consequently, under \cref{assu:integrability_at_extremes},
\bitems
\item There exist constants $\mcZ_{\text{min}},\mcZ_{\text{max}}$ such that for all $\beta\in[0,1]$,
\[
0<\mcZ_{\text{min}}\leq\mcZ(\beta)\leq\mcZ_{\text{max}}<\infty.
\]
\item For $p\in[0,2]$ and all $\beta\in[0,1]$,
\[
\EEi{\beta}[|V|^p] \leq M_p := \frac{1}{\mcZ_\text{min}}\max\left\{\mcZ(1)\EEi{1}[|V|^p], \EEi{0}[|V|^p] \right\}.
\]
\eitems
\end{lemma}
\begin{proof}
Note that we can trivially write $\beta V$ as a convex combination
\[
\beta V = \beta V + (1-\beta)\cdot 0.
\]
And because $x\mapsto e^{-x}$ is convex,
\[
e^{-\beta V} \leq \beta e^{-V} + (1-\beta)
\]
Hence,
\[
\mcZ(\beta)\EEi{\beta}[f(V)] &= \EEi{0}[f(V) e^{-\beta V}] \\
&\leq \beta \EEi{0}[f(V)e^{-V}] + (1-\beta)\EEi{0}[f(V)] \\
&= \beta\mcZ(1) \EEi{1}[f(V)] + (1-\beta)\EEi{0}[f(V)].
\]

We can now tackle the second set of statements. 
The existence of $\mcZ_{\text{max}}$ follows from \cref{lemma:finite_norm_const}
and \cref{eq:unif_upper_bound_functions} with $f\equiv1$.
Let us now suppose \cref{assu:integrability_at_extremes} is in place. By the
convexity of $x\mapsto e^{-\beta x}$,
\[
\mcZ(\beta)=\EEi{0}[e^{-\beta V}] \geq e^{-\beta\EEi{0}[V]} \geq e^{-\EEi{0}[V]\ind\{\EEi{0}[V]>0\}} =: \mcZ_{\text{min}} > 0.
\]
Now, for $f_p(v)=|v|^p$ with $p\in[0,2]$, \cref{assu:integrability_at_extremes}
means $f_p$ is integrable under both $\pi_0$ and $\target$. Hence,
\[
\EEi{\beta}[|V|^p] &\leq \frac{1}{\mcZ(\beta)}\left(\beta\mcZ(1) \EEi{1}[|V|^p] + (1-\beta)\EEi{0}[|V|^p]\right) \\
&\leq \frac{1}{\mcZ_\text{min}}\left(\beta\mcZ(1) \EEi{1}[|V|^p] + (1-\beta)\EEi{0}[|V|^p]\right) \\
&\leq \frac{1}{\mcZ_\text{min}}\max\left\{\mcZ(1)\EEi{1}[|V|^p], \EEi{0}[|V|^p] \right\},
\]
as required.
\end{proof}

\begin{lemma}\label{lemma:unif_bound_dder_rejection}
Under
\cref{assu:c_smooth,assu:nullset,assu:integrability_at_extremes,assu:bounded_c},
there exist functions $b_1$ and $b_2$ such that for any partition $\mcP$, 
\[
\left|\D{\rho_i}{\dbeta}(\bar\beta_i, \dbeta_i)\right| \leq b_1(\|\mcP\|) && 
\left|\DD{\rho_i}{\dbeta}(\bar\beta_i, \dbeta_i)\right| \leq b_2(\|\mcP\|)
\]
uniformly on $i\in\indsetone$, and $b_1,b_2$ have finite limits for $\|\mcP\|\to0$.
\end{lemma}
\begin{proof}
From the definition of $\bar\rho$ (\cref{eq:def_sym_rej_arbitrary_pair}) 
follows that for $i\in\{1,2\}$,
\[
\left|\Di{}{\dbeta}\bar\rho_{\beta,\beta'}\right| \leq \frac{1}{2}\left(\left|\Di{}{\dbeta}\acc{\beta,\beta'}\right| + \left|\Di{}{\dbeta}\acc{\beta',\beta} \right|\right).
\]
Thus, it suffices to bound both terms on the right. By symmetry, we can focus only 
on the ones for the forward jump. Now, recall from \cref{lemma:acc_prob_as_ratio} 
that the acceptance probabilities can be expressed as ratios of functions,
\[
\acc{\beta,\beta'} = \frac{\EEi{\bar\beta}[\exp(-\dbeta|c'(\check\beta) - V|/2)]}{\EEi{\bar\beta}[\exp(-\dbeta (c'(\check\beta)- V)/2)]} = \frac{x(z)}{y(z)},
\]
where, for clarity of exposition, we have written $z=\dbeta$ and made the 
dependence on $\bar\beta$ implicit. By \cref{lemma:ratio_expansion},
\[\label{eq:derivatives_acc_prob}
\D{\acc{\beta,\beta'}}{\dbeta} &= \frac{1}{y(z)^2}(y(z) x'(z) - x(z) y'(z)) \\
\DD{\acc{\beta,\beta'}}{\dbeta} &= \frac{1}{y(z)^3} (y(z)^2 x''(z) - y(z) (2 x'(z) y'(z) + x(z) y''(z)) + 2 x(z) y'(z)^2).
\]
To prove that these derivatives are uniformly bounded, we will first show that the
denominator $y(z)$ is uniformly bounded away from zero. Indeed, from
\cref{eq:acc_prob_as_ratio_denom} in \cref{lemma:acc_prob_as_ratio}, we know that
\[\label{eq:lower_bound_denominator}
y(z) = \EEi{\bar\beta}[\exp(-\dbeta (c'(\check\beta)- V)/2)] &=  \frac{\mcZ(\bar\beta-\dbeta/2)}{\mcZ(\bar\beta)}e^{-c'(\check\beta)\dbeta /2} \\
&\geq \frac{\mcZ_{\text{min}}}{\mcZ_{\text{max}}}e^{-|c'(\check\beta)|\dbeta /2} \\
&\geq \frac{\mcZ_{\text{min}}}{\mcZ_{\text{max}}}e^{-c_1\|\mcP\| /2},
\]
where the first bound follows from \cref{lemma:unif_bound_V_moments_with_extremes},
while the second uses \cref{assu:bounded_c} and the definition of $\|\mcP\|$.

To continue, note that the numerators on the right-hand side of 
\cref{eq:derivatives_acc_prob} are, respectively, first and second order 
polynomials on the functions $(x,y)$ and their derivatives $(x',x'',y',y'')$.
Hence, we proceed by showing that the absolute value of each of these functions
is uniformly bounded under the assumptions given. Indeed, note that
\[
0 < x(z) &= \EEi{\bar\beta}[\exp(-\dbeta |c'(\check\beta)- V|/2)] \\
&\leq \EEi{\bar\beta}[\exp(-\dbeta (c'(\check\beta)- V)/2)] \\
&= y(z) \\
&\leq \frac{\mcZ_\text{max}}{\mcZ_\text{min}}e^{c_1\|\mcP\| /2},
\]
where we used an argument similar to the one for \cref{eq:lower_bound_denominator}.
To continue with the derivatives, recall the expressions given by 
\cref{lemma:deriv_expectations}
\[
x'(z) &=-\frac{1}{2}\EEi{\bar\beta}[\exp(- \dbeta|\gamma(\dbeta)-V|/2)(|\gamma(\dbeta)-V| + \dbeta\gamma'(\dbeta)\sgn(\gamma(\dbeta)-V))] \\
y'(z) &= -\frac{1}{2}\EEi{\bar\beta}[\exp(- \dbeta(\gamma(\dbeta)-V)/2)(\gamma(\dbeta)-V + \dbeta\gamma'(\dbeta))],
\]
where $\gamma(\dbeta):=c'(\check{\beta}(\dbeta))$. Now,
\[
|x'(z)| &\leq \frac{1}{2}\EEi{\bar\beta}[\exp(- \dbeta|\gamma(\dbeta)-V|/2)||\gamma(\dbeta)-V| + \dbeta\gamma'(\dbeta)\sgn(\gamma(\dbeta)-V) |] \\
&\leq  \frac{1}{2}\EEi{\bar\beta}[\exp(- \dbeta(\gamma(\dbeta)-V)/2)||\gamma(\dbeta)-V| + \dbeta\gamma'(\dbeta)\sgn(\gamma(\dbeta)-V) |] \\
&\leq  \frac{1}{2}\EEi{\bar\beta}[\exp(- \dbeta(\gamma(\dbeta)-V)/2)(|\gamma(\dbeta)-V| + \dbeta|\gamma'(\dbeta)|)].
\]
Note that $|y'|$ admits the same bound,
\[
|y'(z)| &\leq \frac{1}{2}\EEi{\bar\beta}[\exp(- \dbeta(\gamma(\dbeta)-V)/2)|\gamma(\dbeta)-V + \dbeta\gamma'(\dbeta)|] \\
&\leq \frac{1}{2}\EEi{\bar\beta}[\exp(- \dbeta(\gamma(\dbeta)-V)/2)(|\gamma(\dbeta)-V| + \dbeta|\gamma'(\dbeta)|)].
\]
Moreover,
\[\label{eq:first_deriv_bound}
&\EEi{\bar\beta}[\exp(- \dbeta(\gamma(\dbeta)-V)/2)(|\gamma(\dbeta)-V| + \dbeta|\gamma'(\dbeta)|)] \\
&\leq (|\gamma(\dbeta)| + \dbeta |\gamma'(\dbeta)|)y(\dbeta) + \EEi{\bar\beta}[|V|\exp(- \dbeta(\gamma(\dbeta)-V)/2)]\\
&\leq (c_1 + c_2\|\mcP\|)\frac{\mcZ_\text{max}}{\mcZ_\text{min}}e^{c_1\|\mcP\| /2} + \EEi{\bar\beta}[|V|\exp(- \dbeta(\gamma(\dbeta)-V)/2)].
\]
Before searching for a uniform bound for the second term on the right-hand side,
let us show that $|x''|$ and $|y''|$ allow for a bound of the same form. Indeed,
from the expressions in \cref{lemma:deriv_expectations}
\[
|x''(z)| &= \frac{1}{2}\left|\EEi{\bar\beta}[\exp(- \dbeta|\gamma(\dbeta)-V|/2)\sgn(\gamma(\dbeta)-V)(2\gamma'(\dbeta) + \dbeta\gamma''(\dbeta))] \right. \\
&\phantom{='}+\left.\frac{1}{4}\EEi{\bar\beta}[\exp(- \dbeta|\gamma(\dbeta)-V|/2)(|\gamma(\dbeta)-V| + \dbeta\gamma'(\dbeta)\sgn(\gamma(\dbeta)-V))^2]\right|\\
&\leq \frac{1}{2}(2|\gamma'(\dbeta)| + \dbeta|\gamma''(\dbeta)|)\EEi{\bar\beta}[\exp(- \dbeta|\gamma(\dbeta)-V|/2)]  \\
&\phantom{='}+\frac{1}{2}\EEi{\bar\beta}[\exp(- \dbeta|\gamma(\dbeta)-V|/2)((\gamma(\dbeta)-V)^2 + \dbeta^2\gamma'(\dbeta)^2)]\\
&= \frac{1}{2}(2|\gamma'(\dbeta)| + \dbeta|\gamma''(\dbeta)|+\dbeta^2\gamma'(\dbeta)^2)\EEi{\bar\beta}[\exp(- \dbeta|\gamma(\dbeta)-V|/2)]  \\
&\phantom{='}+\frac{1}{2}\EEi{\bar\beta}[\exp(- \dbeta|\gamma(\dbeta)-V|/2)(\gamma(\dbeta)-V)^2] \\
&\leq \frac{1}{2}(2|\gamma'(\dbeta)| + \dbeta|\gamma''(\dbeta)|+\dbeta^2\gamma'(\dbeta)^2)\EEi{\bar\beta}[\exp(- \dbeta(\gamma(\dbeta)-V)/2)]  \\
&\phantom{='}+\frac{1}{2}\EEi{\bar\beta}[\exp(- \dbeta(\gamma(\dbeta)-V)/2)(\gamma(\dbeta)-V)^2].
\]
Similarly,
\[\label{eq:second_deriv_bound}
|y''(z)| &\leq \frac{1}{2}(2|\gamma'(\dbeta)| + \dbeta|\gamma''(\dbeta)|)\EEi{\bar\beta}[\exp(- \dbeta(\gamma(\dbeta)-V)/2)] \\
&\phantom{M='}+\frac{1}{4}\EEi{\bar\beta}[\exp(- \dbeta(\gamma(\dbeta)-V)/2)(\gamma(\dbeta)-V + \dbeta\gamma'(\dbeta))^2] \\
&\leq \frac{1}{2}(2|\gamma'(\dbeta)| + \dbeta|\gamma''(\dbeta)| + \dbeta^2\gamma'(\dbeta)^2)\EEi{\bar\beta}[\exp(- \dbeta(\gamma(\dbeta)-V)/2)] \\
&\phantom{M='}+\frac{1}{2}\EEi{\bar\beta}[\exp(- \dbeta(\gamma(\dbeta)-V)/2)(\gamma(\dbeta)-V)^2] \\
&\leq \frac{1}{2}(2|\gamma'(\dbeta)| + \dbeta|\gamma''(\dbeta)| + \dbeta^2\gamma'(\dbeta)^2 + 2\gamma(\dbeta)^2)\EEi{\bar\beta}[\exp(- \dbeta(\gamma(\dbeta)-V)/2)] \\
&\phantom{M='}+\EEi{\bar\beta}[V^2\exp(- \dbeta(\gamma(\dbeta)-V)/2)].
\]
Thus, both functions admit the same bound. Furthermore,
\[
&(2|\gamma'(\dbeta)| + \dbeta|\gamma''(\dbeta)| + \dbeta^2\gamma'(\dbeta)^2 + 2\gamma(\dbeta)^2)\EEi{\bar\beta}[\exp(- \dbeta(\gamma(\dbeta)-V)/2)] \\
&\leq (2c_2 + \|\mcP\|c_3 + \|\mcP\|^2[c_2^2+2c_1^2])\frac{\mcZ_\text{max}}{\mcZ_\text{min}}e^{c_1\|\mcP\| /2}.
\]
It remains to bound the second term on the right-hand sides of \cref{eq:first_deriv_bound,eq:second_deriv_bound}. Note that for $p\geq 0$,
\[
&\EEi{\bar\beta}[|V|^p\exp(- \dbeta(\gamma(\dbeta)-V)/2)]\\
&= \frac{e^{\dbeta\gamma(\dbeta)/2}\mcZ(\bar\beta-\dbeta/2)}{\mcZ(\bar\beta)}\EEi{\bar\beta-\dbeta/2}[|V|^p] \\
&\leq \frac{\mcZ_\text{max}}{\mcZ_\text{min}}e^{c_1\|\mcP\| /2} M_p.
\]
where $M_p$ is the uniform $p$-moment bound from
\cref{lemma:unif_bound_V_moments_with_extremes}. We conclude by noting that
all the bounds have finite limits for $\|\mcP\|\to0$, as required.
\end{proof}

\begin{proof}[\textbf{Proof of \cref{prop:rejs_Riemann_sums}}]
Let us begin by showing that
\[
\lim_{\|\mcP\|\to 0}\sum_{i=1}^{\ell(\beta)} \rho_i = \Lambda(\beta).
\]
To this end, let us extend the definition of symmetrized rejection probability 
$\rho_i$ to arbitrary pairs $\beta<\beta'$ via
\[\label{eq:def_sym_rej_arbitrary_pair}
\bar\rho_{\beta,\beta'} := 1-\frac{\acc{\beta,\beta'}+\acc{\beta',\beta}}{2},
\]
so that $\rho_i=\bar\rho_{\beta_{i-1},\beta_i}$ for all $i\in\indsetone$. 
Now, consider a first-order Taylor expansion of $\bar\rho_{\beta,\beta'}$ around 
$\dbeta=0$ and fixed $\bar\beta=(\beta+\beta')/2$, with Lagrange remainder
\[
\bar\rho_{\beta,\beta'} &= \rho'(\bar\beta)\dbeta + \frac{1}{2}\DD{\bar\rho_{\beta,\beta'}}{\dbeta}(z)\dbeta^2,
\]
for some $z\in[0,\dbeta]$ and all $\dbeta\leq r(\bar\beta)$, with 
$r(\bar\beta)$ as in \cref{eq:def_rbarbeta}. Let us write 
$\bar\beta_i := (\beta_{i-1}+\beta_i)/2$ and $\dbeta_i:=\beta_i-\beta_{i-1}$. Since
$\beta_i\in[0,1]$, we see that
\[
\dbeta_i = \beta_i - \beta_{i-1} = 2\bar\beta_i - 2\beta_{i-1} \leq 2\bar\beta_i,
\]
and
\[
\dbeta_i = 2(\beta_i - \bar\beta_i) \leq 2(1 - \bar\beta_i).
\]
So the restriction $\dbeta_i\leq r(\bar\beta_i)$ is satisfied for all 
$i\in\indsetone$ and for any partition $\mcP\subset[0,1]$. Hence,
\[
\rho_i &= \rho'(\bar\beta_i)\dbeta_i  + \frac{1}{2}\DD{\bar\rho_{\beta,\beta'}}{\dbeta}(z_i)\dbeta_i^2.
\]
for some $z_i\in[0,\dbeta_i]$. Summing these expansions over
$i\in\{1,\dots,\ell(\beta)\}$ yields
\[
\sum_{i=1}^{\ell(\beta)} \rho_i &= \sum_{i=1}^{\ell(\beta)} \rho'(\bar\beta_i)\dbeta_i  + \frac{1}{2}\sum_{i=1}^{\ell(\beta)}\DD{\bar\rho_{\beta,\beta'}}{\dbeta}(z_i)\dbeta_i^2
\]
The first term on the right-hand side is a Riemann sum that converges to 
$\Lambda(\beta)$. We must show that the second sum vanishes in the limit. Indeed,
by \cref{lemma:unif_bound_dder_rejection},
\[
\left| \sum_{i=1}^{\ell(\beta)}\DD{\bar\rho_{\beta,\beta'}}{\dbeta}(z_i)\dbeta_i^2 \right| &\leq b_2(\|\mcP\|)\sum_{i=1}^{\ell(\beta)}\dbeta_i^2 \leq b_2(\|\mcP\|) \|\mcP\| \sum_{i=1}^{\ell(\beta)}\dbeta_i \leq b_2(\|\mcP\|) \|\mcP\| \beta,
\]
and the right-hand side goes to $0$ as $\|\mcP\|\to0$.

We now turn to proving that
\[
\lim_{\|\mcP\|\to 0}\sum_{i=1}^{\ell(\beta)} \frac{\rho_i}{\alpha_i}=\lim_{\|\mcP\|\to 0}\sum_{i=1}^{\ell(\beta)} \frac{\rho_i}{1-\rho_i} = \Lambda(\beta).
\]
To this end, consider the function
\[
f(z) &= \frac{x(z)}{1-x(z)}.
\]
Its first and second derivatives are
\[
f'(z) = \frac{x'(z)}{(1-x(z))^2}
\]
and
\[
f''(z) = \frac{x''(z)}{(1-x(z))^2} + 2\frac{(x'(z))^2}{(1-x(z))^3}.
\]
We want to use these expressions with $z\gets \dbeta$ and 
$x(z)\gets \bar\rho_{\beta,\beta'}$. Note that if we can bound uniformly (on $i$) 
the second derivative $f''(\dbeta)$, then the same proof technique used for the
first series yields the result. Indeed, $\rho_i$ and $\rho_i/(1-\rho_i)$ are 
first-order identical, since $f(0)=0$ and $f'(0)=\rho'(\bar\beta)$. 
Now, recall that
\[
\bar\rho_{\beta,\beta'} &= \rho'(\bar\beta)\dbeta + \frac{1}{2}\DD{\bar\rho_{\beta,\beta'}}{\dbeta}(z)\dbeta^2.
\]
Since $\rho'(\bar\beta)=\D{\bar\rho_{\beta,\beta'}}{\dbeta}(0)$, 
\cref{lemma:unif_bound_dder_rejection} shows that both derivatives in the above
display are uniformly bounded. Hence, $\bar\rho_{\beta,\beta'}=O(\|\mcP\|)$. And
since we are concerned with the limit $\|\mcP\|\to0$, for any
$\varepsilon\in(0,1)$, we may only consider fine enough partitions $\mcP$ such 
that $\rho_i\leq \varepsilon$ for all $i\in\indsetone$. Under this scenario, we 
have
\[
|f''(\dbeta)| \leq \frac{\left|\DD{\bar\rho_{\beta,\beta'}}{\dbeta}\right|}{(1-\varepsilon)^2} + 2\frac{\left(\D{\bar\rho_{\beta,\beta'}}{\dbeta}\right)^2}{(1-\varepsilon)^3}.
\]
Again, the two derivatives appearing in the right-hand side are uniformly bounded
by \cref{lemma:unif_bound_dder_rejection}, which concludes the proof.
\end{proof}

\subsection{Proofs for \cref{sec:adaptation}}

\begin{lemma}\label{lemma:opt_constrained_sum_convex}
For convex $f:\reals\to\reals$ and $c\in\reals$, the optimization problem
\[
\min_{x\in\reals^N}\sum_{i=1}^N f(x_i) \qquad \mathrm{s.t. }\; \sum_{i=1}^N x_i = c
\]
has a unique solution given by $x_i^\star = \frac{c}{N}$ for all
$i\in\{1,\dots,N\}$.
\end{lemma}
\begin{proof}
By Jensen's inequality,
\[
\sum_{i=1}^N f(x_i) = N \frac{1}{N}\sum_{i=1}^N f(x_i) \geq Nf\left(\frac{1}{N}\sum_{i=1}^N x_i\right) = Nf\left(\frac{c}{N}\right),
\]
and the bound on the right-hand side is achieved whenever 
$x_i^\star = \frac{c}{N}$ for all $i\in\{1,\dots,N\}$.
\end{proof}

\begin{lemma}\label{lemma:maximize_rational_fun}
Let $\alpha,\beta,\gamma\in\reals$ such that $\alpha\leq\gamma$ and 
$\beta\leq \gamma$. Consider the function $\psi:[\gamma,\infty)\to(0,\infty)$ 
given by
\[
\psi(t) := \frac{t-\gamma}{(t-\alpha)(t-\beta)}.
\]
Then $\psi$ has a unique maximizer at
\[
t^\star := \gamma + \sqrt{(\gamma-\alpha)(\gamma-\beta)}.
\]
\end{lemma}
\begin{proof}
Consider the change of variable
\[
t(s) := \gamma + e^{s}\sqrt{(\gamma-\alpha)(\gamma-\beta)}.
\]
Then,
\[
\psi(t(s)) &= \frac{e^{s}\sqrt{(\gamma-\alpha)(\gamma-\beta)}}{(\gamma-\alpha)(\gamma-\beta) + e^{2s}(\gamma-\alpha)(\gamma-\beta) + e^{s}(\gamma-\alpha+\gamma-\beta)\sqrt{(\gamma-\alpha)(\gamma-\beta)}} \\
&= \frac{1}{e^{-s}\sqrt{(\gamma-\alpha)(\gamma-\beta)} + e^{s}\sqrt{(\gamma-\alpha)(\gamma-\beta)} + (2\gamma-\alpha-\beta)} \\
&= \frac{1}{2\cosh(s)\sqrt{(\gamma-\alpha)(\gamma-\beta)} + (2\gamma-\alpha-\beta)} \\
&= \phi(\cosh(s)),
\]
where 
\[
\phi(u) := \frac{1}{2u\sqrt{(\gamma-\alpha)(\gamma-\beta)} + (2\gamma-\alpha-\beta)}
\]
is strictly decreasing. It follows that 
$\psi(t(s))$ is maximized when $\cosh(s)$ is minimized; i.e., for $s^\star=0$.
Translating this back to the original scale gives
\[
t^\star = \gamma + \sqrt{(\gamma-\alpha)(\gamma-\beta)},
\]
as required.
\end{proof}

\begin{lemma}\label{lemma:TE_over_tourlength_maximizer}
Let $\Lambda>0$. The function $F:(\Lambda,\infty)\to\reals_+$ given by
\[
F(t) = \frac{t-\Lambda}{(t+1)(t(1+2\Lambda)-\Lambda)},
\]
has a unique maximizer
\[
t^\star = \Lambda\left(1+\sqrt{1+\frac{1}{1+2\Lambda}}\right).
\]
\end{lemma}
\begin{proof}
Note that
\[
F(t) = \frac{1}{1+2\Lambda}\left(\frac{t-\Lambda}{(t+1)(t-\frac{\Lambda}{1+2\Lambda})}\right).
\]
Therefore, we can apply \cref{lemma:maximize_rational_fun} with $\gamma:=\Lambda$, $\alpha:=-1<\gamma$, and $\beta:=\frac{\Lambda}{1+2\Lambda}<
\gamma$ to obtain
\[
t^\star &= \Lambda + \sqrt{(\Lambda+1)\left(\Lambda-\frac{\Lambda}{1+2\Lambda}\right)} 
= \Lambda\left(1 + \sqrt{1+\frac{1}{1+2\Lambda}} \right),
\]
as required.
\end{proof}

\end{document}